\documentclass[10pt,a4paper]{article}

\usepackage{amsmath,amsgen,latexsym}
\usepackage{amstext,amssymb,amsfonts,latexsym}
\usepackage{theorem}
\usepackage{pifont}
\usepackage[pdftex]{graphics}

 \newcommand{\bs}{\bigskip}
 \newcommand{\ms}{\medskip}
 \newcommand{\n}{\noindent}
 \newcommand{\s}{\smallskip}
 \newcommand{\hs}[1]{\hspace*{ #1 mm}}
 \newcommand{\vs}[1]{\vspace*{ #1 mm}}

 \newcommand{\setempty}{\varnothing}
 \newcommand{\real}{\mathbb{R}}
 \newcommand{\nat}{\mathbb{N}}
 
 \newcommand{\integer}{\mathbb{Z}}
 \newcommand{\rational}{\mathbb{Q}}
 \newcommand{\complex}{\mathbb{C}}

 \newcommand{\co}{\mathrm{co}\mbox{-}}

 \newcommand{\FF}{{\cal F}}

 \newcommand{\HH}{{\cal H}}

 \newcommand{\GG}{{\cal G}}

 \newcommand{\SSS}{{\cal S}}
 \newcommand{\PP}{{\cal P}}
 
 \newcommand{\RR}{{\cal R}}

\theoremstyle{plain}
\theoremheaderfont{\bfseries}
 \newtheorem{theorem}{Theorem}[section]
 \newtheorem{lemma}[theorem]{Lemma}
 \newtheorem{proposition}[theorem]{{\bf Proposition}}
 \newtheorem{corollary}[theorem]{Corollary}

 {\theorembodyfont{\rmfamily}
  \newtheorem{definition}[theorem]{Definition}}
 {\theorembodyfont{\rmfamily} \newtheorem{example}[theorem]{Example}}
 {\theorembodyfont{\rmfamily} }

 \newtheorem{claim}[theorem]{Claim}

 \newenvironment{proofof}[1]{\vspace*{5mm} \par \noindent
         {\bf Proof of #1.\hs{2}}}{\hfill$\Box$ \vspace*{3mm}}

 \newenvironment{proof}{\par \noindent
            {\bf Proof. \hs{2}}}{\hfill$\Box$ \vspace*{3mm}}

 \newcommand{\ceilings}[1]{\lceil #1 \rceil}
 \newcommand{\floors}[1]{\lfloor #1 \rfloor}

 \newcommand{\qubit}[1]{| #1 \rangle}
 \newcommand{\bra}[1]{\langle #1 |}
 \newcommand{\ket}[1]{| #1 \rangle}
 \newcommand{\measure}[2]{\langle #1 | #2 \rangle}

\newcommand{\ignore}[1]{}

 \newcommand{\bfzero}{{\bf 0}}

 \newcommand{\squareqp}{\Box^{\mathrm{QP}}_{1}}
 \newcommand{\hatsquareqp}{\widehat{\Box^{\mathrm{QP}}_{1}}}

 \newcommand{\ilog}{\mathrm{ilog}}
 \newcommand{\iloglog}{\mathrm{iloglog}}

 \newcommand{\bqpolylogtime}{\mathrm{BQPOLYLOGTIME}}
 \newcommand{\dlogtime}{\mathrm{DLOGTIME}}
 \newcommand{\nlogtime}{\mathrm{NLOGTIME}}
 \newcommand{\ppolylogtime}{\mathrm{PPOLYLOGTIME}}
 \newcommand{\bqp}{\mathrm{BQP}}
 \newcommand{\nqp}{\mathrm{NQP}}
 \newcommand{\bqlogtime}{\mathrm{BQLOGTIME}}

\begin{document}

\pagestyle{plain}
\pagenumbering{arabic}
\setcounter{page}{1}
\setcounter{footnote}{0}


\begin{center}
{\Large {\bf Elementary Quantum Recursion Schemes That Capture Quantum Polylogarithmic Time Computability of Quantum Functions}}\footnote{A preliminary report  \cite{Yam22b} has appeared under a slightly different title in the Proceedings of the 28th International Conference on Logic, Language, Information, and Computation (WoLLIC 2022), Ia\c{s}i, Romania, September 20--23, 2022, Lecture Notes in Computer Science, vol. 13468, pp. 88--104, Springer, 2022.}
\bs\s\\

{\sc Tomoyuki Yamakami}\footnote{Present Affiliation: Faculty of Engineering, University of Fukui, 3-9-1 Bunkyo, Fukui 910-8507, Japan} \bs\\
\end{center}
\ms


\begin{abstract}
Quantum computing has been studied over the past four decades based on two computational models of quantum circuits and quantum Turing machines. To capture quantum polynomial-time computability, a new recursion-theoretic approach  was taken lately by Yamakami [J. Symb. Logic 80, pp.~1546--1587, 2020] by way of recursion schematic definition, which constitutes six  initial quantum functions and three construction schemes of composition, branching, and multi-qubit quantum recursion. By taking a similar approach, we look into quantum polylogarithmic-time computability and further explore the expressing power of elementary schemes designed for such quantum computation.
In particular, we introduce an elementary form of the quantum recursion, called the fast quantum recursion, and formulate $EQS$ (elementary quantum schemes) of ``elementary'' quantum functions. This class $EQS$ captures  exactly quantum polylogarithmic-time computability, which forms the complexity class BQPOLYLOGTIME.
We also demonstrate the separation of BQPOLYLOGTIME from NLOGTIME and PPOLYLOGTIME.
As a natural extension of $EQS$, we further consider an algorithmic procedural scheme that implements the well-known divide-and-conquer strategy. This divide-and-conquer scheme helps compute the parity function but the scheme cannot be realized within our system $EQS$.

\ms
\n{\bf Keywords:}  {recursion schematic definition, quantum Turing machine, fast quantum recursion, quantum polylogarithmic-time computability, divide-and-conquer strategy}
\end{abstract}

\sloppy
\section{Background, Motivations, and Challenges}\label{sec:introduction}

We provide a quick overview of this work from its background to challenging open questions to tackle.

\subsection{Recursion Schematic Definitions That Capture Quantum Computability}\label{sec:schematic-def}

Over the four decades, the study of quantum computing has made significant progress in both theory and practice. Quantum computing is one of the most anticipated nature-inspired computing paradigms today because it relies on the principle of quantum mechanics, which is assumed to govern nature. The core of quantum computing is an exquisite  handling of \emph{superpositions}, which are linear combinations of basic quantum states expressing classical strings,  and \emph{entanglement}, which binds multiple quantum bits (or qubits, for short) with no direct contact, of quantum states.
Early models of quantum computation were proposed by  Benioff \cite{Ben80}, Deutsch \cite{Deu85,Deu89}, and Yao \cite{Yao93} as quantum analogues of Turing machines and Boolean circuits. Those models have been used in a theoretical study of quantum computing within the area of computational complexity theory since their introduction. The quantum Turing machine (QTM) model is a natural quantum extension of the classical Turing machine (TM) model, which has successfully served for decades as a basis to theoretical aspects of computer science.
The basic formulation of the QTM model\footnote{Bernstein and Vazirani discussed only single-tape QTMs. The multiple-tape model of QTMs was distinctly discussed in \cite{Yam99,Yam03}. The foundation of QTMs was also studied in \cite{NO02,ON00}.} used today attributes to Bernstein and Vazirani \cite{BV97}.
QTMs provide a blueprint of algorithmic procedures by describing how to transform each superposition of configurations of the machines.
There have been other computational models proposed to capture different aspects of quantum computing, including a (black-box) query model.

A \emph{recursion schematic approach} has been a recent incentive for the full characterization of the notion of quantum polynomial-time computability \cite{Yam20}, which can be seen as a quantum extension of deterministic polynomial-time computability.
The primary purpose of taking such an exceptional approach toward quantum computability stems from a successful development of recursion theory (or recursive function theory) for classical computability. Earlier, Peano, Herbrand, G\"{o}del, and Kleene \cite{Kle36,Kle43} all made significant contributions to paving a straight road to a coherent study of computability and decidability from a purely logical aspect (see, e.g., \cite{Soa96} for further references therein). In this theory, recursive functions are formulated in the following simple ``schematic'' way. (i) We begin with a few initial functions. (ii) We sequentially build more complicated functions by applying a small set of construction schemes to the already-built functions. The description of how to construct a recursive function  actually serves as a blueprint of the function, which resembles like a ``computer program'' of the modern times in such a way that each construction scheme is a command line of a computer program.
Hence, a schematic definition itself may be viewed as a high-level programming language by which we can render a computer program that describes the movement of any recursive function. As the benefit of the use of such a schematic definition, when we express a target function as a series of schemes, the ``size'' of this series can serve as a ``descriptional complexity measure'' of the function. This has lead to a fully enriched study of the descriptive complexity of functions. This schematic approach sharply contrasts the ones based on Turing machines as well as Boolean circuit families to formulate the recursive functions.
A similar approach taken in \cite{Yam20} aims at capturing the polynomial-time computability of \emph{quantum functions},\footnote{This notion should be distinguished from the same terminology used in \cite{Yam03}, where ``quantum functions'' mean mappings from binary strings to acceptance probabilities of QTMs.} each of which maps a finite-dimensional Hilbert space to itself, in place of the aforementioned recursive functions.
Unlike quantum transitions of QTMs, the recursion schemes can provide a  clear view of how quantum transforms act on target quantum states in the Hilbert space.

Two important classes of quantum functions, denoted $\squareqp$ and $\hatsquareqp$, were formulated in \cite{Yam20} from six initial quantum functions and three basic construction schemes. The major difference between $\squareqp$ and $\hatsquareqp$ is the permitted use of \emph{quantum measurement operations}, which make quantum states collapse to classical states with incurred probabilities.
Unfortunately, this is a non-reversible operation.
Shown in \cite{Yam20} are a precise characterization of quantum polynomial-time computability and a quantum analogue of the normal form theorem. A key in his recursion schematic definition is an introduction of the \emph{multi-qubit quantum recursion scheme}, which looks quite different  in its formulation from the corresponding classical recursion scheme. We will give its formal definition in Section \ref{sec:elementary-scheme} as \emph{Scheme T}. This quantum recursion scheme turns out to be so powerful that it has helped us capture the notion of quantum polynomial-time computability.

There are a few but important advantages of using recursion schematic definitions to introduce quantum computing over other quantum computational models, such as QTMs and quantum circuits. The most significant advantage is in fact that there is no use of the \emph{well-formedness requirement} of QTMs and the \emph{uniformity requirement} of quantum circuit families. These natural but cumbersome requirements are necessary to guarantee the unitary nature of QTMs and the algorithmic construction of quantum circuit families. Recursion schematic definitions, on the contrary, avoid such extra  requirements and make it much simpler to design
a quantum algorithm in the form of a quantum function for solving a target combinatorial problem. This is because each basic scheme naturally embodies ``well-formedness''  and ``uniformity'' in its formulation.
More accurately, the scheme produces only  ``well-formed''  and ``uniform''  quantum functions by way of applying the scheme directly to the existing quantum functions.
In fact, the schematic definition of \cite{Yam20} constitutes only six  initial quantum functions together with three basic construction schemes and this fact helps us render a short procedural sequence of how to construct each target quantum function.
A descriptional aspect of the recursion schematic definition of \cite{Yam20} has become a helpful guide to develop even a suitable form of quantum programming language \cite{HPS23}.
It is worth mentioning that such a recursion schematic definition has been  made possible because of an extensional use of the bra and ket notations. See Section \ref{sec:bra-ket-notation} for detailed explanations.

A schematic approach of \cite{Yam20} to quantum computing has exhibited  a great possibility of treating quantum computable functions in such a way that is quite different from the traditional ways with QTMs and quantum circuit families. Recursion schemes further lead us to the descriptional complexity of quantum functions.
Concerning the recursion schematic definitions, numerous questions still remain unanswered. It is important and also imperative to expand the scope of our research to various resource-bounded quantum computing.
In this work, we particularly wish to turn our attention to a ``limited'' form of quantum computing, which can be naturally implemented by small runtime-restricted QTMs as well as families of small depth-bounded quantum circuits, and we further wish to examine how the schematic approach of \cite{Yam20} copes with such restricted quantum computing.

Of practical importance, resource-bounded computability has been studied in the classical setting based on various computational models, including families of Boolean circuits and time-bounded (classical) Turing machines (or TMs, for short).
Among all reasonable resource bounds, we are particularly keen to \emph{(poly)logarithmic time}. The logarithmic-time (abbreviated as logtime) computability may be one of the most restricted but meaningful, practical resource-bounded computabilities and it
was discussed earlier by Buss \cite{Bus87} and Barrington, Immerman, and Straubing \cite{BIS90} in terms of TMs equipped with index tapes to access particular locations of input symbols.
Hereafter, we denote by $\dlogtime$ (resp., $\nlogtime$) the class of decision problems solvable by logtime deterministic (resp., nondeterministic) TMs.
The logtime computability has played a central role in characterizing a uniform notion of constant-depth Boolean circuit families. As shown in \cite{BIS90}, logtime computability is also closely related to the first-order logic with the special predicate $BIT$.
As slightly more general resource-bounded computability,  \emph{polylogarithmic-time} (abbreviated as \emph{polylogtime}) computability has been widely studied in the literature. Numerous polylogtime  algorithms were developed to solve, for example, the matrix chain ordering problem \cite{BRS94} and deterministic graph construction \cite{HLT01}.
There were also studies on data structures that support polylog operations \cite{Mun84} and probabilistic checking of polylog bits of proofs \cite{BFLS91}. We denote by $\ppolylogtime$ the complexity class of decision problems solvable by polylogtime probabilistic TMs with unbounded-error probability.

In the setting of quantum computing, nevertheless, we are able to take a similar approach using resource-bounded QTMs. A quantum analogue of $\dlogtime$, called $\bqlogtime$, was lately discussed in connection to an oracle separation between $\bqp$ and the polynomial hierarchy in \cite{RT22}.
In this work, however, we are interested in polylogtime-bounded quantum computations, and thus we wish to look into the characteristics of quantum polylogtime computability.
Later in Section \ref{sec:polylogtime-QTM}, we formally describe the fundamental model of \emph{polylogtime QTMs} as a quantum analogue of classical polylogtime TMs.
These polylogtime QTMs naturally induce the corresponding bounded-error complexity class $\bqpolylogtime$.
In Section \ref{sec:BQLOGTIME}, $\bqpolylogtime$ is shown to differ in computational power from $\nlogtime$ as well as $\ppolylogtime$.

\subsection{Our Challenges in This Work}\label{sec:challenge}

We intend to extend the scope of the existing research on resource-bounded quantum computability by way of an introduction of a small set of recursion schemes defining runtime-restricted quantum computability.
In particular, we look into quantum polylogtime computability and seek out an appropriate recursion schematic definition to exactly capture such computability.
Since $\squareqp$ and $\hatsquareqp$ precisely capture quantum polynomial-time computability, it is natural to raise a question of how we can capture quantum polylogtime computability in a similar fashion.

As noted earlier, the multi-qubit quantum recursion scheme of \cite{Yam20} is capable of precisely capturing quantum polynomial-time computability.
A basic idea of this quantum recursion scheme is, starting with $n$ qubits,  to modify at each round of recursion the first $k$ qubits and continue to the next round with the rest of the qubits after discarding the first $k$ qubits.
This recursive process slowly consumes the entire $n$ qubits and finally grinds to halt after linearly-many recursive rounds.
To describe limited quantum computability, in contrast, how can we weaken the quantum recursion scheme?
Our attempt in this work is to speed up this recursive process significantly by discarding a bundle of $\ceilings{n/2}$(or $\floors{n/2}$) qubits from the entire $n$ qubits at each round. Such a way of halving the number of qubits at each recursive round makes the recursive process terminate significantly faster; in fact, ending in at most $\ceilings{\log{n}}$ recursive rounds. For this very reason, we intend to call this weakened recursive process the \emph{(code-controlled) fast quantum recursion scheme}.

Although the fast quantum recursion scheme is significantly weaker in power than the quantum recursion scheme of \cite{Yam20}, the scheme still generates numerous interesting and useful quantum functions, as listed in Lemmas \ref{Skip-function}--\ref{AND-OR-function} and \ref{SIZE-function}--\ref{various-schemes}.
More importantly, it is possible to implement the binary search strategy, which has been a
widely-used key programming technique.
We call the set of all quantum functions generated from six initial quantum functions and two   construction schemes together with the fast quantum recursion scheme by $EQS$ (elementary quantum schemes) in  Section \ref{sec:definitions}.
This class $EQS$ turns out to precisely characterize quantum polylogtime  computability (Theorems \ref{QTM-simulation}--\ref{converse-simulation}).

To cope with polylogtime computability, we need an appropriate encoding of various types of ``objects'' (such as numbers, graphs, matrices, tape symbols) into qubits of the same length and also another efficient way of decoding the original ``objects'' from those encoded qubits. For this purpose, we introduce in Section \ref{sec:different-object} a  code-skipping scheme, which recognizes encoded segments to skip them one by one. This scheme in fact plays a key role in performing the fast recursion scheme.

Throughout this work, we try to promote much better understandings of resource-bounded quantum computing in theory and in practice.


\bs
\n{\bf New Materials after the Preliminary Conference Report.}\hs{2}
The current work corrects and significantly alters the preliminary conference report \cite{Yam22b}, particularly in the following points. Five schemes, which constitute $EQS$, in the preliminary report have been modified and reorganized, and thus they look slightly different in their formulations. In particular, Scheme V is introduced to establish a precise characterization of quantum functions computable by polylogtime QTMs.

Beyond the scope of the preliminary report \cite{Yam22b}, this work further looks into another restricted quantum algorithmic procedure, known as the \emph{divide-and-conquer strategy} in Section \ref{sec:divide-and-conquer}. A key idea behind this strategy is to inductively split an given instance into small pieces and then inductively assemble them piece by piece after an appropriate modification.
This new scheme helps us capture the parity function, which requires more than polylogtime quantum computation.
This fact implies that this new procedure is not ``realized'' in the framework of $EQS$.

\section{Preparation: Notions and Notation}\label{sec:preparation}

We will briefly explain basic notions and notation used throughout this work.

\subsection{Numbers, Strings, and Languages}\label{sec:numbers}

The notations $\integer$, $\rational$, and $\real$ respectively denote the sets of all integers, of all rational numbers, and of all real numbers. In particular, we set $\bar{\rational}=\rational\cap [0,1]$.
The \emph{natural numbers} are nonnegative integers and the notation $\nat$ expresses the set of those numbers.
We further write $\nat^{+}$ to denote the set $\nat-\{0\}$.
For two integers $m,n$ with $m\leq n$, $[m,n]_{\integer}$ indicates the \emph{integer interval} $\{m,m+1,m+2,\ldots,n\}$. In particular, we abbreviate $[1,n]_{\integer}$ as $[n]$ when $n\geq1$.
In addition, $\complex$ denotes the set of all complex numbers.
We use the notations $\ceilings{\cdot}$ and $\floors{\cdot}$ for the ceiling function and the floor function, respectively.
In this work, all \emph{polynomials} have nonnegative integer coefficients and all \emph{logarithms} are taken to the base $2$.
We further define $\ilog(n)$ and $\iloglog(n)$ to be $\ceilings{\log{n}}$  and $\ceilings{\log\log{n}}$ for any $n\in\nat$, respectively.
To circumvent any cumbersome description to avoid the special case of $n=0$, we intentionally set $\ilog{0}=\iloglog(0)=0$.
The notation $\imath$ denotes $\sqrt{-1}$ and $e$ does the base of natural logarithms.  Given a complex number $\alpha$, $\alpha^*$ denotes the \emph{complex conjugate} of $\alpha$.
As for a nonempty set of quantum amplitudes, which is a subset of $\complex$, we use the notation $K$.

A nonempty finite set of ``symbols'' (or ``letters'') is called an \emph{alphabet} and a sequence of such symbols from a fixed alphabet $\Sigma$ is called a \emph{string} over $\Sigma$. The total number of occurrences of symbols in a given string $x$ is the \emph{length} of $x$ and is denoted $|x|$. The \emph{empty string} $\lambda$ is a unique string of length $0$.
A collection of those strings over $\Sigma$ is a \emph{language} over $\Sigma$.
Throughout this work, we deal only with \emph{binary strings}, which are sequences of 0s and 1s. Given a number $n\in\nat$, $\Sigma^{n}$ (resp., $\Sigma^{\leq n}$) denotes the set of all strings of length exactly $n$ (resp., at most $n$). Let $\Sigma^* = \bigcup_{n\in\nat}\Sigma^n$ and $\Sigma^{+} = \Sigma^*-\{\lambda\}$.
Given a string $x$ and a bit $b\in\{0,1\}$, the notation $\#_b(x)$ denotes the total number of occurrences of the symbol $b$ in $x$.
Due to the nature of quantum computation, we also discuss ``promise'' decision problems. A pair $(A,R)$ of subsets of $\Sigma^*$ is said to be a  \emph{promise (decision) problem} over $\Sigma$ if $A\cup R \subseteq \Sigma^*$ and $A\cap R =\setempty$. Intuitively, $A$ and $R$ are  respectively composed of accepted strings and rejected strings. When $A\cup R=\Sigma^*$, $(A,R)$ becomes $(A,\Sigma^*-A)$ and this can be identified with the language $A$.

Given a number $n\in\nat$, we need to partition it into two halves. To describe the ``left half'' and the ``right half'' of $n$, we introduce two special functions: $LH(n) = \ceilings{n/2}$ and $RH(n) = \floors{n/2}$.

For a function $f$ and a number $k\in\nat^{+}$, we write $f^k$ for the $k$ consecutive applications of $f$ to an input. For example, $f^{2}(x) = f\circ f(x)=f(f(x))$ and $f^{3}(x) = f\circ f \circ f(x) = f( f ( f(x)))$.
The functions $LH$ and $RH$ satisfy the following property.

\begin{lemma}\label{HALF-value}
Let $k,n\in\nat^{+}$ be any two numbers. (1) If $n$ is in $[2^{k-1}+1,2^k]_{\integer}$ (i.e., $\ceilings{\log{n}}=k$), then $LH^{k}(n)=1$ holds. Moreover, if $k\geq 2$, then $LH^{k-1}(n) =2$. (2) If $n$ is in $[2^{k},2^{k+1}-1]_{\integer}$ (i.e., $\floors{\log{n}}=k$), then $RH^{k}(n)=1$ holds. Moreover, if $k\geq 2$, then $RH^{k-1}(n)\in\{2,3\}$.
\end{lemma}

\begin{proof}
(1) Consider a series: $n$, $LH(n)$, $LH^2(n)$, $\cdots$, $LH^k(n)$.  It follows that, for any number $n > 2$, $n\in [2^{k-1}+1,2^k]_{\integer}$ iff  $LH(n) \in [2^{k-2}+1,2^{k-1}]_{\integer}$. Thus, we obtain $LH^{k-1}(n)\in\{2\}$ and $LH^{k}(n)\in\{1\}$. When $n\geq3$, the last two numbers of the above series must be $2$ and $1$.

(2) In a similar fashion to (1), let us consider a series: $n$, $RH(n)$, $RH^2(n)$, $\ldots$, $RH^k(n)$. It then follows that $n\in[2^k,2^{k+1}-1]_{\integer}$ iff $RH(n)\in [2^{k-1},2^k-1]_{\integer}$. We then obtain $RH^{k-1}(n)\in\{2,3\}$ and $RH^{k}(n)\in\{1\}$.
\end{proof}

We assume the standard \emph{lexicographical order} on $\{0,1\}^*$: $\lambda,0,1,00,01,10,11,000,\cdots$ and, based on this ordering, we assign natural numbers to binary strings as $bin(0)=\lambda$, $bin(1)=0$, $bin(2)=1$, $bin(3)=00$, $bin(4)=01$, $bin(5)=10$, $bin(6)=11$, $bin(7)=000$, etc.
In contrast, for a fixed number $k\in\nat^{+}$, $bin_k(n)$ denotes lexicographically the $n$th string in $\{0,1\}^k$ as long as  $n\in[2^k]$. For instance, we obtain $bin_3(1)=000$, $bin_3(2)=001$, $bin_3(3)=010$, $bin_3(4)=011$, etc. Remember that $bin_k(0)$ is not defined whereas $bin(0)$ means $\lambda$.

\subsection{Quantum States and Hilbert Spaces}

We assume the reader's familiarity with basic quantum information and computation (see, e.g., textbooks \cite{KSV02,NC16}).  A basic concept in quantum computing is a Hilbert space and unitary transformations over it. The \emph{ket notation} $\qubit{\phi}$ expresses a (column) vector in a Hilbert space and its transposed conjugate is expressed as a (row) vector of the dual space and is denoted by the \emph{bra notation} $\bra{\phi}$.
The notation $\measure{\psi}{\phi}$ denotes the \emph{inner product} of two vectors $\qubit{\psi}$ and $\qubit{\phi}$.
We use the generic notation $I$ to denote the identity matrix of an arbitrary dimension. A square complex matrix $U$ is called \emph{unitary} if $U$ satisfies $UU^{\dagger} = U^{\dagger}U = I$, where $U^{\dagger}$ is the transposed conjugate of $U$.

The generic notation $\bfzero$ is used to denote the \emph{null vector}  of an arbitrary dimension.
A \emph{quantum bit} (or a \emph{qubit}, for short) is a linear combination of two basis vectors $\qubit{0}$ and $\qubit{1}$.
The notation $\HH_2$ indicates the Hilbert space spanned by $\qubit{0}$ and $\qubit{1}$. More generally, $\HH_{2^n}$ refers to the Hilbert space spanned by the computational basis $B_n=\{\qubit{s}\mid s\in\{0,1\}^n\}$.
For convenience, we write $\HH_{\infty}$ for the collection of all quantum states in $\HH_{2^n}$ for any $n\in\nat^{+}$. Remember that $\HH_{\infty}$ does not form a Hilbert space.
Given a non-zero quantum state $\qubit{\phi}\in\HH_{\infty}$, its \emph{length} $\ell(\qubit{\phi})$ denotes a unique number $n\in\nat$ for which  $\qubit{\phi}\in\HH_{2^n}$. We stress that the length of $\qubit{\phi}$ is defined only for a quantum state residing in $\HH_{\infty}$.
As a special case, we set $\ell(\bfzero)=0$ for the null vector $\bfzero$ (although $\bfzero$ belongs to $\HH_k$ for any $k\in\nat^{+}$). For convenience, we also set $\ell(\alpha)=0$ for any scalar $\alpha\in\complex$.
For a quantum state $\qubit{\phi}$, if we make a measurement in the computational basis, then each binary string $x$ of length $\ell(\qubit{\phi})$ is observed with probability $|\measure{x}{\phi}|^2$.
Thus, $\qubit{\phi}$ can be expressed as $\ket{0}\measure{0}{\phi} + \ket{1}\measure{1}{\phi}$ and also as $\sum_{x\in\{0,1\}^n} \qubit{x}\otimes \measure{x}{\phi}$, where $\qubit{\phi}\otimes \qubit{\psi}$ is the \emph{tensor product} of $\qubit{\phi}$ and $\qubit{\psi}$.
The \emph{norm} of $\qubit{\phi}$ is defined as $\sqrt{\measure{\phi}{\phi}}$ and is denoted $\|\qubit{\phi}\|$.

A \emph{qustring of length $n$} is a unit-norm quantum state in $\HH_{2^n}$. As a special case, the null vector  is treated as the \emph{qustring of length $0$} (although its norm is $0$). Let $\Phi_{n}$ denote the set of all qustrings of length $n$. Clearly, $\Phi_{n}\subseteq \HH_{2^n}$ holds for all $n\in\nat$. We then set $\Phi_{\infty}$ to be the collection of all qustrings of length $n$ for any $n\in\nat$.

Abusing the aforementioned notations of $LH(n)$ and $RH(n)$, we set $LH(\qubit{\phi}) = \ceilings{\ell(\qubit{\phi})/2}$ and $RH(\qubit{\phi})=\floors{\ell(\qubit{\phi}/2}$ for any quantum state $\qubit{\phi}\in\HH_{\infty}$. It then follows that $\ell(\qubit{\phi})=LH(\qubit{\phi})+RH(\qubit{\phi})$.

We are interested only in functions mapping $\HH_{\infty}$ to $\HH_{\infty}$ and these functions are generally referred to as \emph{quantum functions on $\HH_{\infty}$} to differentiate from ``functions'' working with natural numbers or strings. Such a quantum function $f$ is said to be {\em dimension-preserving} (resp., \emph{norm-preserving}) if  $\ell(\qubit{\phi}) = \ell(f(\qubit{\phi}))$
(resp., $\|\qubit{\phi}\|=\|f(\qubit{\phi})\|$) holds for any quantum state $\qubit{\phi}\in\HH_{\infty}$.

\subsection{Conventions on the Bra and the Ket Notations}\label{sec:bra-ket-notation}

To calculate the outcomes of quantum functions on $\HH_{\infty}$, we must take advantage of making a purely symbolic treatment of the bra and the ket notations together with the tensor product notation $\otimes$.
We generally follow \cite{Yam20} for the conventions on the specific usage of these  notations. These notational conventions in fact help us simplify the description of various qubit operations in later sections. Since we greatly deviate from the standard usage of those notations, hereafter, we explain how to use them throughout this work.

It is important to distinguish between the number $0$ and the null vector $\bfzero$. To deal with the null vector, we conveniently take the following specific conventions concerning the operator ``$\otimes$''.
For any $\qubit{\phi}\in\HH_{\infty}$, (i) $0\otimes \qubit{\phi} = \qubit{\phi}\otimes 0 =\bfzero$ (scalar case), (ii) $\qubit{\phi}\otimes \bfzero = \bfzero\otimes \qubit{\phi}=\bfzero$, and (iii) if $\qubit{\psi}$ is the null vector, then $\measure{\phi}{\psi} = \measure{\psi}{\phi}=\bfzero$.
In the case of (i), we here extend the standard usage of $\otimes$ to even cover ``scalar multiplication'' to simplify later calculations.
Given two binary strings $u,s\in\{0,1\}^+$, if $|s|=|u|$, then $\measure{u}{s}=1$ if $u=s$, and $\measure{u}{s}=0$ otherwise.
On the contrary, when $|u|<|s|$, $\measure{u}{s}$ expresses the quantum state $\qubit{z}$ if $s=uz$ for a certain string $z$, and $0$ otherwise.
By sharp contrast, if $|u|>|s|$, then  $\measure{u}{s}$ always denotes $\bfzero$.
More generally, if $\qubit{\phi}=\sum_{u\in\{0,1\}^n}\sum_{s\in\{0,1\}^m} \alpha_{u,s}\qubit{u}\qubit{s}$ and $u_0\in\{0,1\}^n$, then $\measure{u_0}{\phi}$ expresses the quantum state $\sum_{s\in\{0,1\}^m}\alpha_{u_0,s}\qubit{s}$.
For instance, if $\qubit{\phi}=\frac{1}{\sqrt{2}}(\qubit{00}+\qubit{11})$, then $\measure{0}{\phi}$ equals $\frac{1}{\sqrt{2}}\qubit{0}$ and $\measure{1}{\phi}$ does $\frac{1}{\sqrt{2}}\qubit{1}$.
These notations $\measure{0}{\phi}$ and $\measure{1}{\phi}$  together make it possible to express any qustring $\qubit{\phi}$ as $\ket{0}\otimes \measure{0}{\phi} + \ket{1}\otimes \measure{1}{\phi}$. In a more general case with $\qubit{\xi}=\sum_{u\in\{0,1\}^n} \beta_u\qubit{u}$, $\measure{\xi}{\phi}$ means  $\sum_{u\in\{0,1\}^n} \beta_u^{*} \measure{u}{\phi}$, which equals $\sum_{u\in\{0,1\}^n} \sum_{s\in\{0,1\}^m} \alpha_{u,s}\beta_u^{*}\qubit{s}$.
Moreover, $\qubit{\phi}$ can be expressed as $\sum_{u\in\{0,1\}^n} \beta_u \ket{u}\otimes \measure{u}{\phi}$.
We often abbreviate $\qubit{\phi}\otimes \qubit{\psi}$ as $\qubit{\phi}\qubit{\psi}$. In particular, for strings $s$ and $u$, $\qubit{s}\otimes \qubit{u}$ is further abbreviated as $\qubit{s,u}$ or even $\qubit{su}$.
When $\measure{s}{\phi}$ is just a scalar, say, $\alpha\in\complex$, the notation $\qubit{s}\otimes \measure{s}{\phi}$ (or equivalently, $\qubit{s}\measure{s}{\phi}$) is equal to  $\alpha\qubit{s}$.

We also expand the norm notation $\|\cdot \|$ for vectors to scalars in the following way. For any quantum states $\qubit{\phi}$ and $\qubit{\psi}$ with $\ell(\qubit{\phi}) = \ell(\qubit{\psi})>0$, the notation $\|\measure{\phi}{\psi}\|$ is used to express the absolute value $|\measure{\phi}{\psi}|$. In general, for any scalar $\alpha\in\complex$, $\|\alpha\|$ denotes $|\alpha|$. This notational convention is quite useful in handling the value obtained after making a measurement without worrying about whether the resulting object is a quantum state or a scalar.

Later, we will need to encode (or translate) a series of symbols into an appropriate qustring. In such an encoding, the empty string $\lambda$ is treated quite differently. Notably, we tend to automatically translate $\qubit{\lambda}$ into the null vector $\bfzero$ unless otherwise stated.

\section{A Recursion Schematic Definition of EQS}\label{sec:definitions}

We will present a recursion schematic definition to formulate a class of special quantum functions on
$\HH_{\infty}$ (i.e., from $\HH_{\infty}$ to $\HH_{\infty}$), later called $EQS$. To improve the readability, we first provide a skeleton system of $EQS$ and later we expand it to the full-fledged system.

\subsection{Skeleton EQS}\label{sec:elementary-scheme}

As a starter, we discuss a ``skeleton'' of our recursion schematic definition for ``elementary'' quantum functions, which are collectively called $EQS$ (elementary quantum schemes) in Definition \ref{EQS-definition}, involving six initial quantum functions and two construction schemes. All initial quantum functions were already presented when defining $\squareqp$ and $\hatsquareqp$ in \cite{Yam20}.

Throughout the rest of this work, we fix a nonempty amplitude set $K$  and we often ignore the clear reference to $K$ as long as the choice of $K$ is not important in our discussion.

\begin{definition}\label{def:initial}
The skeleton class $EQS_0$ is composed of all quantum functions  constructed by selecting the initial quantum functions of Scheme I and then inductively applying Schemes II--III a finite number of times. In what follows,  $\qubit{\phi}$ denotes an arbitrary quantum state in $\HH_{\infty}$. When the item 6) of Scheme I is not used, we denote the resulting class by $\widehat{EQS}_0$.

\begin{enumerate}\vs{-2}
  \setlength{\topsep}{-2mm}%
  \setlength{\itemsep}{2mm}%
  \setlength{\parskip}{0cm}%

\item[I.] The \emph{initial quantum functions}.  Let
$\theta\in[0,2\pi)\cap K$ and $a\in\{0,1\}$. \vs{1}

1) $I(\qubit{\phi}) = \qubit{\phi}$. (identity) \vs{1}

2) $PHASE_{\theta}(\qubit{\phi}) = \qubit{0}\!\measure{0}{\phi} + e^{\imath\theta}\qubit{1}\!\measure{1}{\phi}$. (phase shift) \vs{1}

3) $ROT_{\theta}(\qubit{\phi}) = \cos\theta \qubit{\phi}
+ \sin\theta
(\qubit{1}\!\measure{0}{\phi} - \qubit{0}\!\measure{1}{\phi})$. (rotation around $xy$-axis at angle $\theta$) \vs{-1}

4) $NOT(\qubit{\phi}) = \qubit{0}\!\measure{1}{\phi} +
\qubit{1}\!\measure{0}{\phi}$. (negation) \vs{1}

5) $SWAP(\qubit{\phi}) = \left\{ \begin{array}{ll}
 \qubit{\phi} & \mbox{if $\ell(\qubit{\phi})\leq 1$,} \\
 \sum_{a,b\in\{0,1\}} \qubit{ab}\!\measure{ba}{\phi}
 & \mbox{otherwise.}
 \end{array}\right.$ \text{(swapping of 2 qubits)} \vs{1}

6) $MEAS[a](\qubit{\phi}) = \qubit{a}\!\measure{a}{\phi}$. (partial projective measurement)

\item[II.] The \emph{composition scheme}.
From $g$ and $h$, we define $Compo[g,h]$ as
follows: \vs{1}

\n\hs{5}$Compo[g,h](\qubit{\phi}) = g\circ h(\qubit{\phi})$
($=g(h(\qubit{\phi}))$).

\item[III.] The \emph{branching scheme}.\footnote{We remark that $Branch[g,h]$ can be expressed as an appropriate unitary matrix (if $g$ and $h$ are expressed as unitary matrices) and thus it is a legitimate quantum operation to consider.}
From $g$ and $h$, we define $Branch[g,h]$
as: \vs{1}

\n\hs{2}(i) $Branch[g,h](\qubit{\phi}) =
\qubit{\phi}$ \hs{44}if $\ell(\qubit{\phi})\leq 1$, \vs{1} \\
\n\hs{1}(ii) $Branch[g,h](\qubit{\phi}) =
\qubit{0}\otimes g(\measure{0}{\phi}) + \qubit{1}\otimes
h(\measure{1}{\phi})$ \hs{3}otherwise.
\end{enumerate}
\end{definition}

We remark that all quantum functions in Items 1)--4) and 6) in Scheme I directly manipulate only the first qubit of $\qubit{\phi}$ and that $SWAP$ manipulates the first two qubits of $\qubit{\phi}$. All other qubits of $\qubit{\phi}$ are intact.
Notice that the length function $\ell(\qubit{\phi})$ is not included as part of $EQS_0$.
Since $\qubit{\phi}$ is always taken from $\HH_{\infty}$, the value $\ell(\qubit{\phi})$ is uniquely determined from $\qubit{\phi}$.
It is also important ro remark that, for any $\qubit{\phi}$ and $a\in\{0,1\}$, $\ell(MEAS[a](\qubit{\phi})) = \ell(\qubit{\phi})$.

In Schemes II and  III, the constructions of $Compo[g,h]$ and $Branch[g,h]$ need two quantum functions, $g$ and $h$, which are assumed to have been already constructed in their own construction processes.
To refer to these supporting quantum functions used in the scheme, we call them the \emph{ground (quantum) functions} of $Compo[g,h]$ and $Branch[g,h]$.
In the case of a special need to emphasize $K$, then we write $EQS_{K,0}$ and $\widehat{EQS}_{K,0}$ with a clear reference to $K$.

We quickly provide simple examples of how to compute the quantum functions induced by appropriate applications of Schemes I--III.

\begin{example}
It follows that $PHASE_{\pi}(\qubit{1}) = - \qubit{1}$ and $PHASE_{\pi/2}(\frac{1}{\sqrt{2}}(\qubit{0}+\qubit{1})= \frac{1}{\sqrt{2}}(\qubit{0}+\imath\qubit{1})$, $ROT_{\pi/4}(\qubit{0}) = \frac{1}{\sqrt{2}}(\qubit{0}+\qubit{1})$ and $ROT_{\pi/4}(\qubit{1}) = \frac{1}{\sqrt{2}}(\qubit{0}-\qubit{1})$, $NOT(\qubit{a})=\qubit{1-a}$, and $SWAP(\qubit{abc})=\qubit{bac}$ for three bits $a,b,c\in\{0,1\}$.
We also obtain $MEAS[0](ROT_{\pi/4}(\qubit{00})) = MEAS[0](\frac{1}{\sqrt{2}}(\qubit{0}+\qubit{1})\otimes \qubit{0}) = \frac{1}{\sqrt{2}}\qubit{00}$.
As for Scheme III, it is important to note that
$Branch[g,g]$ is not the same as $g$ itself. For example, $Branch[NOT,NOT](\qubit{b}\qubit{\phi})$ equals $\qubit{b}\otimes NOT(\qubit{\phi})$ by Scheme III(ii) whereas $NOT(\qubit{b}\qubit{\phi})$ is just $\qubit{1-b}\qubit{\phi}$, where $b$ denotes any bit.
The EPR pair $\frac{1}{\sqrt{2}}(\qubit{00}+\qubit{11})$ is obtainable from $\qubit{0}\qubit{0}$ as follows. By first applying $ROT_{\pi/4}$ to $\qubit{0}\qubit{0}$, we obtain $\qubit{\phi} = ROT_{\pi/4}(\qubit{0}\qubit{0})$. We then apply $Branch[I,NOT]$ and obtain $Branch[I,NOT](\qubit{\phi}) = Branch[I,NOT](\frac{1}{\sqrt{2}}(\qubit{00}+\qubit{10})) =
\qubit{0}\otimes I(\frac{1}{\sqrt{2}}\ket{0}) + \qubit{1}\otimes
NOT(\frac{1}{\sqrt{2}}\ket{0}) =
\frac{1}{\sqrt{2}}(\qubit{00}+\qubit{11})$.
\end{example}


For later reference, let us review from \cite{Yam20} the \emph{multi-qubit quantum recursion scheme}, which is a centerpiece of schematically characterizing quantum polynomial-time computability. In this work, this special scheme is referred to as {\em Scheme T}.

{\it
\begin{itemize}\vs{-2}
  \setlength{\topsep}{-2mm}%
  \setlength{\itemsep}{1mm}%
  \setlength{\parskip}{0cm}%

\item[{\rm (T)}] The \emph{multi-qubit quantum recursion scheme} {\rm \cite{Yam20}}.
Let $g$, $h$, $p$ be quantum functions and let $k,t\in\nat^{+}$ be numbers. Assume that $p$ is dimension-preserving. For  $\FF_k=\{f_u\}_{u\in\{0,1\}^k}$, we define
$F\equiv kQRec_t[g,h,p|\FF_k]$ as follows:

\vs{2}
\n\hs{4}{\rm (i)} $F(\qubit{\phi}) = g(\qubit{\phi})$
\hs{39} if $\ell(\qubit{\phi})\leq t$, \vs{1} \\
\n\hs{3}{\rm (ii)} $F(\qubit{\phi}) =  h( \sum_{u:|u|=k} \qubit{u} \otimes  f_u (  \measure{u}{\psi_{p,\phi}} ))$ \hs{6}otherwise, \\
where $\qubit{\psi_{p,\phi}} = p(\qubit{\phi})$ and $f_u\in \{F,I\}$ for any $u\in\{0,1\}^k$.
\end{itemize}
}

This quantum recursion scheme together with Schemes I--III given above and another quantum function called $REMOVE$ (removal) defines the function class $\squareqp$, which constitutes all polynomial-time computable quantum functions \cite{Yam20}.
Although we do not consider $REMOVE$ here, its restricted form, called $CodeREMOVE[\cdot]$, will be introduced in Section \ref{sec:different-object}.


Let us present a short list of typical quantum functions ``definable'' within $EQS_0$; that is, those quantum functions are actually constructed from the items 1)--5) of Scheme I and by finitely many applications of Schemes II--III.
As the first simple application of Scheme I--III, we demonstrate how to construct  the special quantum function called $Skip_k[\cdot]$.

\begin{lemma}\label{Skip-function}
Let $g$ denote any quantum function in $EQS_0$ and let $k\in\nat^{+}$. The quantum function $Skip_k[g]$ is defined by $Skip_k[g](\qubit{\phi}) = \qubit{\phi}$ if $\ell(\qubit{\phi})\leq k$, and $\sum_{u:|u|=k} \ket{u}\otimes g(\measure{u}{\phi})$ otherwise.
This $Skip_k[g]$ is definable within $EQS_0$.
\end{lemma}

\begin{proof}
We construct the desired quantum function $Skip_k[g]$ inductively for any  $k\in\nat^{+}$. When $k=1$, we set $Skip_1[g] \equiv Branch[g,g]$.
Clearly, when $\ell(\qubit{\phi})>1$, we obtain $Skip_1[g](\qubit{\phi}) = Branch[g,g](\qubit{\phi}) = \sum_{a\in\{0,1\}} \qubit{a}\otimes g(\measure{a}{\phi})$.
For any index $k\geq2$, we define $Skip_{k+1}[g]$ as $Branch[Skip_k[g],Skip_k[g]]$.
It then follows that, if $\ell(\qubit{\phi})>k+1$, then  $Skip_{k+1}[g](\qubit{x}\qubit{\phi}) = \sum_{a\in\{0,1\}} \qubit{a}\otimes Skip_k[g](\measure{a}{\phi}) = \sum_{a\in\{0,1\}} \sum_{s:|s|=k} \qubit{a}\qubit{s}\otimes g(\measure{s}{\psi_{Skip_k[g],\measure{a}{\phi}}}) = \sum_{s':|s'|=k+1} \qubit{s'}\otimes g(\measure{s'}{\phi})$, as requested.
\end{proof}

\begin{lemma}\label{lemma:special}
Fix $\theta\in[0,2\pi)\cap K$ and $i,j,k\in\nat^{+}$  and $i<j$. Let $\qubit{\phi}$ be any quantum state in $\HH_{\infty}$, let $a$, $b$, and $c$ be any bits. The following quantum functions are definable by Schemes I--III and thus in $EQS_0$.
\begin{enumerate}\vs{-2}
  \setlength{\topsep}{-2mm}%
  \setlength{\itemsep}{1mm}%
  \setlength{\parskip}{0cm}%

\item $CNOT(\qubit{\phi}) = \qubit{\phi}$ if $\ell(\qubit{\phi})\leq 1$ and $CNOT(\qubit{\phi}) = \ket{0}\measure{0}{\phi} + \ket{1}\otimes NOT(\measure{1}{\phi})$ otherwise. (controlled NOT)

\item $GPS_{\theta}(\qubit{\phi}) = e^{\imath\theta}\qubit{\phi}$. (global phase shift)

\item $WH(\qubit{\phi}) = \frac{1}{\sqrt{2}}\ket{0}\otimes (\measure{0}{\phi}+\measure{1}{\phi}) + \frac{1}{\sqrt{2}}\ket{1}\otimes (\measure{0}{\phi}-\measure{1}{\phi})$. (Walsh-Hadamard transform)

\item $Z_{1,\theta}(\qubit{\phi}) = e^{\imath\theta}\ket{0}\measure{0}{\phi} + \ket{1}\measure{1}{\phi}$.

\item $zROT_{\theta}(\qubit{\phi}) =  e^{\imath\theta}\ket{0}\measure{0}{\phi} + e^{-\imath\theta} \ket{1}\measure{1}{\phi}$. (rotation around the $z$-axis)

\item $C_{\theta}(\qubit{\phi}) = \qubit{\phi}$ if $\ell(\qubit{\phi})\leq 1$ and $C_{\theta}(\qubit{\phi}) = \ket{0}\measure{0}{\phi} + \ket{1}\otimes ROT_{\theta}(\measure{1}{\phi})$ otherwise. (controlled ROT$_{\theta}$)

\item $CPHASE_{\theta} (\qubit{\phi}) = \qubit{\phi}$ if $\ell(\qubit{\phi})\leq 1$, and $CPHASE_{\theta} (\qubit{\phi}) =  \frac{1}{\sqrt{2}}\sum_{b\in\{0,1\}} (\ket{0}\measure{b}{\phi} + e^{\imath\theta b}\ket{1}\measure{b}{\phi})$ otherwise. (controlled PHASE)

\item $CSWAP(\qubit{a}\qubit{\phi}) = \ket{0}\otimes \measure{0}{a}\ket{\phi} + \ket{1}\otimes SWAP( \measure{1}{a}\ket{\phi} )$ (controlled SWAP)

\item $LengthQ_k(\qubit{b}\qubit{\phi}) = \qubit{b}\qubit{\phi}$ if $\ell(\qubit{\phi})< k$  and $LengthQ_k(\qubit{b}\qubit{\phi}) = \qubit{1-b}\qubit{\phi}$ otherwise. (length query)

\item $SWAP_{i,j}(\qubit{\phi}) = \sum_{a_1\cdots a_j\in\{0,1\}} \ket{a_1\cdots a_{i-1}a_{j}a_{i+1} \cdots a_{j-1}a_{i} a_{j+1}\cdots a_n} \otimes \measure{a_1\cdots a_{i-1} a_i a_{i+1}\cdots a_{j-1} a_j a_{j+1}\cdots a_n}{\phi}$ if $\ell(\qubit{\phi})\geq j$ and $SWAP_{i,j}(\qubit{\phi}) =\qubit{\phi}$ otherwise, where $n=\ell(\qubit{\phi})$.
\end{enumerate}
\end{lemma}

\begin{proof}
(1)--(5) These quantum functions are constructed in \cite[Lemma 3.3]{Yam20}. (6) We set $C_{\theta} \equiv Branch[I,ROT_{\theta}]$.
(7) This was shown in \cite[Lemma 3.6]{Yam20}.
(8) We set $CSWAP$ to be $Branch[I,SWAP]$. This gives the equation  $CSWAP(\qubit{a}\qubit{\phi}) = \qubit{0}\otimes I(\measure{0}{a}\ket{\phi}) + \qubit{1}\otimes SWAP(\measure{1}{a}\ket{\phi})$.

(9) For simplicity, we write $g\circ f$ in place of $Compo[g,f]$. Let $h\equiv Branch[NOT,NOT]\circ SWAP$.
We construct $LengthQ_k$ inductively for any $k\geq1$ as follows. Firstly, when $k=1$,  $LengthQ_1$ is set to be $SWAP\circ Skip_1[SWAP] \circ Skip_1[h]\circ SWAP$. It then follows that $LengthQ_1(\qubit{b}) = \qubit{b}$ and $LengthQ_1(\qubit{b}\qubit{\phi}) = \qubit{1-b}\qubit{\phi}$. Letting $k\geq2$, assume that $LengthQ_{k-1}$ has been already defined. We then define $LengthQ_{k}$ as $SWAP \circ Skip_1[Branch[LengthQ_{k-1},LengthQ_{k-1}]] \circ SWAP$.

(10) Notice that $SWAP_{1,2}$ coincides with $SWAP$ of Scheme I.
For any fixed constant $k\geq1$, we first define $SWAP_{k,k+1}$ to be $Skip_{k-1}[SWAP]$.
For any indices $i,j\geq1$ with $i<j$, we further define $MOVE_{i,j}$ to be $SWAP_{i,i+1}\circ SWAP_{i+1,i+2}\circ \cdots \circ SWAP_{j-1,j}$.
Note that
$MOVE_{1,j-1}^{-1}$ equals $SWAP_{j-2,j-1}\circ SWAP_{j-3,j-2} \circ \cdots \circ SWAP_{1,2}$.
With the use of $MOVE_{1,j}$, we define $SWAP_{1,j}$ as $MOVE_{1,j}\circ MOVE_{1,j-1}^{-1}$. The target quantum function $SWAP_{i,j}$ is finally set to be $Skip_{i-1}[SWAP_{1,j-i}]$.
\end{proof}

\sloppy
\begin{lemma}\label{lemma:special-combi}
Fix $i,j,k\in\nat^{+}$ with $k\geq2$ and $i<j$, and let $\qubit{\phi}$ denote  any quantum state in $\HH_{\infty}$. The following quantum functions are definable using Schemes I--III and thus in $EQS_0$.
\begin{enumerate}\vs{-3}
  \setlength{\topsep}{-2mm}%
  \setlength{\itemsep}{1mm}%
  \setlength{\parskip}{0cm}%

\item $SecSWAP^{(k)}_{i,j}(\qubit{x_1}\qubit{x_2}\cdots \qubit{x_{i-1}}\qubit{x_i}\qubit{x_{i+1}} \cdots \qubit{x_{j-1}}\qubit{x_j}\qubit{\phi})
    = \qubit{x_1}\qubit{x_2}\cdots \qubit{x_{i-1}}\qubit{x_j}\qubit{x_{i+1}} \cdots \qubit{x_{j-1}}\qubit{x_i}\qubit{\phi}$ for any $k$-bit strings $x_1,x_2,\ldots,x_i,\ldots,x_j\in\{0,1\}^k$. (section SWAP)

\item $SecMOVE^{(k)}_{i,j}(\qubit{x_1}\qubit{x_2}\cdots \qubit{x_{i-1}}\qubit{x_i}\qubit{x_{i+1}} \cdots \qubit{x_j}\qubit{\phi})
    = \qubit{x_1}\qubit{x_2}\cdots \qubit{x_{i-1}}\qubit{x_{i+1}} \cdots \qubit{x_j}\qubit{x_i} \qubit{\phi}$ for any $k$-bit strings $x_1,x_2,\ldots,x_i,\ldots,x_j\in\{0,1\}^k$. (section move)
\end{enumerate}
\end{lemma}

\begin{proof}
(1) Since $x_1,x_2,\ldots,x_i,\ldots,x_j$ are $k$-bit strings, we obtain  $|x_1x_2\cdots x_{i-1}|=(i-1)k$ and $|x_1x_2\cdots x_{j-1}|=(j-1)k$. For convenience, we write $i_k$ for $(i-1)k$ and $j_k$ for $(j-1)k$.
We then define the desired quantum function $SecSWAP^{(k)}_{i,j}$ to be $SWAP_{i_k+1,j_k+1} \circ SWAP_{i_k+2,j_k+2}\circ \cdots \circ SWAP_{i_k+k,j_k+k}$.

(2) The quantum function $SecMOVE^{(k)}_{i,j}$ is set to be $SecSWAP^{(k)}_{j-1,j}\circ \cdots \circ SecSWAP^{(k)}_{i+1,i+2}\circ SecSWAP^{(k)}_{i,i+1}$.
\end{proof}


Hereafter, let us construct more quantum functions in $EQS_0$.

\begin{lemma}\label{COPY-function}
Consider $COPY_1$ that satisfies $COPY_1(\qubit{a}\qubit{\phi}) = \sum_{s\in\{0,1\}} \ket{a\oplus s}\qubit{s} \measure{s}{\phi}$ for any quantum state $\qubit{\phi}\in\HH_{\infty}$ and any bit $a$, where $\oplus$ denotes the bitwise XOR.
More generally, for each fixed constant $k\geq2$,
let $COPY_k(\qubit{x}\qubit{\phi}) = \sum_{z\in\{0,1\}^k} \qubit{x\oplus z}\qubit{z} \measure{z}{\phi}$  for any quantum state $\qubit{\phi}\in\HH_{\infty}$ and any $k$-bit string $x$.
These quantum functions are all definable using Schemes I--III.
\end{lemma}

\begin{proof}
We inductively define $COPY_k$ for any index $k\geq1$. In the case of $k=1$, we set  $COPY_1$ to be $SWAP\circ CNOT\circ SWAP$. It then follows that $COPY_1(\qubit{a}\qubit{\phi}) = SWAP\circ CNOT\circ SWAP (\qubit{a}\qubit{\phi}) = SWAP\circ CNOT (\sum_{s\in\{0,1\}} \ket{s}\ket{a}\otimes \measure{s}{\phi}) = \sum_{s\in\{0,1\}} SWAP(\ket{s}\ket{a\oplus s}\otimes \measure{s}{\phi}) = \sum_{s\in\{0,1\}} \ket{a\oplus s}\measure{s}{\phi}$.

Let $k\geq2$. Assume by induction hypothesis that $COPY_{k-1}$ has been already defined. The quantum function $COPY_k(\qubit{x}\qubit{\phi})$ is obtained by taking the following process. To $\qubit{x}\qubit{\phi}$,
we first apply $SWAP_{[2,k]} \equiv SWAP_{2,3}\circ SWAP_{3,4}\circ \cdots \circ SWAP_{k-1,k}$. Next, we apply $COPY_1$ and then $Skip_2[COPY_{k-1}]$. Finally, we apply $SWAP_{[2,k]}^{-1}$. It is not difficult to check that the obtained quantum function matches $COPY_k$.
\end{proof}


Let us construct the basic quantum functions $g_{AND}$ and $g_{OR}$, which ``mimic'' the behaviors of the two-bit operations $AND$ and $OR$, using only  Schemes I--III.

\begin{lemma}\label{AND-OR-function}
There exist two quantum functions $g_{AND}$ and $g_{OR}$ that satisfy the following. Let $x,y\in\{0,1\}$ and $b\in\{0,1\}$.
(1) $AND(x,y)=b$ iff $\|\measure{b}{\psi_{0xy}^{(AND)}}\| =1$, where $\qubit{\psi_{0xy}^{(AND)}} = g_{AND}(\qubit{0}\qubit{x}\qubit{y})$.
(2) $OR(x,y)=b$ iff $\|\measure{b}{\psi_{0xy}^{(OR)}}\| =1$, where $\qubit{\psi_{0xy}^{(OR)}} = g_{OR}(\qubit{0}\qubit{x}\qubit{y})$. These quantum functions are defined by Schemes I--III.
\end{lemma}

\begin{proof}
Recall the quantum functions $SWAP_{i,j}$ and $COPY_1$ respectively from Lemmas \ref{lemma:special} and \ref{COPY-function}. We first define $g_{OR}$ to be $SWAP_{1,3}\circ CSWAP\circ COPY_1$. From this definition, for any $x,y\in\{0,1\}$, it follows that $g_{OR}(\qubit{0}\qubit{x}\qubit{y}) = \measure{0}{x}\qubit{y}\qubit{x}\qubit{0}+ \measure{1}{x}\qubit{x}\qubit{y}\qubit{1}$. Thus, we obtain   $g_{OR}(\qubit{0}\qubit{0}\qubit{y}) = \qubit{y}\qubit{0}\qubit{0}$ and $g_{OR}(\qubit{0}\qubit{1}\qubit{y})=\qubit{1}\qubit{y}\qubit{1}$.

We next define $g_{AND}$ to be $SWAP_{1,3}\circ SWAP_{2,3}\circ CSWAP\circ COPY_1$. It then follows that $g_{AND}(\qubit{0}\qubit{x}\qubit{y}) = \measure{0}{x}\qubit{x}\qubit{y}\qubit{0} + \measure{1}{x}\qubit{y}\qubit{x}\qubit{1}$. From this equality, we obtain $g_{AND}(\qubit{0}\qubit{0}\qubit{y})= \qubit{0}\qubit{y}\qubit{0}$ and $g_{AND}(\qubit{0}\qubit{1}\qubit{y}) = \qubit{y}\qubit{1}\qubit{1}$.
\end{proof}


For later use, we define another quantum function, which splits the entire input qubits into two halves and then swaps them. We do not intend to include this quantum function to our system but it will be used to support the description of a new scheme given in Section \ref{sec:different-object}.
We remark that, to construct this quantum function, we need Schemes I--III together with the quantum recursion (Scheme T).  By recalling the left-half function $LH$ and the right-half function $RH$ from Section \ref{sec:numbers}, let us introduce $HalfSWAP$ as:

{\it
\begin{enumerate}\vs{-3}
  \setlength{\topsep}{-2mm}%
  \setlength{\itemsep}{1mm}%
  \setlength{\parskip}{0cm}%

\item[*)] $HalfSWAP(\qubit{\phi}) = \left\{ \begin{array}{ll}
 \qubit{\phi} & \hs{5}\mbox{if $\ell(\qubit{\phi})\leq 1$,} \\
  \sum_{s:|s|=LH(\qubit{\phi})} \measure{s}{\phi}\otimes \ket{s}
 & \hs{5}\mbox{otherwise.}
 \end{array}\right.$
\end{enumerate}
}

The \emph{inverse} of $HalfSWAP$, denoted $HalfSWAP^{-1}$, matches the quantum function obtained from $HalfSWAP$ by replacing $LH(\qubit{\phi})$ in its definition with $RH(\qubit{\phi})$. In the special case where $\ell(\qubit{\phi})$ is even, $HalfSWAP\circ HalfSWAP(\qubit{\phi})$ equals $\qubit{\phi}$ since $LH(\qubit{\phi})=RH(\qubit{\phi})$.

\begin{example}
Consider a quantum state $\qubit{\phi}= \sum_{u:|u|=3}\alpha_u \qubit{u}$, which is $\alpha_{000}\qubit{000} +  \alpha_{001}\qubit{001} + \alpha_{010}\qubit{010} + \alpha_{011}\qubit{011} + \cdots +  \alpha_{111}\qubit{111}$. Notice that $LH(\qubit{\phi})= \ceilings{\ell(\qubit{\phi})/2} = 2$ and $RH(\ell(\qubit{\phi}))=\floors{\ell(\qubit{\phi})/2} = 1$. Since $\measure{s}{\phi}= \sum_{u:|u|=3} \alpha_u \measure{s}{u}$, it then follows that $HalfSWAP(\qubit{\phi}) = \sum_{s:|s|=2}\sum_{u:|u|=3} \alpha_u \measure{s}{u}\otimes \ket{s}$, which equals $\alpha_{000}\qubit{000} +  \alpha_{001}\qubit{100} + \alpha_{010}\qubit{001} + \alpha_{011}\qubit{101} + \cdots +  \alpha_{111}\qubit{111}$.
\end{example}

\subsection{Binary Encoding of Various Types of Objects}\label{sec:different-object}

All quantum functions discussed in Section \ref{sec:elementary-scheme} take only ``single'' quantum states in $\HH_{\infty}$ as their inputs.
For practical applications of these quantum functions, we need to further deal with a problem that consists of various types of ``objects'', such as numbers, graphs, matrices, tape symbols, etc.
Since each QTM uses multiple tapes with their associated tape heads, these tapes hold possibly different qubits and their tape heads move separately. This is an advantage of the QTM model; however, each quantum function takes only one quantum state as its input.
For our purpose, it is imperative to set up an appropriate \emph{binary encoding} to transform those objects into a single quantum state with an appropriate use of designated ``separators''.
In the polynomial-time setting \cite{Yam20}, the scheme of quantum recursion (Scheme T) is powerful enough to handle several inputs altogether as a single encoded quantum state in a way similar to a multiple-tape QTM being simulated by a single-tape QTM with only polynomial overhead.
However, since we aim at capturing quantum polylogtime computability instead, we cannot take the same approach to cope with multiple inputs. We thus need to ponder how to circumvent this difficulty to expand the scope of our quantum functions. As a feasible solution to this problem, we introduce ``extra'' schemes that help us handle intended binary encodings.

In what follows, we attempt to design an encoding (or a translation) of various types of objects into binary strings of the same fixed length so that a series of these encoded objects forms a larger-dimensional quantum state. Each ``segment'' of such a quantum state representing one encoded object is referred to as a  \emph{section}, and this fixed-length encoding makes it possible to work with each section separately.

Let us describe our encoding scheme for eight symbols in $\{0, 1, 2, {\dashv}, B, H, S, T\}$, where $\dashv$, $B$, $H$, $S$, and $T$ respectively stand for an ``ending'', a ``blank'', a ``head'', a ``separator'', and a ``time''.
With the use of three bits, we take the following abbreviations:  $\hat{0}=000$, $\hat{1}=001$, $\hat{B}=010$, $\hat{\dashv}=011$, $\hat{2}=111$, $\hat{H}= 100$, $\hat{S}=110$, and $\hat{T}=101$. Given a binary string $s=s_1s_2\cdots s_k$, the notation $\tilde{s}^{(-)}$ denotes $\hat{s}_1\hat{s}_2\cdots \hat{s}_k$ and $\widetilde{s}$ denotes $\tilde{s}^{(-)}\hat{\dashv}$. For convenience, we also define $\widetilde{\lambda}$ to be $\hat{\dashv}$.
This makes us encode, for example, the number $8$ into $\widetilde{bin(8)} = \widetilde{001} = \hat{0}\hat{0}\hat{1}\hat{\dashv}$ while  $\widetilde{bin_3(2)}$ also equals $\widetilde{001}$.
Given an arbitrary quantum state $\qubit{\phi}$ in $\HH_{\infty}$, we finally define its encoding $\qubit{\tilde{\phi}}$ as $\sum_{s:|s|=\ell(\qubit{\phi})} \qubit{\tilde{s}}\measure{s}{\phi}$. Notice that $\ell(\qubit{\tilde{\phi}}) = 3\ell(\qubit{\phi})+3$.

To mark the end of a series of encoded objects, we use a designated separator, say, $r_0$ in $\{0,1\}^+$. We fix such $r_0$ in the following discussion.
We make each series of encoded objects have length proportional to the section size $|r_0|$, and thus any encoding $x$ satisfies $|x|=k|r_0|$ for an appropriate fixed number $k\in\nat^{+}$.
This helps us partition $x$ section-wise as $x_1x_2\cdots x_k$ with $|x_i|=|r_0|$ for all indices $i\in[k]$.
We also demand that no $x_i$ should match $r_0$. Here, we say that $x$ \emph{section-wise contains no $r_0$} if $x_i\neq r_0$ holds for all indices $i\in[k]$. Let $NON_{r_0} =\{x\in\{0,1\}^+\mid |x|\equiv0 \:(\mathrm{mod}\:|r_0|), \text{ $x$ section-wise contains no $r_0$ }\}$. Similarly, we set $NON_{r_0}(\qubit{\phi}) =\{x\in NON_{r_0} \mid \measure{xr_0}{\phi}\neq\bfzero \}$.
For convenience, when $r_0=\hat{2}$, we tend to omit $r_0$ from $NON_{r_0}$ and $NON_{r_0}(\qubit{\phi})$.

\begin{example}
Fix $r_0\in\{0,1\}^+$. Choose three strings $x_1,x_2,x_3\in\{0,1\}^{|r_0|}$ satisfying $x_i\neq r_0$ for all $i\in[3]$, and consider three quantum states $\qubit{\psi} = \qubit{x_1x_2x_3}\qubit{r_0}\qubit{\phi}$, $\qubit{\psi'}= \qubit{x_1x_2 r_0 x_3}\qubit{\phi}$, and $\qubit{\psi''}=\qubit{x_1r_0x_2x_3}\qubit{
\phi}$ for any $\qubit{\phi}\in \HH_{\infty}$.
We then obtain $NON_{r_0}(\qubit{\psi}) = \{x_1x_2x_3\}$, $NON_{r_0}(\qubit{\psi'}) = \{x_1x_2\}$, and $NON_{r_0}(\qubit{\psi''}) = \{x_1\}$. By contrast, when $\qubit{\psi}$ has the form $(\alpha\qubit{x_1}\qubit{r_0}+ \beta\qubit{x_2}\qubit{r_0})\qubit{\phi}$, $NON_{r_0}(\qubit{\psi})$ equals $\{x_1,x_2\}$.

As another example, if $x=0111$ and $r_0=\hat{2}$, then $\widetilde{x}$ equals $\widetilde{x}^{(-)}\hat{\dashv} = \hat{0}\hat{1}\hat{1}\hat{1}\hat{\dashv}$. Notice that  $|\widetilde{x}|\equiv 0\:(\mathrm{mod}\:|r_0|)$. Thus, we obtain  $NON_{r_0}(\qubit{\widetilde{x}}\qubit{r_0}\qubit{\phi}) = \{\widetilde{x}\}$.
\end{example}


We then need a quantum function that splits an input into sections and apply  ``section-wise'' two pre-determined quantum operations.
We actually introduce two slightly different \emph{code skipping schemes} described below. We do not unconditionally include them to $EQS$ but we use them in a certain restricted situation, which will be discussed later. Such a restriction is in fact necessary because these schemes are too powerful to  use for quantum polylogtime computability.

{\it
\begin{enumerate}\vs{-2}
  \setlength{\topsep}{-2mm}%
  \setlength{\itemsep}{1mm}%
  \setlength{\parskip}{0cm}%

\item[*)] The \emph{code skipping schemes}. From $g$, $h$  and $r_0\in\{0,1\}^+$, we define $CodeSKIP_{+}[r_0,g,h]$ and $CodeSKIP_{-}[r_0,g,h]$ as follows:

\begin{enumerate}\vs{1}
  \setlength{\topsep}{-2mm}%
  \setlength{\itemsep}{2mm}%
  \setlength{\parskip}{0cm}%

\item[{\rm (i)}] $CodeSKIP_{+}[r_0,g,h](\qubit{\phi}) = \left\{ \begin{array}{ll}
 \qubit{\phi} \hs{39}\mbox{if $NON_{r_0}(\qubit{\phi}) = \setempty$,} & \\
  \sum_{x\in NON_{r_0}(\qubit{\phi})} ( g( \qubit{xr_0} ) \otimes h( \measure{xr_0}{\phi} )) \hs{4}\mbox{otherwise.} &
 \end{array}\right.$

\item[{\rm (ii)}] $CodeSKIP_{-}[r_0,g,h](\qubit{\phi}) = \left\{ \begin{array}{ll}
 \qubit{\phi} \hs{39} \mbox{if $NON_{r_0}(\qubit{\phi}) = \setempty$,} & \\
  \sum_{x\in NON_{r_0}(\qubit{\phi})} ( g( \qubit{x} ) \otimes h( \measure{x}{\phi} )) \hs{10} \mbox{otherwise.} & \end{array}\right.$
\end{enumerate}
\end{enumerate}
}

The difference between $CodeSKIP_{+}[r_0,g,h]$ and $CodeSKIP_{-}[r_0,g,h]$ looks subtle but becomes clear in the following example. When $\qubit{\phi}=\qubit{x}\qubit{r_0}\qubit{y}$ with $|x|=|r_0|$ and $x\neq r_0$, it follows that $CodeSKIP_{+}[r_0,g,h](\qubit{\phi}) = g(\qubit{xr_0})\otimes h(\qubit{y})$ but $CodeSKIP_{-}[r_0,g,h](\qubit{\phi}) = g(\qubit{x})\otimes h(\qubit{r_0y})$. These schemes are not  interchangeable in most applications.

A quantum function $g$ is said to be \emph{query-independent} if, in the process of constructing $g$, on any input of the form $\qubit{xr_0}\qubit{\phi}$,  any quantum function that appears in this construction process does not directly access $\qubit{\phi}$ and thus it does not depend on $\qubit{\phi}$. This instantly implies that $g(\qubit{xr_0}\qubit{\phi}) = g(\qubit{xr_0})\otimes \qubit{\phi}$ for any $x$ and $\qubit{\phi}$.
Using this terminology, when $h=I$ in the code skipping schemes, $CodeSKIP_{+}[r_0,g,I]$ and $CodeSKIP_{-}[r_0,g,I]$ are query-independent.

Here, we present a few more examples of $CodeSKIP_{+}$.

\begin{example}
Let $g=ROT_{\pi/4}$ and $h=NOT$.
Let $r_0=0^5$ and let $\qubit{\phi} = \qubit{x_1x_2}\qubit{r_0}\qubit{x_3x_4}\qubit{r_0}\qubit{\psi}$ with binary strings $x_i=bin_5(i)$ for any $i\in[4]$ and a qustring $\qubit{\psi}$. Since  $NON_{r_0}(\qubit{\phi}) = \{x_1x_2\}$, we obtain $CodeSKIP_{+}[r_0,g,h](\qubit{\phi}) = ROT_{\pi/4}(\qubit{x_1x_2r_0})\otimes NOT(\qubit{x_3x_4r_0}\qubit{\psi})$.

Let $\qubit{\phi'} = \alpha\qubit{y_1r_0}\qubit{\psi_1} + \beta \qubit{y_2r_0}\qubit{\psi_2}$ with $|y_1|=|y_2|=|r_0|$ and $y_1,y_2\notin\{r_0\}$. In this case, $NON_{r_0}$ is the set $\{y_1,y_2\}$. It then follows that $CodeSKIP_{+}[r_0,g,h](\qubit{\psi'}) = \alpha g(\qubit{y_1r_0})\otimes h(\qubit{\psi_1}) + \beta g(\qubit{y_2r_0})\otimes h(\qubit{\psi_2})$.
\end{example}

We wish to recall the two useful quantum functions $REMOVE$ (removal) and $REP$ (replacement) introduced in \cite{Yam20}. We introduce the ``code-controlled'' versions of them. Let $r_0\in\{0,1\}^{+}$ be a separator.

{\it
\begin{enumerate}\vs{-3}
  \setlength{\topsep}{-2mm}%
  \setlength{\itemsep}{1mm}%
  \setlength{\parskip}{0cm}%

\item[] \hs{-3}{\rm (i)} $CodeREMOVE[r_0](\qubit{\phi}) = \left\{ \begin{array}{ll}
 \qubit{\phi} \hs{50}\mbox{if $NON_{r_0}(\qubit{\phi}) = \setempty$,} & \\
  \sum_{x\in NON_{r_0}(\qubit{\phi})} \sum_{a\in\{0,1\}} ( \measure{a}{x} \otimes \qubit{ar_0} \otimes \measure{xr_0}{\phi} ) \hs{1}\mbox{otherwise.} &
 \end{array}\right.$

\item[] \hs{-3}{\rm (ii)} $CodeREP[r_0](\qubit{\phi}) = \left\{ \begin{array}{ll}
 \qubit{\phi} \hs{54} \mbox{if $NON_{r_0}(\qubit{\phi}) = \setempty$,} & \\
  \sum_{x\in NON_{r_0}(\qubit{\phi})} \sum_{u:|u|=|x|-1} ( \measure{u}{x} \otimes \qubit{ur_0} \otimes \measure{xr_0}{\phi} ) \hs{1} \mbox{otherwise.} & \end{array}\right.$
\end{enumerate}
}

Notice that the quantum functions $CodeREMOVE[r_0]$ and $CodeREP[r_0]$ are query-independent.

We wish to include a simple example of $CodeREMOVE$ and $CodeREP$.

\begin{example}
Let $\qubit{\phi}=\alpha\qubit{x_1r_0}\qubit{\psi_1} + \beta\qubit{x_2r_0}\qubit{\psi_2} + \gamma\qubit{r_0}\qubit{\psi_3}$ with $\ell(\qubit{x_1r_0}\qubit{\psi_1}) = \ell(\qubit{x_2r_0}\qubit{\psi_2}) = \ell(\qubit{r_0}\qubit{\psi_3})$, $NON_{r_0}(\qubit{\psi}) =\{x_1,x_2\}$, $|\alpha|^2+|\beta|^2+|\gamma|^2=1$, and $\alpha\beta\gamma\neq0$. If $x_1=0y_1$ and $x_2=1y_2$ for two strings $y_1$ and $y_2$, then $CodeREMOVE[r_0](\qubit{\psi})$ equals $\alpha\qubit{y_10r_0}\qubit{\psi_1} + \beta\qubit{y_21r_0}\qubit{\psi_2} + \gamma\qubit{r_0}\qubit{\psi_3}$. If $x_1=z_10$ and $x_2=z_21$, then we obtain $CodeREP[r_0](\qubit{\phi}) = \alpha\qubit{0z_1r_0}\qubit{\psi_1} + \beta\qubit{1z_2r_0}\qubit{\psi_2} + \gamma\qubit{r_0}\qubit{\psi_3}$.
\end{example}

\subsection{Code-Controlled Fast Quantum Recursion Scheme}\label{sec:fast-scheme-IV}

Let us introduce a new scheme, called Scheme IV, which is a variant of the multi-qubit quantum recursion scheme (Scheme T) geared up with the code skipping schemes in Section \ref{sec:different-object}. Recall that, in Scheme T, we inductively discard $k$ qubits from an input quantum state $\qubit{\phi}$ until we consume all qubits except for the last (at most) $t$ qubits.
Unlike Scheme T, since our access to input qubits is quite limited, we need to split the whole input into two separate parts, which play quite different roles as we will see.

Before formally introducing the complexity class $EQS$, for each quantum function $f$ in $EQS_0$, we define its ``code-controlled'' version $f^*$ by setting $f^*(\qubit{xr_0}\qubit{\phi}) = \qubit{xr_0}\qubit{\phi}$ if $|x|\leq 2$, $\ell(\qubit{\phi})\leq 1$, or $|x|>|r_0|\ilog(\ell(\qubit{\phi}))$, and $f^*(\qubit{xr_0}\qubit{\phi}) = f(\qubit{xr_0})\otimes \qubit{\phi}$ otherwise, for any $x\in\{0,1\}^*$ and any $\qubit{\phi}\in\HH_{\infty}$. In the rest of this work, it is convenient to identify $f$ with $f^*$. It is important to note that $f^*$ does not access $\qubit{\phi}$ in $\qubit{xr_0}\qubit{\phi}$ by its definition.

In the fast quantum recursion, on the contrary, we discard a half of the second part of input quantum state $\qubit{xr_0}\otimes \qubit{\phi}$. By halving the input at each step, this recursive process quickly terminates.

\begin{definition}\label{def-EQS}
We introduce the following scheme.
\begin{enumerate}\vs{-3}
  \setlength{\topsep}{-2mm}%
  \setlength{\itemsep}{1mm}%
  \setlength{\parskip}{0cm}%

\item[IV.] The \emph{code-controlled fast quantum recursion scheme}.
Assume that we are given quantum functions $d$, $g$, $h$, a number $t\in\nat^{+}$, and a string  $r_0\in\{0,1\}^+$
(where $d$ is not defined using $MEAS[\cdot ]$  but $d$ and $h$ may be defined using $CodeSKIP_{+}[\cdot]$ and $CodeSKIP_{-}[\cdot]$).
We then define $F\equiv CFQRec_t[r_0,d,g,h| \PP_{|r_0|},\FF_{|r_0|}]$ for $\PP_{|r_0|}=\{p_u\}_{u\in\{0,1\}^{|r_0|}}$ with $p_u\in\{I,HalfSWAP\}$ and $\FF_{|r_0|}=\{f_u\}_{u\in\{0,1\}^{|r_0|}}$ with $f_u\in\{I,F\}$ as follows. For any $x\in\{0,1\}^*$ and any $\qubit{\phi}\in\HH_{\infty}$, let

\s
\n\hs{2}(i) $F(\qubit{xr_0} \qubit{\phi}) = g(\qubit{xr_0} \qubit{\phi})$
\hs{25} if $x=\lambda$, $\ell(\qubit{\phi})\leq t$, or $|x|>|r_0|k$, \vs{1} \\
\n\hs{1}(ii) $F(\qubit{xr_0} \qubit{\phi}) \\
\n\hs{10} =  \sum_{u:|u|=|r_0|} \sum_{v:|v|= \ell(\measure{u}{xr_0})} ( h(  \qubit{u} \qubit{v} ) \otimes p_u^{-1} ( \measure{v}{\zeta_{u,p_u,\phi}^{(x'r_0)}} ) )$ \hs{12}otherwise, \s\\
where $k=\ilog(\ell(\qubit{\phi}))$,
$x\in NON_{r_0}$, $\qubit{\zeta_{u,p_u,\phi}^{(x'r_0)}} = \sum_{s:|s|= m_u(\qubit{\phi})} ( f_u(  \measure{u}{x'r_0}  \otimes  \ket{s} ) \otimes \measure{s}{\psi_{p_u,\phi}} )$,  $d(\qubit{xr_0}) = \qubit{x'r_0}$ with $x'\in NON_{r_0}$, $\qubit{\psi_{p_u,\phi}} = p_{u}(\qubit{\phi})$. Moreover, $m_u(\cdot)$ is determined to be $LH$ if $p_u=I$ and $RH$ if $p_u=HalfSWAP$. Notice that $u\neq r_0$ follows from $x'\in NON_{r_0}$.
In the other case where an input, say, $\qubit{y}$ to $F$  satisfies $NON_{r_0}(\qubit{y})=\setempty$, we automatically set $F(\qubit{y}) =\qubit{y}$.
For readability, the prefix term ``code-controlled'' is often dropped and $\qubit{\zeta_{u,p_u,\phi}^{(x'r_0)}}$ is expressed briefly as $\qubit{\zeta_{u,p_u,\phi}}$ as long as $x'r_0$ is clear from the context.
\end{enumerate}
\end{definition}

The quantum functions $d,g,h$ in the above definition are called \emph{ground (quantum) functions} of $F$. If $g$ is query-independent, Scheme IV is said to be \emph{query-independent}.

Although the description of the conditions (i)--(ii) in Scheme IV refers to a classical string $xr_0$, when we plug in any quantum state $\qubit{\psi}=\sum_{x}\alpha_x\qubit{xr_0}\otimes \qubit{\phi_x}$ to $F$, we  obtain the result $\sum_{x}\alpha_x F(\qubit{xr_0}\qubit{\phi_x})$.

Since $p_u\in\{I,HalfSWAP\}$ and $|x'|=|x|$, it follows that $\ell(\qubit{\psi_{u,p_u,\phi}}) = \ell(\qubit{\phi})$,  $\ell(\qubit{\zeta_{u,p_u,\phi}^{(x'r_0)}}) = \ell(\measure{u}{x'r_0}) + \ell(\qubit{\phi})$, and $\ell(\measure{v}{\zeta_{u,p_u,\phi}^{(x'r_0)}}) = \ell(\qubit{\phi})$.
In Item (ii), since the size $|s|$ is exactly $LH(\qubit{\phi})$ (or $RH(\qubit{\phi})$), $f_u$ takes inputs of length $(\ell(\qubit{x'r_0})-|r_0|)+\ceilings{\ell(\qubit{\phi})/2}$ (or $(\ell(\qubit{x'r_0})-|r_0|)+\floors{\ell(\qubit{\phi})/2}$). Hence,
within $\ceilings{\log{\ell(\qubit{\phi})}}$ recursive steps, the whole  process terminates.

As a special case of Scheme IV, when $r_0=1$, Item (ii) has the following simple form:

\s
\n\hs{10} (ii$'$) \hs{1} $F(\qubit{0^n1} \qubit{\phi}) =
\sum_{v:|v|=n}  ( h( \qubit{0} \otimes \qubit{v} ) \otimes p_0^{-1} ( \measure{v}{\zeta_{0,p_0,\phi}^{(0^n1)}} ))$,
\s

\n where
$\qubit{\zeta_{0,p_0,\phi}^{(0^n1)}} =
\sum_{s:|s|= m(\qubit{\phi})} ( f_0(  \qubit{0^{n-1}1}  \otimes  \ket{s} )
\otimes \measure{s}{\psi_{p_0,\phi}} )$.

At this moment, it is worth remarking the usage of the length function $\ell(\cdot)$. In Scheme IV, the input quantum state $\qubit{xr_0}\qubit{\phi}$ is reduced to $\measure{u}{x'r_0}\qubit{s}$ (with $|s|=m_u(\qubit{\phi})$ and $|u|=|r_0|$) so that we can inductively apply $F$ (when $f_u=F$) to it. The length $\ell(\qubit{\phi})$ then becomes $\ell(\qubit{s})$, and thus the conditional execution of (i) depends on the value $\ell(\qubit{s})$.

To understand Scheme IV better, we provide in Example \ref{example-scheme} a simple example of how to calculate the quantum function $F\equiv CFQRec_t[r_0,d,g,h | \PP_{|r_0|},\FF_{|r_0|}]$ step by step.

\begin{example}\label{example-scheme}In this example, we wish to show
how to calculate the quantum function $F\equiv CFQRec_t[r_0,d,g,h | \PP_{|r_0|},\FF_{|r_0|}]$ defined with the parameters $t=1$, $r_0=1$, $d=I$,  and $g\equiv I$. Furthermore, $h$ is defined as  $h(\qubit{01})=\qubit{11}$ and $h(\qubit{010^m1})=\qubit{10^{m+1}1}$ for any $m\in\nat$.
We also set $p_0=HalfSWAP$, $p_1=I$, $f_0=F$, and $f_1=I$.
Let $x\in\{0\}^*$ and $\qubit{\phi}\in\HH_{\infty}$.
In what follows, we will calculate $F(\qubit{x1}\qubit{\phi})$ in a  ``bottom-up'' fashion.

(1) We start with the base case. If $x=\lambda$, then $F(\qubit{1}\qubit{\phi})=g(\qubit{1}\qubit{\phi}) = \qubit{1}\qubit{\phi}$. If $x=0^m1$ with $m\geq1$ and $\ell(\qubit{\phi})\leq1$, then $F(\qubit{0^m1}\qubit{\phi}) = g(\qubit{0^m1}\qubit{\phi}) = \qubit{0^{m}1}\qubit{\phi}$.

(2) Hereafter, we assume that $x\neq\lambda$ and $\ell(\qubit{\phi})\geq2$. For simplicity, let  $\qubit{\phi}$ denote $\qubit{u}$ for a certain string $u\in\{0,1\}^*$.
When $u=u_1u_2\cdots u_n$ and $1\leq i<j\leq n$, we use the succinct notation $u_{[i,j]}$ to express the string $u_{i}u_{i+1}\cdots u_{j}$.

(a) We first calculate $F(\qubit{01}\qubit{u})$ for $u=u_1u_2u_3u_4$.
Since $\qubit{u}=\qubit{u_{[1,2]}}\qubit{u_{[3,4]}}$, we obtain  $p_0(\qubit{u})=\qubit{u_{[3,4]}}\qubit{u_{[1,2]}}$.
We also obtain $\qubit{\zeta_{0,p_0,u}}= F(\qubit{1}\qubit{u_{[3,4]}})\otimes \qubit{u_{[1,2]}} = \qubit{1} \otimes \qubit{u_{[3,4]}}\qubit{u_{[1,2]}}$ by (1), and thus  $F(\qubit{01}\qubit{u})$ equals $h(\qubit{0}\otimes \qubit{1} ) \otimes p_0^{-1}( \qubit{u_{[3,4]}}\qubit{u_{[1,2]}} )$, which is $\qubit{11}\qubit{u_{[1,2]}}\qubit{u_{3,4]}} = \qubit{11}\qubit{u}$.

(b) Next, we calculate $F(\qubit{001}\qubit{u'})$ for $u'=u_1u_2\cdots u_8$. Note that $\qubit{\zeta_{0,p_0,u'}} = F(\qubit{01}\qubit{u'_{[5,8]}})\otimes \qubit{u'_{[1,4]}} = \qubit{11} \otimes \qubit{u'_{[5,8]}}\qubit{u'_{[1,4]}}$ by (a).
From this, we obtain $F(\qubit{001}\qubit{u'}) = h(\qubit{0}\otimes \qubit{11}) \otimes p_0^{-1}(\qubit{u'_{[5,8]}}\qubit{u'_{[1,4]}}) = \qubit{101}\qubit{u'}$.
\end{example}

\subsection{Power of Scheme IV}

In what follows, we intend to show the usefulness of Scheme IV by applying it to construct a quantum function, which calculates the logarithmic value of (part of) input size.
Formally, for any $\qubit{\phi}\in\HH_{\infty}$ and $m\in\nat^{+}$, we define $SIZE_1$ as $SIZE_1(\qubit{0^m1}\qubit{\phi}) = \qubit{0^k1}\qubit{0^{m-k-1}1} \qubit{\phi}$ if $\ell(\qubit{\phi})\leq 2^{m-1}$, where $k=\ilog(\ell(\qubit{\phi}))$, and $SIZE_1(\qubit{0^m1}\qubit{\phi}) = \qubit{0^m1}\qubit{\phi}$ otherwise.
More generally, for any $r_0\notin\{0\}^*$ with $|r_0|\geq1$,
we define $SIZE_{r_0}$ as  $SIZE_{r_0}(\qubit{0^{m|r_0|}r_0}\qubit{\phi}) = \qubit{0^{k|r_0|}r_0} \qubit{0^{(m-k-1)|r_0|}r_0}\qubit{\phi}$.
The choice of $0^{m|r_0|}$ here is only for simplicity.

For brevity, we intend to use the notation $EQS_0+IV$ to express the set of quantum functions constructed by applying Schemes I--IV.

\begin{lemma}\label{SIZE-function}
Let $r_0\notin \{0\}^*$ with $|r_0|\geq1$. The above quantum function $SIZE_{r_0}$ can be definable within $EQS_0+IV$.
\end{lemma}

\begin{proof}
We first prove the lemma for the simple case of $r_0=1$.
Given $m\in\nat^{+}$ and $\qubit{\phi}\in\HH_{\infty}$, let  $\qubit{\xi}=\qubit{0^{m}1}\qubit{\phi}$.
Recall the quantum function $LengthQ_1$ from Lemma \ref{lemma:special} and, for simplicity, write $g$ for $LengthQ_1$.
It then follows that $g(\qubit{1}) = \qubit{1}$ and $g(\qubit{0^{m}1}) = \qubit{10^{m-1}1}$ if $m\geq1$.
We denote by $F$ the quantum function $CFQRec_1[1,I,g,I | \{p_0,p_1\},\{f_0,f_1\}]$ with the parameters $f_{0}=F$, $f_{1}=I$, and $p_0=p_1=I$.

If either $m=0$ or $\ell(\qubit{\phi}) \leq 1$, then $F(\qubit{1}\qubit{\phi})$ equals $g(\qubit{1})\otimes \qubit{\phi} = \qubit{1}\qubit{\phi}$. Hereafter, we assume that $m\geq1$ and  $\ell(\qubit{\phi})\geq2$. By induction hypothesis, we obtain  $F(\qubit{0^{m-1}1}\otimes \qubit{s}) = \qubit{0^{k-1}1}\qubit{0^{m-k-1}1}\qubit{\phi}$ if $\ell(\qubit{s})\leq 2^{m-2}$ and $k-1=\ilog(\ell(\qubit{s}))$.
Starting with $\qubit{\xi} = \qubit{0^m1}\qubit{\phi}$,  $F(\qubit{\xi})$  equals
$\qubit{0}\otimes \sum_{s:|s|=LH(\qubit{\phi})} F( \qubit{0^{m-1}1} \otimes \qubit{s} ) \otimes  \measure{s}{\phi} = \qubit{0}\otimes \qubit{0^{k-1}1}\qubit{0^{m-k-1}1}\qubit{\phi} = \qubit{0^k1}\qubit{0^{m-k-1}1}\qubit{\phi}$.

The general case of $r_0\neq1$ is similarly handled.
The desired quantum function $SIZE_1$ is therefore set to be $F$.
\end{proof}


Given an input of the form $\qubit{xr_0}\qubit{\phi}$, it is possible to apply any quantum function $g$ in $EQS_0$ to the first segment  $\qubit{xr_0}$ with keeping the second segment $\qubit{\phi}$ intact. This can be done by the use of Scheme IV as follows.

\begin{lemma}\label{EQS-zero-property}
For any string $r_0\in\{0,1\}^+$ and any quantum function $g\in \widehat{EQS}_0$ satisfying $g(\qubit{r_0}\otimes \qubit{\phi}) = g(\qubit{r_0})\otimes \qubit{\phi}$,  there exists another quantum function $F$ definable within $EQS_0+IV$ such that $F(\qubit{xr_0}\qubit{\phi}) = g(\qubit{xr_0})\otimes \qubit{\phi}$ for any $x\in NON_{r_0}$ and any $\qubit{\phi}\in\HH_{\infty}$.
\end{lemma}

\begin{proof}
Let us recall the quantum function $Skip_k[g]$ from Lemma \ref{Skip-function}.
For a given $g\in \widehat{EQS}_0$, we set $h\equiv g\circ Skip_{|r_0|}[g^{-1}]$.
Notice that $h$ belongs to $\widehat{EQS}_0$ because  $g^{-1}$ exists within  $\widehat{EQS}_0$ (as shown in Lemma \ref{inverse-function}).
The desired quantum function $F$ is defined by Scheme IV as $F\equiv CFQRec_{1}[r_0,I,g,h|\PP_{|r_0|},\FF_{|r_0|}]$, where $\PP_{|r_0|} =\{p_u\}_{u\in\{0,1\}^{|r_0|}}$ with $p_u=I$ for all $u$'s and $\FF_{|r_0|} =\{f_u\}_{u\in\{0,1\}^{|r_0|}}$ with $f_u = F$ for all $u$'s.
As a special case, we then obtain $F(\qubit{r_0}\qubit{\phi}) = g(\qubit{r_0}\qubit{\phi}) =  g(\qubit{r_0})\otimes \qubit{\phi}$.

Assume by induction hypothesis that $F(\qubit{xr_0}\qubit{\phi}) = g(\qubit{xr_0})\otimes \qubit{\phi}$.
Let us consider $F(\qubit{axr_0}\qubit{\phi})$ for an arbitrary string $a\in\{0,1\}^{|r_0|}\cap NON_{r_0}$. It then follows that $\qubit{\zeta_{a,I,\phi}^{(xr_0)}} = \sum_{s:|s|=LH(\qubit{\phi})} ( F(\qubit{xr_0}\otimes \qubit{s} ) \otimes \measure{s}{\phi} ) = \sum_{s:|s|=LH(\qubit{\phi})} ( g(\qubit{xr_0})\otimes \ket{s}\measure{s}{\phi} ) = g(\qubit{xr_0})\otimes \qubit{\phi}$.
We therefore conclude that $F(\qubit{axr_0}\qubit{\phi}) = \sum_{v:|v|=|xr_0|} ( h(\qubit{a}\qubit{v}) \otimes \measure{v}{\zeta_{a,I,\phi}^{(xr_0)}}) = \sum_{v:|v|=|xr_0|} ( h(\qubit{a}\qubit{v}) \otimes \measure{v}{\psi_{g,xr_0}}\otimes \qubit{\phi} ) = h(\qubit{a}\qubit{\psi_{g,xr_0}} )\otimes \qubit{\phi} = h(\qubit{a}\otimes g(\qubit{xr_0}))\otimes \qubit{\phi}$
since $|v|=|xr_0|$, where $\qubit{\psi_{g,xr_0}} = g(\qubit{xr_0})$.
Because $h(\qubit{a}\otimes g(\qubit{xr_0}))$ equals $g(\qubit{a}\otimes g^{-1}(g(\qubit{xr_0})) = g(\qubit{axr_0})$, it follows that $F(\qubit{0xr_0}\qubit{\phi}) = g(\qubit{axr_0})\otimes \qubit{\phi}$, as requested.
\end{proof}


As shown in Proposition \ref{limited-recursion}, Scheme IV turns out to be so powerful that it generates quantum functions, which can modify the first segment, $\qubit{xr_0}$, $\ilog(\ell(\qubit{\phi}))$ times for any given input of the form $\qubit{xr_0}\qubit{\phi}$.


Let $h$ be any quantum function defined by Schemes I--III. Consider a quantum function $\hat{h}$ defined inductively as \begin{equation*}
(*) \hs{5} \hat{h}(\qubit{r_0}) = \qubit{r_0} \hs{3}\text{   and   }\hs{3}  \hat{h}(\qubit{xr_0}) =  \sum_{u:|u|=|r_0|}h( \ket{u} \otimes \hat{h}( \measure{u}{xr_0}))
\end{equation*}
for any $x\in NON_{r_0}$ with $x\neq\lambda$.
This recursive construction scheme (*) looks similar to Scheme T but it is not supported in our system $EQS_0+IV$.
Nevertheless, as shown in Proposition \ref{limited-recursion}, whenever the length of an input qustring $\qubit{xr_0}$ is ``short'' enough compared to another supplemental input qustring $\qubit{\phi}$, it may be possible to ``realize'' this scheme within $EQS_0+IV$.

\begin{example}
As a concrete example of the above function $\hat{h}$, we consider  $CodeSKIP_{+}[r_0,g,I]$ (as well as $CodeSKIP_{-}[r_0,g,I]$) for a norm-preserving quantum function $g$. To see this, we set $h$ to be $g\circ Skip_{|r_0|}[g^{-1}]$ and define $\hat{h}$ from $h$ by applying the above scheme (*).
Here, we wish to claim that this quantum function $\hat{h}$ coincides with $CodeSKIP_{+}[r_0,g,I]$.
Initially, we obtain $CodeSKIP_{+}[r_0,g,I](\qubit{r_0}) = \qubit{r_0} = \hat{h}(\qubit{r_0})$ by (*). For any two strings $a,x\in NON_{r_0}\cap \{0,1\}^{+}$ with $|a|=|r_0|$, it follows from (*) that $\hat{h}(\qubit{axr_0}) = \sum_{u:|u|=|r_0|} h(\qubit{u}\otimes \hat{h}(\measure{u}{axr_0})) = h(\qubit{a}\otimes CodeSKIP_{+}[r_0,g,I](\qubit{xr_0})) = h(\qubit{a}\otimes g(\qubit{xr_0})) = g\circ Skip_{|r_0|}[g^{-1}](\qubit{a}\otimes g(\qubit{xr_0})) = g(\qubit{axr_0})$.
\end{example}

The quantum function $\hat{h}$ given by the recursive scheme (*)  may not be constructed by the only use of Schemes I--IV since the $k$-qubit
quantum recursion scheme is required.
For relatively ``short'' inputs, however, it is possible to compute the value of $\hat{h}$ within the existing system $EQS_0+IV$.

\begin{proposition}\label{limited-recursion}
For a quantum function $h$ in $EQS_0$, let $\hat{h}$ satisfy the conditions of the aforementioned recursive scheme (*). There exists a quantum function $F$ definable  within $EQS_0+IV$ such that, for any $(x,\qubit{\phi})$ with $x\in NON_{r_0}$ and $\qubit{\phi}\in\HH_{\infty}$, if $|x|\leq |r_0|\log{\ell(\qubit{\phi})}$, then $F(\qubit{xr_0}\qubit{\phi}) = \hat{h}( \qubit{xr_0}) \otimes \qubit{\phi}$ holds. However, there is no guarantee that $F(\qubit{xr_0}\qubit{\phi})$ matches $\hat{h}(\qubit{xr_0})\otimes \qubit{\phi}$ when $|x|>|r_0|\log\ell(\qubit{\phi})$.
\end{proposition}

\begin{proof}
Assume that $|x|\leq |r_0|\log{\ell(\qubit{\phi})}$. Consider the quantum function $F \equiv  CFQRec_{|r_0|-1}[r_0,I,g,h |\PP_{|r_0|}, \FF_{|r_0|}]$, where $\PP_{|r_0|}=\{p_u\}_{u\in\{0,1\}^{|r_0|}}$ and $\FF_{|r_0|} =\{f_u\}_{u\in\{0,1\}^{|r_0|}}$ with $p_u\equiv I$ and $f_u\equiv F$ for all  $u\in\{0,1\}^{|r_0|}$.
We verify the proposition by induction on the number of applications of $F$.
In the base case, we obtain $F(\qubit{r_0}\qubit{\phi}) = \qubit{r_0} \qubit{\phi} = \hat{h}(\qubit{r_0})\otimes \qubit{\phi}$ since $\hat{h}(\qubit{r_0}) = \qubit{r_0}$ by (*).
Next, we consider any string $ax$ with $a\in\{0,1\}^{|r_0|}\cap NON_{r_0}$ and $x\in NON_{r_0}$. We obtain $F(\qubit{xr_0}\qubit{\phi}) = \hat{h}(\qubit{xr_0})\otimes \qubit{\phi}$ by induction hypothesis.
It then follows that $F(\qubit{axr_0}\qubit{\phi}) =  \sum_{v:|v|=|xr_0|}  ( h( \qubit{a} \otimes \qubit{v}) \otimes \measure{v}{\zeta_{a,I,\phi}^{(xr_0)}} )$, where $\qubit{\zeta_{a,I,\phi}^{(xr_0)}} = \sum_{s:|s|=LH(\qubit{\phi})} (F(\qubit{xr_0}\otimes \qubit{s}) \otimes \measure{s}{\phi})$.
Since $F(\qubit{xr_0}\qubit{\phi}) = \hat{h}(\qubit{xr_0})\otimes \qubit{\phi}$, $\qubit{\zeta_{a,I,\phi}^{(xr_0)}}$ equals $\sum_{s:|s|=LH(\qubit{\phi})} (\hat{h}(\qubit{xr_0})\otimes \qubit{s} \otimes  \measure{s}{\phi}) = \hat{h}(\qubit{xr_0}) \otimes \qubit{\phi}$.
We write $\qubit{\psi_{\hat{h},xr_0}}$ for $\hat{h}(\qubit{xr_0})$. We then obtain $F(\qubit{axr_0}\qubit{\phi}) = \sum_{v:|v|=|xr_0|} ( h(\qubit{a}\otimes \qubit{v})  \otimes \measure{v}{\psi_{\hat{h},xr_0}} ) \otimes \qubit{\phi}$.
This implies that $F(\qubit{axr_0}\qubit{\phi}) = h(\qubit{a}\otimes \hat{h}(\qubit{xr_0})) \otimes \qubit{\phi} = \hat{h}(\qubit{axr_0})\otimes \qubit{\phi}$ by the definition of $\hat{h}$.
\end{proof}

Since $CodeSKIP_{+}[\cdot]$ can be realized by the recursive construction scheme (*), Proposition \ref{limited-recursion} allows us to use $CodeSKIP_{+}[\cdot]$ freely as if it is a quantum function in $EQS_0+IV$.

\begin{corollary}
There exists a quantum function $F$ definable within $EQS_0+IV$ such that $F(\qubit{xr_0}\qubit{\phi}) = CodeSKIP_{+}[r_0,g,I](\qubit{xr_0}\qubit{\phi})$ if $|x|\leq |r_0|\log{\ell(\qubit{\phi})}$.
\end{corollary}

\subsection{Elementary Quantum Schemes}\label{sec:LCompo-EQS}

Formally, let us introduce  $\widehat{EQS}$ and  $EQS$. We have already explained Schemes I--IV. Now, we wish to add the final piece of construction scheme, called Scheme V, which intuitively supports successive $\ilog(\ell(\qubit{\phi}))$ applications of $Compo[g,g]$ for a given quantum function $g$ taking $\qubit{xr_0}\qubit{\phi}$ as an input.

\begin{definition}\label{def-LCompo}
We introduce the following scheme.
\begin{enumerate}\vs{-3}
  \setlength{\topsep}{-2mm}%
  \setlength{\itemsep}{1mm}%
  \setlength{\parskip}{0cm}%

\item[V.] The \emph{logarithmically-many composition scheme}.
From $g$, we define $LCompo[g]$ as follows:

\s
\n\hs{2}(i) $LCompo[g](\qubit{xr_0} \qubit{\phi}) = \qubit{xr_0} \otimes \qubit{\phi}$
\hs{8} if $x=\lambda$, $\ell(\qubit{\phi})\leq 1$, or $|x|>|r_0|k$, \vs{1} \\
\n\hs{1}(ii) $LCompo[g](\qubit{xr_0} \qubit{\phi}) =  g^{k}(\qubit{xr_0}\qubit{\phi})$  \hs{7}otherwise, \s\\
where $x\in NON_{r_0}$ and $k=\ilog(\ell(\qubit{\phi}))$.
\end{enumerate}
\end{definition}

When $g$ is query-independent, we say that Scheme V is \emph{query-independent}. We remark that query-independent Scheme V is actually redundant because Proposition \ref{limited-recursion} helps us  realize $LCompo[g]$ by an application of Scheme IV.

\begin{definition}\label{EQS-definition}
The class $EQS$ is the smallest set of (code-controlled) quantum functions that contains the quantum functions of Scheme I  and is closed under Schemes II--V. Similarly, $\widehat{EQS}$ is defined with no use of Item 6) of Scheme I.
\end{definition}

Note that any quantum function $F$ in $EQS$ is constructed by sequential applications of Schemes I--V. Such a finite series is referred to as a \emph{construction history} of $F$. The length of this construction history serves as a \emph{descriptive complexity measure} of $F$. Refer to \cite{Yam20} for more discussions.

\section{Quantum Functions Definable within EQS}\label{sec:definable-EQS}

In Section \ref{sec:LCompo-EQS}, we have introduced the system  $\widehat{EQS}$ as well as $EQS$. In this section, we will study the basic properties of all quantum functions in $\widehat{EQS}$ and $EQS$ by introducing several useful quantum functions and extra schemes, which are definable within $EQS$. One of the most important properties we can show is an implementation of a simple ``binary search strategy'' in $EQS$.

\subsection{Basic Properties of $\widehat{EQS}$}

The only difference between  $\widehat{EQS}$ and $EQS$ is the free use of item 6) (quantum measurement) of Scheme I. Since $\widehat{EQS}$ does not involve quantum measurement, we naturally expect that  $\widehat{EQS}$ enjoys the unitary nature of quantum computation described in the following three  lemmas, Lemmas \ref{preserving-property}--\ref{inverse-function}.

\begin{lemma}\label{preserving-property}
Any quantum function in $\widehat{EQS}$ is dimension-preserving and norm-preserving.
\end{lemma}

\begin{proof}
Let us check Schemes I--V separately to verify the lemma.
Note that Items 1)--5) of Scheme I are clearly dimension-preserving and norm-preserving. The lemma was already shown for Schemes II--III in \cite{Yam20}.
Therefore, it suffices to check Schemes IV and V. Let $F\equiv CFQRec_t[r_0,g,h| \PP_{|r_0|},\FF_{|r_0|}]$ be any quantum function defined by Scheme IV. It is easy to verify that $HalfSWAP$ is dimension-preserving and norm-preserving.
Thus, so are all quantum functions in $\PP_{|r_0|}$. By induction hypothesis, we assume that $g$ and $h$ are dimension-preserving and norm-preserving. In what follows, we argue by way of induction on the input length of $F$ that $F$ is dimension-preserving and norm-preserving. If either $x=\lambda$ or $\ell(\qubit{\phi})\leq t$, then we obtain $\ell(F(\qubit{xr_0}\qubit{\phi})) = \ell(g(\qubit{xr_0}\qubit{\phi})) = \ell(\qubit{xr_0}\qubit{\phi})$ and $\|F(\qubit{xr_0}\qubit{\phi})\| = \|g(\qubit{xr_0}\qubit{\phi})\| = \|\qubit{xr_0}\qubit{\phi}\|$.

Let us consider the case where $x\neq\lambda$ and $\ell(\qubit{\phi})>t$.
We then obtain $\ell(h(\qubit{u}\qubit{v}) \otimes p_u^{-1}(\measure{v}{\zeta_{u,p_u,\phi}})) = \ell(\qubit{u})+\ell(\ket{v}\measure{v}{\zeta_{u,p_u,\phi}}) = \ell(\qubit{xr_0})-|r_0| + \ell(\qubit{\zeta_{u,p_u,\phi}})$.
By induction hypothesis, we obtain $\ell(F(\measure{u}{xr_0}\ket{s})) = \ell(\measure{u}{xr_0}\ket{s})$.
This implies that $\ell(\qubit{\zeta_{u,p_u,\phi}}) = \ell(\sum_{s}F(\measure{u}{xr_0}\ket{s})\otimes \measure{s}{\psi_{p_u,\phi}}) =\ell(\measure{u}{xr_0}\qubit{\phi})$.
It then follows that $\ell(F(\qubit{xr_0}\qubit{\phi})) = \ell( \sum_{u}\sum_{v} h(\qubit{u}\qubit{v}) \otimes p_u^{-1}(\measure{v}{\zeta_{u,p_u,\phi}})) = \ell(\qubit{xr_0}\qubit{\phi})$. The property of norm-preserving is similarly proven.

For Scheme V, it is obvious that, if $g$ is dimension-preserving and norm-preserving, then so is $g^k$ for any number $k\in\nat^{+}$. Thus, $LCompo[g]$ satisfies the lemma.
\end{proof}

We further discuss two useful construction schemes, which are definable within $\widehat{EQS}$.

\begin{lemma}\label{various-schemes}
Let $k\in\nat^{+}$ and let $\GG_k=\{g_u\}_{u\in\{0,1\}^k}$ be a series of $\widehat{EQS}$-functions. The following quantum functions all belong to $\widehat{EQS}$. The lemma also holds even if  $\widehat{EQS}$ is replaced by  $EQS$. Let $\qubit{\phi}$ be any quantum state in $\HH_{\infty}$
and let $x\in\{0,1\}^k$.

\begin{enumerate}\vs{-2}
  \setlength{\topsep}{-2mm}%
  \setlength{\itemsep}{1mm}%
  \setlength{\parskip}{0cm}%

\item $Compo[\GG_k](\qubit{\phi}) = g_{s_1}\circ g_{s_2}\circ \cdots \circ  g_{s_{2^k}}(\qubit{\phi})$. (multiple composition)

\item $Branch_k[\GG_k](\qubit{\phi})= \qubit{\phi}$  if $\ell(\qubit{\phi})<  k$ and $Branch_k[\GG_k](\qubit{\phi})= \sum_{s:|s|=k}\qubit{s}\otimes g_s(\measure{s}{\phi})$ otherwise.
\end{enumerate}
\end{lemma}

\begin{proof}
(1)--(2) The above schemes were shown in \cite[Lemma 3.6]{Yam20} to be valid for $\squareqp$ and $\hatsquareqp$, and their argument also work for $EQS$ and $\widehat{EQS}$.
\end{proof}

As a quick example of $Branch_k[\GG_k]$, let us recall the quantum function $Skip_k[g]$ from Lemma \ref{Skip-function}. It is obvious that $Skip_k[g]$ is simply  defined to be  $Branch_k[\{g_u\}_{u\in\{0,1\}^k}]$ with $g_u=g$ for all $u\in\{0,1\}^k$.


For any given quantum function, when there exists another quantum function $g$ satisfying $f\circ g(\qubit{\phi}) = g\circ f(\qubit{\phi}) = \qubit{\phi}$ for any $\qubit{\phi}\in\HH_{\infty}$, we express this $g$ as $f^{-1}$ (or sometimes $f^{\dagger}$) and
call it the \emph{inverse (function)} of $f$.

\begin{lemma}\label{inverse-function}
For any quantum function $g$ in $\widehat{EQS}$, its inverse function $g^{-1}$ exists in $\widehat{EQS}$. Moreover, if $g$ is in $\widehat{EQS}_0$, then $g^{-1}$ is also in $\widehat{EQS}_0$
\end{lemma}

\begin{proof}
For Scheme I, it is obvious that $PHASE_{\theta}^{-1}=PHASE_{-\theta}$, $ROT_{\theta}^{-1} = ROT_{-\theta}$, $NOT^{-1}=NOT$, and $SWAP^{-1}=SWAP$. For Scheme II, $Compo[g,h]^{-1}$ equals $Compo[h^{-1},g^{-1}]$.
For Scheme III, it suffices to set $Branch[g,h]^{-1}$ to be $Branch[g^{-1},h^{-1}]$.
For Scheme IV, let $F\equiv CFQRec_t[r_0,d,g,h | \PP_{|r_0|}, \FF_{|r_0|}]$.
If $x=\lambda$, $\ell(\qubit{\phi})\leq t$, or $|x|>|r_0|k$ with $k=\ilog(\ell(\qubit{\phi}))$, then we set $F^{-1}\equiv g^{-1}$. In this case, we obtain $F^{-1}\circ F(\qubit{xr_0}\qubit{\phi}) = F^{-1}(g(\qubit{xr_0}\qubit{\phi})) = g^{-1}\circ g(\qubit{xr_0}\qubit{\phi})=\qubit{xr_0}\qubit{\phi}$. Assuming otherwise, we define $F^{-1}$ to be the consecutive applications of two quantum functions: $h^{-1}$ and $G\equiv CFQRec_t[r_0,d^{-1},g^{-1},I |\PP_{|r_0|}, \FF^{-1}_{|r_0|}]$, where
$\FF^{-1}_{|r_0|} = \{f_u^{-1}\}_{u\in\{0,1\}^{|r_0|}}$.
It is important to note that we do not use $\PP_{|r_0|}^{-1} = \{p_u^{-1}\}_{u\in\{0,1\}^{|r_0|}}$ in place of $\PP_{|r_0|}$ in the definition of $G$.
Let us show that $F^{-1}\circ F(\qubit{xr_0}\qubit{\phi}) = \qubit{xr_0}\qubit{\phi}$.
Assume that $F(\qubit{xr_0}\qubit{\phi})$ has the form
$\sum_{u:|u|=|r_0|} \sum_{v:|v|= \ell(\measure{u}{xr_0})} ( h(  \qubit{u} \qubit{v} ) \otimes p_u^{-1} ( \measure{v}{\zeta_{u,p_u,\phi}^{(x'r_0)}} ) )$.
We first apply $h^{-1}$ to the first $|xr_0|$ qubits and then obtain
$\sum_{u:|u|=|r_0|} \sum_{v:|v|= \ell(\measure{u}{xr_0})}   \qubit{u} \qubit{v}  \otimes p_u^{-1} ( \measure{v}{\zeta_{u,p_u,\phi}^{(x'r_0)}} ) )$. To this quantum state, we further apply $G$.
By an application of $p_u$ to $p_u^{-1}(\measure{v}{\zeta^{(x'r_0)}_{u,p_u,\phi}})$, we obtain
$\sum_{u:|u|=|r_0|} \sum_{v:|v|= \ell(\measure{u}{xr_0})} \qubit{u} \qubit{\zeta_{u,p_u,\phi}^{(x'r_0)}}$.
Notice that $\qubit{\zeta_{u,p_u,\phi}^{(x'r_0)}}$ equals $\sum_{s:|s|=m_u(\qubit{\phi})} ( f_u(\measure{u}{x'r_0} \otimes \qubit{s})\otimes \measure{s}{\psi_{p_u,\phi}})$.
An application of $f_u^{-1}$ to $f_u(\measure{u}{x'r_0}\otimes \qubit{s})$ leads to $\sum_{u:|u|=|r_0|} \measure{u}{x'r_0}\otimes \qubit{\psi_{p_u,\phi}}$. Since $\qubit{x'r_0}= d(\qubit{xr_0})$, by applying $d^{-1}$ and $p_u^{-1}$ (which come from the definition of $G$), we finally obtain $\qubit{xr_0}\otimes \qubit{\phi}$.
For Scheme V, it suffices to set $LCompo[g]^{-1}$ to be $LCompo[g^{-1}]$.

The above argument also proves the second part of the lemma.
\end{proof}

\subsection{Section-Wise Handling of Binary Encoding}

We have discussed in Section \ref{sec:different-object} the binary encodings of various objects usable as part of inputs.
We further explore the characteristics of quantum functions that can handle  these binary encodings.


Given $k$ encoded strings  $\widetilde{x_1},\widetilde{x_2},\ldots,\widetilde{x_k}$, we merge them into a single string of the form $\widetilde{x_1} \widetilde{x_2} \cdots \widetilde{x_k}\hat{2}$, where $\hat{2}$ serves as an endmarker.
Let $k\geq2$ and let $\GG=\{g_i\}_{i\in[k]}$ denote a series of $k$  quantum functions. We then consider a simultaneous application of all quantum functions in $\GG$, $MultiApp_k[\GG]$, given as:
\begin{eqnarray*}
\lefteqn{MultiApp_k[\hat{2},\GG] (\qubit{\widetilde{x_1} \widetilde{x_2} \cdots \widetilde{x_k}\hat{2} } \qubit{\phi})} \hs{40} \\
&=& \hs{2} g_1(\qubit{\widetilde{x_1}}) \otimes g_2(\qubit{\widetilde{x_2}}) \otimes \cdots \otimes  g_k(\qubit{\widetilde{x_k} \hat{2}}) \otimes \qubit{\phi}.
\end{eqnarray*}

\begin{lemma}\label{MultiApp-case}
For any $k\geq2$ and any series $\GG=\{g_i\}_{i\in[k]}$ of $k$ quantum functions, the quantum function $MultiApp_k[\hat{2},\GG]$ is definable from $\GG$ and the code skipping scheme.
\end{lemma}

\begin{proof}
Let  $\GG=\{g_i\}_{i\in[k]}$ denote any series of $k$ quantum functions, not necessarily in $EQS$. For brevity, we set $\bar{r}_0=\hat{\dashv}$ and $r_0 = \hat{2}$. The quantum function $MultiApp_k[\hat{2},\GG]$ is constructed from $\GG$ by applying $CodeSKIP_{+}[\cdot]$
in the following inductive way.
Initially, we set  $G_{[k,k]}\equiv CodeSKIP_{+}[r_0,g_k,I]$.
Let $X_k = \widetilde{x_k}$ ($= \widehat{x_k}\hat{\dashv}$).
Since $NON_{r_0}(\qubit{X_k}\qubit{r_0}\qubit{\phi}) = \{X_k\}$, we obtain $G_{[k,k]}(\qubit{X_k}\qubit{r_0}\qubit{\phi}) = g_k(\qubit{X_k}\qubit{r_0}) \otimes \qubit{\phi}$.
Let $i$ be any index in $[k-1]$ and set $X_{i+1} = \widetilde{x_{i+1}}\widetilde{x_{i+2}}\cdots \widetilde{x_k}$.
We assume by induction hypothesis that  $G_{[i+1,k]}(\qubit{X_{i+1}}\qubit{r_0}\qubit{\phi}) = g_{i+1}(\qubit{\widetilde{x_{i+1}}})\otimes g_{i+2}(\qubit{\widetilde{x_{i+2}}}) \otimes \cdots \otimes g_k(\qubit{\widetilde{x_k}}\qubit{r_0})\otimes \qubit{\phi}$.
We then define $G_{[i,k]}\equiv CodeSKIP_{+}[\bar{r}_0,g_i,G_{[i+1,k]}]$.
Since $NON_{\bar{r}_0}(\qubit{X_{i}}\qubit{r_0}\qubit{\phi}) = \{\widehat{x_i}\}$, it follows that $G_{[i,k]}(\qubit{X_{i}}\qubit{r_0}\qubit{\phi}) = g_i(\qubit{\widehat{x_i}\bar{r}_0}) \otimes G_{[i+1,k]}(\qubit{X_{i+1}}\qubit{r_0}\qubit{\phi}) = g_i(\qubit{\widetilde{x_i}}) \otimes g_{i+1}(\qubit{\widetilde{x_{i+1}}})\otimes \cdots \otimes g_k(\qubit{\widetilde{x_k}}\qubit{r_0})\otimes \qubit{\phi}$.

The desired quantum function $MultiApp_k[r_0,\GG]$ therefore equals $G_{[1,k]}$, which is $g_1(\qubit{\widetilde{x_1}}) \otimes g_2(\qubit{\widetilde{x_2}}) \otimes \cdots \otimes  g_k(\qubit{\widetilde{x_k} \hat{2}}) \otimes \qubit{\phi}$.
\end{proof}

By Lemma \ref{MultiApp-case}, we can construct $MultiApp_k[\hat{2},\GG]$ from $\GG = \{g_i\}_{i\in [k]}$ and $CodeSKIP_{+}[r_0,g,h]$.
Since $CodeSKIP_{+}[r_0,g,h]$ stays outside of $EQS$, $MultiApp_k[\hat{2},\GG]$ in general does not belong to  $EQS$.
However, as the next lemma ensures, we can use $MultiApp_k[\hat{2},\GG]$ as if it is an $EQS$ function under certain circumstances.

To describe our result, let us recall the 3-bit encoding of $\hat{S}$, which indicates ``separator''.

\begin{lemma}\label{code-skip-application}
Let $k\in\nat^{+}$ and let $\GG=\{g_i\}_{i\in[k]}$ be any series of $k$ quantum functions in $EQS$.
There exists a quantum function $F$ in $EQS$ such that $F(\qubit{\hat{S} \widetilde{x}_1 \widetilde{x}_2\cdots \widetilde{x}_k \hat{2}}\qubit{\phi}) =  MultiApp_{k}[\hat{2},\GG](\qubit{\hat{S} \widetilde{x}_1 \widetilde{x}_2 \cdots \widetilde{x}_k \hat{2}}\qubit{\phi})$ as long as $|\hat{S} \widetilde{x}_1\widetilde{x}_2\cdots \widetilde{x}_k \hat{2}|\leq |\hat{2}|\log{\ell(\qubit{\phi})}$. However, the equality may not hold if $|\hat{S} \widetilde{x}_1\widetilde{x}_2\cdots \widetilde{x}_k \hat{2}| > |\hat{2}|\log{\ell(\qubit{\phi})}$.
\end{lemma}

\begin{proof}
Let $k\in\nat^{+}$ and let $\GG= \{g_i\}_{i\in[k]}$ be given as in the premise of the lemma.
Since $|\hat{2}|=3$, we set $h$ to be $Branch_3[\{h_u\}_{u\in\{0,1\}^3}]$ with $h_{\hat{S}} \equiv MultiApp_{k}[\hat{2},\GG]$ and $h_u \equiv I$ for any $u\neq \hat{S}$.
The desired quantum function $F$ is defined to be $CFQRec_1[\hat{2},I,I,h|\PP_{3},\FF_{3}]$, where  $\PP_3=\{p_u\}_{u\in\{0,1\}^3}$ with $p_u=I$ for all $u\in\{0,1\}^3$ and $\FF_3=\{f_u\}_{u\in\{0,1\}^3}$ with $f_u = F$ for all $u\in\{0,1\}^3$.
Because of $|\hat{S} \widetilde{x}_1\widetilde{x}_2\cdots \widetilde{x}_k \hat{2}| \leq |\hat{2}|\log{\ell(\qubit{\phi})}$, in  the following recursive process of computing $F$, it suffices to start with $F(\qubit{\hat{2}}\qubit{\phi'})$ with $\ell(\qubit{\phi'})\geq 1$.

For each index $i\in[k]$, we abbreviate $\widetilde{x_i}\widetilde{x_{i+1}}\cdots \widetilde{x_k}$ as $X_i$.
In what follows, we fix $i\in[k]$ arbitrarily and assume that $F(\qubit{X_i\hat{2}} \qubit{\phi}) = \qubit{X_i\hat{2}} \qubit{\phi}$.
Note that $\qubit{\zeta_{u,I,\phi}} = \sum_{s:|s|=LH(\qubit{\phi})} (F(\measure{u}{X_{i-1}\hat{2}}\qubit{s}) \otimes \measure{s}{\phi})$.
If $u=\widetilde{x_{i-1}}$, then
$\qubit{\zeta_{u,I,\phi}} = \sum_{s:|s|=LH(\qubit{\phi})} ( F(\qubit{X_i\hat{2}}\qubit{s}) \otimes \measure{s}{\phi}) = \sum_{s:|s|=LH(\qubit{\phi})} (\qubit{X_i\hat{2}}\qubit{s} \otimes \measure{s}{\phi}) = \qubit{X_i\hat{2}}\qubit{\phi}$.
From this, it follows that $F(\qubit{X_i\hat{2}} \qubit{\phi}) = h(\qubit{\widetilde{x_{i-1}}}\qubit{X_i\hat{2}}) \otimes \measure{X_i\hat{2}}{\zeta_{\widetilde{x_{i-1}},I,\phi}} = \qubit{X_{i-1}\hat{2}}\qubit{\phi}$.
Finally, we consider the case of $\qubit{\hat{S}}\qubit{X_1\hat{2}}\qubit{\phi}$. We then obtain $F(\qubit{\hat{S}}\qubit{X_1\hat{2}}\qubit{\phi}) = h(\qubit{\hat{S}}\qubit{X_1\hat{2}}) \otimes \measure{X_1\hat{2}}{\zeta_{\hat{S},I,\phi}}$.
Since $\qubit{\zeta_{\hat{S},I,\phi}} =
\qubit{X_1\hat{2}}\otimes \qubit{\phi}$, it follows that $F(\qubit{\hat{S}}\qubit{X_1\hat{2}}\qubit{\phi}) =
h(\qubit{\hat{S}}\qubit{X_1\hat{2}}) \otimes \measure{X_1\hat{2}}{X_1\hat{2}} \qubit{\phi} =
h(\qubit{\hat{S}}\qubit{X_1\hat{2}})\otimes \qubit{\phi} = Branch_3[\{h_u\}_{u\in\{0,1\}^3}](\qubit{\hat{S}}\qubit{X_1\hat{2}})\otimes \qubit{\phi} = MultiApp_k[\hat{2},\GG](\qubit{\hat{S}}\qubit{X_1\hat{2}}) \otimes \qubit{\phi}$.
\end{proof}


In a certain limited case, Lemma \ref{code-skip-application} makes possible  a repeated application of any quantum function within $EQS_0+IV$ (as well as $EQS$). Let us recall the notation $\hat{T}$, which indicates ``time''.

\begin{proposition}
Let $r_0\in\{0,1\}^+$ and let $g$ be any quantum function definable within $EQS_0+IV$ (resp., $EQS$). There exists a quantum function $F$ definable within $EQS_0+IV$ (resp., $EQS$) such that, for any $\qubit{\phi}\in\HH_{\infty}$ and $x\in\{0,1\}^*$, if $\ell(\qubit{\phi})\leq1$, then $F(\qubit{\hat{T}^{\:m(\qubit{\phi})} \hat{S}} \qubit{x\hat{2}}\qubit{\phi}) = \qubit{\hat{T}^{\:m(\qubit{\phi})} \hat{S}} \qubit{x\hat{2}}\qubit{\phi}$, and otherwise,
$F(\qubit{\hat{T}^{\:m(\qubit{\phi})} \hat{S}} \qubit{x\hat{2}}\qubit{\phi}) = \qubit{\hat{T}^{\:m(\qubit{\phi})} \hat{S}} \otimes  g^{m(\qubit{\phi})}(\qubit{x\hat{2}}) \otimes \qubit{\phi}$, where $m(\qubit{\phi}) = \ilog(\ell(\qubit{\phi}))$.
\end{proposition}

\begin{proof}
Let $F$ denote the desired quantum function. Remember that $|\hat{2}|=|\hat{T}|=|\hat{S}|=3$. We set $\PP_3=\{p_u\}_{u\in\{0,1\}^3}$ with $p_u=I$ for all $u\in\{0,1\}^3$ and set $\GG=\{g_u\}_{u\in\{0,1\}^3}$ with
$g_{\hat{T}}=CodeSKIP_{+}[\hat{S},I,g]$, and $g_u=I$ for all other indices $u\in \{0,1\}^3-\{\hat{T}\}$. We then define $h$ to be $Branch_3[\GG]$. Finally, $F$ is defined as $F\equiv CFQRec_1[\hat{2},I,I,h|\PP_3,\FF_3]$, where  $\FF_3=\{f_u\}_{u\in\{0,1\}^3}$ with $f_u=F$ for all $u\in\{0,1\}^3$.

For any index $i\in[0,m(\qubit{\phi})]_{\integer}$, we set   $z_i=\hat{T}^{m(\qubit{\phi})-i}$.
If $\ell(\qubit{\phi})\leq 1$, then we obtain $F(\qubit{z_i \hat{S}}\qubit{x\hat{2}}\qubit{\phi}) = \qubit{z_i \hat{S}}\qubit{x\hat{2}}\qubit{\phi}$ for any $i\in[0,m(\qubit{\phi})]_{\integer}$.
Next, we assume that $\ell(\qubit{\phi})\geq 2$. Let $i\in[0,m(\qubit{\phi})]_{\integer}$ and set $y=z_i\hat{S}x$.
By induction hypothesis, $F(\qubit{z_{i+1}  \hat{S}}\qubit{x\hat{2}}\qubit{s} = \qubit{z_{i+1} \hat{S}} \otimes g^{m(\qubit{\phi})-i-1}(\qubit{x\hat{2}})\otimes \qubit{\phi}$ holds.
Let us consider the value $F(\qubit{z_i \hat{S}}\qubit{x\hat{2}}\qubit{\phi}$.
Note that $h(\qubit{z_i\hat{S}} \otimes g^{m(\qubit{\phi})-i}(\qubit{x\hat{2}}) ) = CodeSKIP_{+}[\hat{S},I,g](\qubit{z_i\hat{S}} \otimes g^{m(\qubit{\phi})-i}(\qubit{x\hat{2}})) = \qubit{z_i\hat{S}} \otimes g^{m(\qubit{\phi})-i-1}(\qubit{x\hat{2}})$.
It thus follows that $F(\qubit{z_i \hat{S}}\qubit{x\hat{2}}\qubit{\phi} = \sum_{u:|u|=|\hat{2}|} \sum_{v:|v|=\ell(\measure{u}{y\hat{2}})} (h(\qubit{u}\otimes \qubit{v}) \otimes \measure{v}{\zeta_{u,I,\phi}})$,
where $\qubit{\zeta_{u,I,\phi}}$ is calculated as follows. If $u\neq \hat{T}$, then $\qubit{\zeta_{u,I,\phi}} = \bfzero$ holds.
Otherwise,
$\qubit{\zeta_{u,I,\phi}}$ equals $\sum_{s:|s|=LH(\qubit{\phi})} (F(\measure{u}{z_i\hat{S}} \qubit{x\hat{2}}\otimes \qubit{s}) \otimes \measure{s}{\phi}) = \sum_{s:|s|=LH(\qubit{\phi})} (F(\qubit{z_{i+1}\hat{S}} \qubit{x\hat{2}}\otimes \qubit{s}) \otimes \measure{s}{\phi}) =
\sum_{s:|s|=LH(\qubit{\phi})} h(\qubit{z_{i+1}\hat{S}} \otimes g^{m(\qubit{\phi})-i-1}( \qubit{x\hat{2}} ) \otimes \qubit{s}  \measure{s}{\phi}) = \qubit{z_{i+1}\hat{S}} \otimes g^{m(\qubit{\phi})-i-1}(\qubit{x\hat{2}})\otimes \qubit{\phi}$.
Therefore, we obtain $F(\qubit{z_i\hat{S}} \qubit{x\hat{2}}\qubit{\phi}) = \qubit{z_i\hat{S}} \otimes g^{m(\qubit{\phi})-i} (\qubit{x\hat{2}})\otimes \qubit{\phi}$.
In particular, when $i=0$, we conclude that $F(\qubit{\hat{T}^{\:m(\qubit{\phi})} \hat{S}} \qubit{x\hat{2}}\qubit{\phi}) = \qubit{\hat{T}^{\:m(\qubit{\phi})} \hat{S}} \otimes  g^{m(\qubit{\phi})}(\qubit{x\hat{2}}) \otimes \qubit{\phi}$.
\end{proof}

\subsection{Implementing the Binary Search Strategy}

The strength of Scheme IV is further exemplified by an implementation of
a simple \emph{binary search} algorithm.
Given any string $x\in\{0,1\}^k$ with $k\geq1$,  assume that $x$ equals $bin_k(m)$ for a certain number $m\in[2^k]$.
For any $s$ with $|s|=2^k$ and $b\in\{0,1\}$, the quantum function
$BinSearch$ satisfies the following equality: $BinSearch(\qubit{\widetilde{x}}\qubit{\hat{b}} \qubit{\hat{2}} \qubit{s}) = \qubit{\widetilde{x}}\qubit{\widehat{b\oplus s_{(m)}}}  \qubit{\hat{2}} \qubit{s}$, where $s_{(m)}$ is the $m$th bit of $s$.
This quantum  function finds the $m$th bit of $s$ and extracts its bit $s_{(m)}$ from $s$ by way of modifying $\qubit{\hat{b}}$ to $\qubit{\widehat{b\oplus s_{(m)}}}$.

\begin{theorem}\label{bit-search-lemma}
The above quantum function $BinSearch$ is definable within $EQS_0+IV$.
\end{theorem}

\begin{proof}
We remark that the length of $s$ given to $BinSearch(\qubit{\widetilde{x}}\qubit{\hat{b}} \qubit{\hat{2}} \qubit{s})$ is a power of $2$. In this proof, we intend to use the quantum functions $SecSWAP^{(3)}_{i,j}$ and $COPY_1$ introduced in Lemmas \ref{lemma:special-combi}(1) and \ref{COPY-function}, respectively.
Given any strings $x=x_1x_2\cdots x_k$, $s=s_1s_2\cdots s_{2^k}$, and $b\in\{0,1\}$ associated with the number $m\in[2^k]$ satisfying $x=bin_k(m)$,
we define the quantum function $g$ by setting
$g(\qubit{{\hat{\dashv}}} \qubit{\hat{b}}\qubit{\hat{2}}\qubit{s_m})  = \qubit{{\hat{\dashv}}} \qubit{\widehat{b\oplus s_m}} \qubit{\hat{2}} \qubit{s_m}$. Notice that $\hat{b}=00b$ and $\widehat{b\oplus s_m} = 00b'$ with $b'=b\oplus s_m$.
It is possible to realize $g$ by simply setting $g\equiv SWAP_{7,10} \circ Skip_5[COPY_1]\circ SWAP_{7,10}$ since $COPY_1(\qubit{b}\qubit{s_m}\qubit{w}) = \qubit{b\oplus s_m}\qubit{s_m}\qubit{w}$ for any $w$.
Let $p_{\hat{0}} \equiv I$, $p_{\hat{1}}\equiv HalfSWAP$, and $p_u=I$ for all indices $u\in\{0,1\}^3-\{\hat{0},\hat{1}\}$.
We then define $F$ to be $CFQRec_1[\hat{2},I,g,I| \PP_3,\FF_3]$, where $\PP_3=\{p_u\}_{u\in\{0,1\}^3}$ and $\FF_3=\{f_u\}_{u\in\{0,1\}^3}$ with $f_u=F$ for all $u\in\{0,1\}^3$. Now, our goal is to verify that $F$ indeed matches $BinSearch$.

For any string $s\in\{0,1\}^{+}$ whose length is a power of $2$ and for any  number $m\in[|s|]$, we explain how to find the $m$th bit $s_m$ of $s$.
We first split $s$ into the left part and the right part of $s$ whose lengths are $LH(|s|)$ and $RH(|s|)$, respectively. We assign $0$ to the left part and $1$ to the right part and we call the left part by $s_0$ and the right part by $s_1$. Starting $s_0$ (resp., $s_1$), we further split it into its  left part, called $s_{00}$ (resp., $s_{10}$), and its right part, called $s_{01}$ (resp., $s_{11}$). Inductively, we repeat this process until the target strings become single symbols. In the end, a string $x\in\{0,1\}^{\ilog(|s|)}$ is assigned to the single symbol obtained by the series of the above-described processes. We denote this unique symbol by $s_x$. If $x=bin_{\ilog(|s|)}(m)$, then $s_x$ coincides with the desired bit $s_m$.
We then treat $x$ as the binary representation of an index of the symbol $s_{x}$ in $s$.

Let $X_{k+1}=\hat{\dashv}$ and let $X_i=\widehat{x_i} \widehat{x_{i+1}} \cdots \widehat{x_k} \hat{\dashv}$ for any index $i\in[k]$.
We split $s$ into two parts as $s=s^{(r)}s^{(l)}$ with $|s^{(r)}|=LH(|s|)$ and $|s^{(l)}|=RH(|s|)$.
We set $m_1=m$ and, for each index $i\in[2,k]_{\integer}$, we take the number $m_i$ satisfying $bin_{k-i+1}(m_i)=x_ix_{i+1}\cdots x_k$.
Note that $s_{m_i} = s^{(r)}_{m_{i+1}}$ if $x_i=1$, and
$s_{m_i} = s^{(l)}_{m_{i+1}}$ otherwise.
When $\ell(\qubit{s})=1$, we conclude that $F(\qubit{X_{k+1}}\qubit{\hat{b}\hat{2}}\qubit{s}) = g(\qubit{X_{k+1}}\qubit{\hat{b}\hat{2}}\qubit{s}) = g(\qubit{\hat{\dashv}}\qubit{\hat{b}}\qubit{\hat{2}}\qubit{s}) = \qubit{\hat{\dashv}}\qubit{\widehat{b\oplus s}}\qubit{\hat{2}}\qubit{s} = BinSearch(\qubit{X_{k+1}}\qubit{\hat{b}\hat{2}}\qubit{s})$. Next, we assume that $\ell(\qubit{s})=i\geq2$.
It follows by induction hypothesis  that $F(\qubit{X_{i+1}}\qubit{\hat{b}\hat{2}} \qubit{s^{(r)}}) = \qubit{X_{i+1}} \qubit{\hat{u}}  \qubit{\hat{2}}\qubit{s^{(r)}}$ with $u= b\oplus s^{(r)}_{m_{i+1}}$ if $x_i=1$, and
$F(\qubit{X_{i+1}}\qubit{\hat{b}\hat{2}} \qubit{s^{(l)}}) = \qubit{X_{i+1}} \qubit{\hat{v}}  \qubit{\hat{2}}\qubit{s^{(l)}}$ with $v= b\oplus s^{(l)}_{m_{i+1}}$ otherwise.
In the case of $u=1$, since $\qubit{\psi_{p_{\hat{u}},s}} = HalfSWAP(\qubit{s})$, it follows that
$\qubit{\zeta_{\hat{u},p_{\hat{u}},s}} = \sum_{t:|t|=RH(|s|)} ( F(\qubit{X_{i+1}}\qubit{\hat{b}\hat{2}}\qubit{t}) \otimes \measure{t}{s^{(r)}s^{(l)}})  = \qubit{X_{i+1}}\qubit{\hat{u}}\qubit{\hat{2}} \qubit{s^{(l)}}$.
Therefore, when $x_i=1$, we obtain  $F(\qubit{X_i}\qubit{\hat{b}\hat{2}}\qubit{s}) = \qubit{x_i}\qubit{X_{i+1}} \qubit{\hat{u}}\qubit{\hat{2}}\qubit{s}  = BinSearch(\qubit{X_i}\qubit{\hat{b}\hat{2}}\qubit{s})$. When $x_i=0$, in  contrast, we obtain $F(\qubit{X_i}\qubit{\hat{b}\hat{2}}\qubit{s}) = \qubit{x_i}\qubit{X_{i+1}} \qubit{\hat{v}}\qubit{\hat{2}}\qubit{s} = BinSearch(\qubit{X_i}\qubit{\hat{b}\hat{2}}\qubit{s})$.

As a result, we conclude that $F = BinSearch$, as requested.
\end{proof}


Hereafter, we demonstrate how to use the quantum function $BinSearch$.
For this purpose, we first show the following statement.

\begin{corollary}\label{Bit-function}
Given $\qubit{\phi}\in\HH_{\infty}$ and $x=bin_k(m)$ for $k\in\nat^{+}$ and $m\in[2^k]$, if $\qubit{\phi}$ is of the form  $\sum_{s:|s|=2^k}\alpha_s \qubit{s}$, then we set  $Bit(\qubit{0^5}\qubit{b} \qubit{\widetilde{x}}\qubit{\phi}) = \sum_{s:|s|=2^k}\alpha_s \qubit{0^5}\qubit{b\oplus s_{(m)}}\qubit{\widetilde{x}}\qubit{s}$, where $s_{(m)}$ is the $m$th bit of $s$.
This quantum function $Bit$ is definable within $EQS_0+IV$.
\end{corollary}

\begin{proof}
Recall the quantum function $Skip_k[g]$ from Lemma \ref{Skip-function}. We first change the quantum state $\qubit{0^5}\qubit{b}\qubit{\widetilde{x}}\qubit{\phi}$ into $\qubit{\hat{b}}\qubit{\hat{2}}\qubit{\widetilde{x}}\qubit{\phi}$ by applying $h\equiv SWAP_{1,4}\circ SWAP_{2,5}\circ SWAP_{3,6} \circ Skip_2[NOT]\circ Skip_1[NOT]\circ NOT$.

Consider the quantum function $f$ defined by $f(\qubit{\phi}) = \sum_{z:|z|=6} \measure{z}{\phi}\otimes \qubit{z}$ for all $\qubit{\phi}\in\HH_{\infty}$. This $f$ satisfies the following recursive property:  $f(\qubit{x}\qubit{\hat{\dashv}}) = SecSWAP^{(3)}_{1,3} \circ SecSWAP^{(3)}_{2,4} ( \qubit{u_1u_2}\otimes f(\measure{u_1u_2}{x{\hat{\dashv}}}))$, where $x=u_1u_2x'$ with $|u_1|=|u_2|=3$.

Proposition \ref{limited-recursion} makes it possible to realize the quantum function $F$ that satisfies $F(\qubit{\hat{b}}\qubit{\hat{2}}\qubit{\widetilde{x}}\otimes \qubit{\phi}) =  f(\qubit{\hat{b}}\qubit{\hat{2}}\qubit{\widetilde{x}})\otimes \qubit{\phi}$. The last term actually equals  $\qubit{\widetilde{x}} \qubit{\hat{b}}\qubit{\hat{2}}\qubit{\phi}$.
We apply $BinSearch$ to $\qubit{\widetilde{x}} \qubit{\hat{b}}\qubit{\hat{2}}\qubit{s}$ to obtain $\qubit{\widetilde{x}} \qubit{\widehat{b\oplus s_{(m)}}}\qubit{\hat{2}}\qubit{s}$.
We further apply $f^{-1}$ and obtain $f^{-1}( \qubit{\widetilde{x}} \qubit{\widehat{b\oplus s_{(m)}}}\qubit{\hat{2}}\qubit{s}) = \qubit{\widehat{b\oplus s_{(m)}}}\qubit{\hat{2}} \qubit{\widetilde{x}} \qubit{s}$. Finally, we apply $h^{-1}$ to $\qubit{\widehat{b\oplus s_{(m)}}}\qubit{\hat{2}} \qubit{\widetilde{x}} \qubit{s}$ and obtain $\qubit{0^5}\qubit{b\oplus s_{(m)}} \qubit{\widetilde{x}} \qubit{s}$. Therefore, $Bit$ can be defined by a finite series of applications of Scheme I--IV.
\end{proof}


As the second application of $BinSearch$,
we wish to ``count'' the number of 0s and 1s in an input string in a quantum-mechanical fashion. It is impossible to do so deterministically in polylogarithmic time.
Fix a constant $\varepsilon\in[0,3/4)$ and consider the \emph{promise decision problem}  $MAJPDP_{\varepsilon}$ in which we  determine whether the total number of $0$s in $x$ is at least $\sqrt{1-\varepsilon}|x|$ or the total number  of $1$s is at least $\sqrt{1-\varepsilon}|x|$.
Formally, $MAJPDP_{\varepsilon}$ is expressed as $(A_{\varepsilon},B_{\varepsilon})$, where $A_{\varepsilon} =\{x\in\{0,1\}^*\mid \#_1(x)\geq \sqrt{1-\varepsilon}|x|\}$ and $B_{\varepsilon} =\{x\in\{0,1\}^*\mid \#_0(x)\geq \sqrt{1-\varepsilon}|x|\}$. We intend to prove the existence of a quantum function in $EQS_0+IV$ that ``solves'' this promise problem $MAJPDP_{\varepsilon}$ in the following sense.

\begin{proposition}\label{access-all-qubits}
Let $\varepsilon$ be any constant in $[0,3/4)$. There exists a quantum function $F$ in $EQS_0+IV$ such that, for any $x\in\{0,1\}^*$, (1) if $x\in A_{\varepsilon}$, then $\|\measure{1}{\psi_{F,x}}\|^2\geq 1-\varepsilon$ and (2) if $x\in B_{\varepsilon}$, then $\|\measure{0}{\psi_{F,x}}\|^2\geq 1-\varepsilon$, where $\qubit{\psi_{F,x}} = F(\qubit{0^{3k}}\qubit{\hat{\dashv}} \qubit{x})$ with $k=\ilog(|x|)$.
\end{proposition}


To prove this proposition, we first demonstrate how to produce a superposition of all ``indices'' of a given input. This can be done by recursively applying $WH$ to $\qubit{0^{3k}}\qubit{\hat{\dashv}}$ as shown below.

\begin{lemma}\label{WH-application}
There exists a quantum function $G$ in $EQS_0+IV$ satisfying $G(\qubit{0^{3k}}\qubit{\hat{\dashv}}\otimes  \qubit{\phi}) = \frac{1}{\sqrt{2^k}} \sum_{x:|x|=k}\qubit{\widetilde{x}}\qubit{\phi}$ for any $n\in\nat^{+}$ and any $\qubit{\phi}\in\HH_{\infty}$,  provided that $k= \ilog(\ell(\qubit{\phi}))$.
\end{lemma}

\begin{proof}
Consider the quantum function $\hat{h}(\qubit{0^{3k}}\qubit{\hat{\dashv}}) = \frac{1}{\sqrt{2^k}} \sum_{x:|x|=k}\qubit{\widetilde{x}}$. This function $\hat{h}$ satisfies the following recursive property: $\hat{h}(\qubit{\hat{\dashv}}) = \qubit{\hat{\dashv}}$ and $\hat{h}(\qubit{0^3}\qubit{w\hat{\dashv}}) = h(\qubit{0^3}\otimes \hat{h}(\qubit{w\hat{\dashv}}))$ for any string $w$, where $h\equiv Branch_2[\{g_u\}_{u\in\{0,1\}^2}]$ with $g_{00}=WH$ and $g_u=I$ for all indices $u\in\{0,1\}^2-\{00\}$.
Let us prove this property.
Assume that $\hat{h}(\qubit{0^{3k}}\qubit{\hat{\dashv}}) = \frac{1}{\sqrt{2^k}}\sum_{x:|x|=k}\qubit{\widetilde{x}}$.
For the value $\hat{h}(\qubit{0^{3k+3}}\qubit{\hat{\dashv}})$, we calculate  $h(\qubit{0^3}\otimes \hat{h}(\qubit{0^{3k}}\qubit{\hat{\dashv}}))$ as $\qubit{00}\otimes g(\qubit{0}\otimes \hat{h}(\qubit{0^{3k}}\qubit{\hat{\dashv}})) = \qubit{00}\otimes \frac{1}{\sqrt{2}}(\qubit{0}+\qubit{1})\otimes \hat{h}(\qubit{0^{3k}}\qubit{\hat{\dashv}}) = \frac{1}{\sqrt{2}}(\qubit{\hat{0}}+\qubit{\hat{1}})\otimes \frac{1}{\sqrt{2^k}}\sum_{x:|x|=3k}\qubit{\widetilde{x}} = \frac{1}{\sqrt{2^{k+1}}}\sum_{y:|y|=k+1}\qubit{\widetilde{y}}$.
The last term clearly equals $\hat{h}(\qubit{0^{3k+3}}\qubit{\hat{\dashv}})$.

By Proposition \ref{limited-recursion}, there exists a quantum function $G$ in $EQS_0+IV$ for which $G(\qubit{0^{3k}}\qubit{\hat{\dashv}} \otimes \qubit{\phi}) = \hat{h}(\qubit{0^{3k}}\qubit{\hat{\dashv}}) \otimes \qubit{\phi}$, where $k= \ilog(\ell(\qubit{\phi}))$. This completes the lemma's proof.
\end{proof}


With the help of Lemma \ref{WH-application}, the proof of  Proposition \ref{access-all-qubits} easily follows.

\begin{proofof}{Proposition \ref{access-all-qubits}}
Consider $MAJPDP_{\varepsilon} = (A_{\varepsilon},B_{\varepsilon})$ defined above. We wish to construct a quantum function, say, $F$ in $EQS_0+IV$ that ``solves'' $MAJPDP_{\varepsilon}$  in the proposition's sense. Recall the quantum function $G$ of Lemma \ref{WH-application}. Since $G$ is norm-preserving and thus in $\widehat{EQS}$, Lemma \ref{inverse-function} ensures that $G^{-1}$ exists.

Let $x$ denote any string. We then set $b_x=1$ if $x\in A_{\varepsilon}$ and $b_x=0$ if $x\in B_{\varepsilon}$.
We start the desired computation with the quantum state $\qubit{\phi_x}= \qubit{0^6} \qubit{0^{3k}}\qubit{\hat{\dashv}} \qubit{x}$, where $k=\ilog(|x|)$, and apply $Skip_6[G]$ to generate $\frac{1}{\sqrt{2^k}}\sum_{u:|u|=k}\qubit{\widetilde{u}}\qubit{x}\qubit{0^6}$.
We then move $\qubit{0^6}$ to the front and obtain  $\frac{1}{\sqrt{2^k}}\sum_{u:|u|=k}\qubit{0^6}\qubit{\widetilde{u}}\qubit{x}$. We further apply $Bit$ to the resulting quantum state and generate  $\frac{1}{\sqrt{2^k}}\sum_{u:|u|=k} \qubit{0^5}\qubit{x_{(m(u))}}\qubit{\widetilde{u}}\qubit{x}$, where $m(u)$ denotes a unique number satisfying that  $u=bin_k(m(u))$ and $x_{(m(u))}$ is the $m(u)$-th bit of $x$. We change it to $\qubit{\gamma_x} = \frac{1}{\sqrt{2^k}}\sum_{u:|u|=k} \qubit{\widetilde{u}}\qubit{x}\qubit{0^5} \qubit{x_{(m(u))}}$ and  then apply $G^{-1}$. We denote the resulting quantum state by $\qubit{\beta_x}$.
Letting $\qubit{\xi}= \qubit{0^{3k}}\qubit{\hat{\dashv}}$, we ish to calculate $\measure{\xi}{\beta_x}$.
Note that $\qubit{\beta_x} = G^{-1}(\qubit{\gamma_x})$ is equivalent to  $G(\qubit{\beta_x}) = \qubit{\gamma_x}$.
It thus follows that  $\measure{\xi}{\beta_x}$ equals $\measure{\psi_{G,\xi}}{\gamma_x}$, which is further calculated to  $(\frac{1}{\sqrt{2^k}}\sum_{v:|v|=k}\bra{\widetilde{v}})\cdot (\frac{1}{\sqrt{2^k}}\sum_{u:|u|=k} \ket{\widetilde{u}}\qubit{x}\qubit{0^5}\qubit{x_{(m(u))}}) = \frac{1}{2^k}\qubit{x}\qubit{0^5}\qubit{x_{(m(u))}}$, where $\qubit{\psi_{G,\xi}}=G(\qubit{\xi})$.

After removing the last qubit to the front, we obtain a unique quantum state, say, $\qubit{\eta_x}$.
Finally, we measure the first qubit of $\qubit{\eta_x}$ in the computational basis. If $x\in A_{\varepsilon}\cup B_{\varepsilon}$, then $b_x$ satisfies that $\|\measure{b_x}{\eta_x}\|^2 \geq (\frac{1}{2^{k}}\sum_{u}\measure{b_x}{x_{(m(u))}})^2 =  (\frac{\#_{b_x}(x)}{2^k})^2 \geq (\sqrt{1-\varepsilon})^2 = 1-\varepsilon$ because of the promise given by $(A_{\varepsilon},B_{\varepsilon})$.
\end{proofof}

\section{Relationships to Quantum Computability}\label{sec:relationships}

Throughout Section \ref{sec:definitions}, we have studied basic properties of quantum functions in $EQS$. In this section, we will look into relationships of these quantum functions to other platforms of limited computability, in particular, a model of polylogarithmic-time (or polylogtime) Turing machine.

\subsection{Runtime-Restricted Quantum Turing Machines}\label{sec:polylogtime-QTM}

We wish to describe a computational model of \emph{quantum Turing machine} (or QTM, for short) that runs particularly in polylogarithmic time. For this purpose, we need to modify a standard model of QTM defined in  \cite{BV97,ON00,Yam99} structurally and behaviorally. This new model also expands the classical model of (poly)logtime Turing machine (TM) discussed in \cite{BIS90}. Notice that quantum polylogtime computability was already discussed in, e.g., \cite{RT22} based on uniform quantum circuit families. For more information, refer to \cite{RT22} and references therein.

A \emph{(random-access) QTM}\footnote{This model is also different from the ``log-space QTMs'' of \cite{Yam22}, which are equipped with ``garbage'' tapes onto which any unwanted information is discarded to continue their quantum computation.} (or just a \emph{QTM} in this work) is equipped with a random-access read-only input tape, multiple rewritable work tapes of $O(\log^k{n})$ cells, and a rewritable index tape of exactly $\ilog(n)+1$ cells, where $k$ is a constant in $\nat^{+}$ and $n$ refers to the length of an input. The index tape indicates  the cell location (or address) of the input tape, specifying which qubit of a given input we wish to access. This QTM $M$ is formally expressed as $(Q,\Sigma,\{\rhd,\lhd\},\delta, q_0,Q_{acc},Q_{rej})$, where $Q$ is a finite set of inner states, $\Sigma=\{0,1\}$ is an alphabet,
$\rhd$ and $\lhd$ are endmarkers,
$\delta$ is a quantum transition function, $q_0$ ($\in Q)$ is the initial (inner) state, and $Q_{acc}$ (resp., $Q_{rej}$) ($\subseteq Q$) is a set of accepting (resp., rejecting) states.
We use an additional convention that the input tape and the index tape, and all the work tapes have two endmarkers to mark the usable areas of these tapes. This in fact helps the machine understand the ``size'' of a given input.
The  QTM $M$ begins with an input qustring $\qubit{\phi}$ given on the input tape marked by the endmarkers.
Let $\qubit{\phi} = \sum_{x\in\Sigma^n}\alpha_x\qubit{x}$ with $n\in\nat^{+}$ and $\alpha_x\in\complex$ for all $x$'s.

Recall from Section \ref{sec:numbers} the binary encoding of natural numbers.
Let $k_n=\ilog(n)$. To access the tape cell indexed by $m$, $M$ first produces the binary string $bin_{k_n}(m)$ on the index tape  with an auxiliary bit $b$ and then enters a designated \emph{query state}, say, $q_{query}$ in $Q$.
If the index tape contains $\qubit{bin_{k_n}(m)}\qubit{b}$, then this quantum state  becomes $\qubit{bin_{k_n}(m)}\qubit{b\oplus x_{(m)}}$ as the immediate consequence of the query, where $x_{(m)}$ is the $m$th input symbol of an  input $x\in\Sigma^+$.
With the proper use of work tapes, we assume that, while writing $\qubit{bin_{k_n}(m)}\qubit{b}$ until entering a query state, the tape head never moves to the left and, whenever it writes a non-blank symbol, it must move to the right. We remark that the number of queries and their timings may vary on all computation paths of $M$ on $\qubit{\phi}$.

Concerning a random access to an input, the classical polylogtime TMs take the following convention \cite{BIS90,Vol99}.  When the machine enters a query state with index-tape content $bin_{k_n}(m)$,
the input-tape head instantly jumps to the target bit $x_{(m)}$ of the input $x$ and reads it. After this query process, the index tape remains unerased
and the corresponding tape head does not automatically return to the start cell. Therefore, for the next query, the machine can save time to rewrite the same query word but it must overwrite a different query word over the previous query word on the index tape.

The quantum transition function $\delta$ takes a quantum transition of the form $\delta(q,\sigma,\tau_1,\tau_2,\ldots,\tau_c) = \sum_{r}\alpha_r \qubit{p,\xi,\eta_1,\eta_2,\ldots,\eta_c, d,d'_1,d'_2,\ldots,d'_c}$, where $r=(p,\xi,\eta_1,\eta_2,\ldots,\eta_c, d,d'_1,d'_2,\ldots,d'_c)$, which indicates that, if $M$ is in inner state $q$ reading $\sigma$ on the index tape and $(\tau_1,\tau_2,\ldots,\tau_c)$ on the $c$ work tapes, then, in a single step, with transition amplitude $\alpha_r$, $M$ changes $q$ to $p$, writes $\xi$ over $\sigma$ moving the input-tape head in direction $d\in\{-1,+1\}$, and writes $\eta_i$ over $\tau_i$ moving the $i$th work-tape head in direction $d'_i\in\{-1,+1\}$.
For practicality, we can limit the scope of transition amplitudes of QTMs. In this work, we allow only the following two forms of quantum transitions: $\delta(q,\sigma,\tau_1,\tau_2,\ldots,\tau_c) =
\cos\theta \qubit{p_1,\xi_1,\eta_{11},\eta_{12},\ldots,\eta_{1c}, d_1,d'_{11},d'_{12},\ldots,d'_{1c}}
+ \sin\theta \qubit{p_2,\xi_2,\eta_{21},\eta_{22},\ldots,\eta_{2c}, d_2,d'_{21},d'_{22},\ldots,d'_{2c}}$
and  $\delta(q,\sigma,\tau_1,\tau_2,\ldots,\tau_c) = e^{\imath\theta} \qubit{p,\xi,\eta_1,\eta_2,\ldots,\eta_c, d,d'_1,d'_2, \ldots,d'_c}$ (based on the universality of a set of quantum gates of \cite{BBC+95}).

A QTM is required to satisfy the so-called ``well-formedness condition'' (see, e.g., \cite{Yam99,Yam03}) to guarantee that the behaviors of the QTM obeys the laws of quantum physics.
A \emph{surface configuration} of $M$ on an input $x$ of length $n$ is an element $(q,u,r,w_1,s_1,w_2,s_2,\ldots,w_c,s_c)$ of the surface-configuration set $Q\times \{\triangleright u\triangleleft\mid u\in\{0,1,B\}^{k_n+1}\} \times [0,k_n+2]_{\integer} \times ( \{\triangleright w\triangleleft \mid w\in\{0,1,B\}^{k_n}\} \times [0,k_n+1]_{\integer})^c$, which depicts the circumstance where $M$ is in inner state $q$, scanning the $r$th bit of the index tape content $u$, and the $s_i$th bit of the $i$th work tape content $w_i$.
We call the space spanned by this set of surface configurations the \emph{surface configuration space} of $M$ on $x$.
The \emph{time-evolution operator} of $M$ on the input $x$ is a map from superpositions of surface configurations of $M$ on $x$ to other superpositions of surface configurations resulting by a single application of $\delta$ of $M$. A QTM $M$ is said to be \emph{well-formed} if its time-evolution operator of $M$ preserves the $\ell_2$-norm in the surface configuration space of $M$ on all inputs.

At every step, we first apply $\delta$ to a superposition of surface configurations and then perform a measurement in the halting (inner) states (i.e., either accepting states or rejecting states). If $M$ is not in a halting state, we move to the next step with the quantum state obtained by tracing out all halting surface configurations. Generally, QTMs can ``recognize'' not only sets of classical strings but also sets of qustrings.

For convenience, we modify $M$ slightly so that $M$ produces $1$ (resp., $0$) in the first cell of the first work tape when $M$ enters a designated final state (in place of accepting/rejecting states). We say that $M$ \emph{produces $b$ with probability $\gamma$} if, after $M$ halts, we observe $b$ as an outcome of $M$ with probability $\gamma$.

We call $M$ \emph{polylogarithmic time} (or \emph{polylogtime}) if we force $M$ to stop its application of $\delta$ after $O(\log^k{n})$ steps for a certain fixed constant $k\in\nat^{+}$, not depending on the choice of inputs. We do not require all computation paths to terminate within the specified time.


Let us consider a language $L$ over $\{0,1\}$ that satisfies the following condition:  there are a constant $\varepsilon\in[0,1/2)$ and a polylogtime QTM $M$ whose amplitude set is $K$ such that (i) for any input $x\in L$,  $M$ accepts $x$ with probability at least $1-\varepsilon$ and (ii) for any $x\notin L$,  $M$ rejects $x$ with probability at least $1-\varepsilon$. These conditions are referred to as \emph{bounded-error probability}. The notation $\bqpolylogtime_{K}$
denotes the collection of all such languages $L$.

\subsection{Computational Complexity of Polylogtime QTMs}\label{sec:BQLOGTIME}

We begin with a discussion on the computational complexity of polylogtime QTMs.
Remember that input tapes of these machines are read-only and accessed by way of writing cell locations onto index tapes.

In the classical setting, the notation $\mathrm{DLOGTIME}$ was used in \cite{BIS90}  to express the family of all languages recognized by logtime deterministic TMs (or succinctly, DTMs). Similarly, we denote the nondeterministic variant of $\mathrm{DLOGTIME}$ by $\nlogtime$. With the use of classical probabilistic TMs (or  PTMs) in place of DTMs,
we say that a PTM $M$ recognizes a language $L$ \emph{with unbounded-error probability} if, for any $x\in L$, $M$ accepts it with probability more than $1/2$ and, for any $x\notin L$, $M$ rejects with probability at least $1/2$. We further define $\ppolylogtime$ by unbounded-error polylogtime PTMs.

\begin{theorem}\label{class-inclusion}
$\bqpolylogtime_{\bar{\rational}} \subsetneqq \ppolylogtime$ and $\nlogtime\nsubseteq \bqpolylogtime_{\complex}$.
\end{theorem}


For the proof of Theorem \ref{class-inclusion}, nevertheless, we first verify the following impossibility result of the parity function and the OR function by polylogtime QTMs, where the parity function,  $Parity$, is defined by $Parity(x) = \bigoplus_{i=1}^{n}x_i$ and the OR function, $OR$, is defined by $OR(x)=\max\{x_i\mid i\in[n]\}$ for any number $n\in\nat^{+}$ and any $n$-bit string $x=x_1x_2\ldots x_n$.

\begin{lemma}\label{polylogtime-QTM-parity-OR}
The parity function and the OR function cannot be computed by any polylogtime QTM  with bounded-error probability.
\end{lemma}

\begin{proof}
This proof comes from a result on the quantum query complexity gap between quantum and deterministic query complexities of the parity function. Assume that a polylogtime QTM, say, $M$ computes the parity function with bounded-error probability.

We encode $M$'s surface configuration $conf$ into a ``single'' quantum state $\qubit{\phi}$. As done in Lemma \ref{WH-application}, it is possible to  produce in polylog time a superposition of all locations of the input-tape cells by repeatedly applying $WH$ to $\qubit{0^{\ilog(n)}}$. This helps us access all input bits quantumly at once with the equal probability. Since the input tape is read-only, this type of input access can be realized as a \emph{(black-box) quantum query model}\footnote{This model is sometimes called a quantum network. See, e.g., \cite{BBC+01}.}
used in the study of quantum query complexity. Refer to, e.g., \cite{Amb02,BBC+01,NY04}.

In such a (black-box) quantum query model, we run the following quantum algorithm on a binary input $x$ of length $n$. We prepare a series of unitary transformations $U_0,U_1,\ldots,U_{t-1},U_t$ and a special oracle transformation\footnote{In \cite{Amb02}, for example, $Q_x$ is defined to change $\qubit{bin_k(m)}\qubit{\phi}$ to $(-1)^{x_{(m)}} \qubit{bin_k(m)}\qubit{\phi}$. This model is in essence equivalent to our current definition by a simple computation shown as follows. Let $\qubit{\xi}=\qubit{bin_k(m)}$. Starting with $\qubit{\xi}\qubit{b}$, swap between $\qubit{\xi}$ and $\qubit{b}$, apply $WH$, apply $CQ_x$ (Controlled-$Q_x$), apply $WH$, and swap back the registers, where $CQ_x(\qubit{0}\qubit{\phi}) =\qubit{0}\qubit{\phi}$ and $CQ_x(\qubit{1}\qubit{\phi}) = (-1)^{x_{(m)}}\qubit{1}\qubit{\phi}$. We then obtain $\qubit{\xi}\qubit{b\oplus x_{(m)}}$.}
$Q_x$ that changes $\qubit{bin_{k_n}(m)}\qubit{b}\qubit{\phi}$ to $\qubit{bin_{k_n}(m)}\qubit{b\oplus x_{(m)}}\qubit{\phi}$, where $k_n=\ilog(n)$ and $x_{(m)}$ is the $m$th bit of $x$.
We start with the initial quantum state $\qubit{\psi_0} = \qubit{0^m}$. We then compute
$U_tQ_xU_{t-1}Q_x\cdots U_1Q_xU_0\qubit{\psi_0}$. Finally, we measure the resulting quantum state in the computational basis.
The number $t$ indicates the total number of queries made by this algorithm on each computation path.

\begin{claim}\label{black-box-query}
Each polylogtime QTM can be simulated by a (black-box) quantum query model with $O(\log^k{n})$ queries for an appropriate constant $k\in\nat^{+}$.
\end{claim}

\begin{proof}
Recall that a QTM has an read-only input tape, which holds an input string.
Whenever a QTM makes a query on the $i$th position by entering a unique query state $q_{query}$, the machine instantly receives the information on the $i$th bit of a given input string written on the input tape.
We view this input tape as an \emph{oracle} of a (black-box) quantum query model and we further view this entire query process of the QTM as a procedure of forming a superposition of query words indicating input-bit positions and receiving their answers from the oracle.

We first construct a unitary transformation to simulate a single non-query transition of $M$. Recall that, when $M$ enters $q_{query}$, it changes $\qubit{bin_{k_n}(m)}\qubit{b}\qubit{\phi}$ to $\qubit{bin_{k_n}(m)}\qubit{b\oplus x_{(m)}}\qubit{\phi}$ in a single step.  To translate $M$'s query process, we generate  $\qubit{bin_{k_n}(1)}\qubit{0}$ in an extra register. If $M$ is in the inner state $q_{query}$, then we swap between this register and the register containing the content of $M$'s index tape. Otherwise, we do nothing.
We then apply $Q_x$ to change $\qubit{bin_{k_n}(m)}\qubit{b}\qubit{\phi}$ to $\qubit{bin_{k_n}(m)}\qubit{b\oplus x_{(m)}}\qubit{\phi}$.
Notice that this process does not alter the inner state of $M$. After applying $Q_x$, we swap back the two registers exactly when $M$'s inner state is $q_{query}$ and then we follow the transition of $M$'s inner state.

Note that the given QTM makes only $O(\log^k{n})$ queries because it runs in $O(\log^k{n})$ time. Therefore, we can transform this QTM to a query model of $O(\log^k{n})$ queries.
\end{proof}

By Claim \ref{black-box-query}, the parity function requires only $O(\log{n})$ queries in the (black box) quantum query model. However, it is shown in \cite{BBC+01} that, for the parity function of $n$ Boolean variables, $n/2$ queries are necessary in the bounded-error quantum query model (while $n$ queries are necessary in the deterministic query model). This is obviously a contradiction. The case of the OR function can be similarly handled.
\end{proof}

Let us return to Theorem \ref{class-inclusion}. Using Lemma \ref{polylogtime-QTM-parity-OR}, we can prove the theorem as described below. The core of its proof is founded on a simulation result of one-tape linear-time QTMs in \cite[Section 8]{TYL10}.

Shown in \cite[Lemma 8]{TYL10} is how to simulate a one-tape well-formed stationary QTM running in linear time on an appropriate one-tape probabilistic Turing machine (or a PTM) in linear time. In a similar vein, we can simulate polylogtime QTMs on polylogtime PTMs.

\begin{proofof}{Theorem \ref{class-inclusion}}
Let us take an arbitrary language $L$ in $\bqpolylogtime_{\bar{\rational}}$ and consider a polylogtime QTM $M$ that recognizes $L$ with bounded-error probability. We intend to show that $L$ falls in $\ppolylogtime$.
We first modify $M$ in the following way.
We prepare two extra work tapes. One of them is used as an \emph{internal clock} by moving a tape head always to the right. To avoid any unwanted interference after a computation halts prematurely, we use the other extra tape as a ``garbage tape'', to which $M$ dumps all information produced at the time of entering halting states, so that $M$ continues its operation without actually halting.
Lemma 8 of \cite{TYL10} shows the existence of a constant $d\in\nat^{+}$ and an NTM $N$ such that $d^{\mathrm{Time}_{M}(x)}\cdot p_{M}(x) = \# N(x) - \#\overline{N}(x)$ for every $x$, where $\#N(x)$ (resp., $\#\overline{N}(x)$) denotes the total number of accepting (resp., rejecting) computation paths of $N$ on input $x$. This equality holds for polylogtime machines.
The desired polylogtime PTM is obtained from $N$ by assigning an equal probability to all nondeterministic transitions.

The OR function can be computed by the polylogtime NTM that nondeterministically writes a number, say $i$ on an index tape, makes a query for the $i$th bit $x_{(i)}$ of an input $x$, and accepts exactly when $x_{(i)}$ is $1$. Thus, the OR function belongs to $\nlogtime$.
The separation between $\bqpolylogtime_{\bar{\rational}}$ and $\mathrm{NLOGTIME}$ comes from Lemma \ref{polylogtime-QTM-parity-OR}.
Since $\ppolylogtime$ includes $\nlogtime$, we obtain the desired separation between $\bqpolylogtime_{\bar{\rational}}$ and $\ppolylogtime$.
\end{proofof}

\subsection{Comparison between EQS and BQPOLYLOGTIME}\label{sec:simulation}

In what follows, we discuss a close relationship between quantum functions definable within $EQS$ and quantum functions computable by polylogtime QTMs despite numerous differences between $EQS$ and polylogtime QTMs. One such difference is that input tapes of QTMs are read-only and thus inputs are not changeable, whereas quantum functions in $EQS$ can freely modify their inputs.
In the following two theorems (Theorems \ref{QTM-simulation} and \ref{converse-simulation}), however,  we can establish the ``computational'' equivalence between  polylogtime QTMs and quantum functions in $EQS$.

To make the later simulation process simpler, we first modify a polylogtime QTM  so that it uses the binary alphabet on  an index tape as well as all work tapes by way of encoding each non-binary tape symbol into a binary one using an appropriately chosen encoding scheme.
This modification makes it possible to assume that the QTM should hold superpositions  $\qubit{\phi}$ of binary strings  on its input tape and its work tapes.
For convenience, a QTM that satisfies these conditions is called \emph{normalized}.

\begin{theorem}\label{QTM-simulation}
Any normalized polylogtime QTM $M$ with $c$ work tapes can be simulated by an appropriate quantum function $F$ in $EQS$ in the sense that, for any $b\in\{0,1\}$, $\alpha\in[0,1]$, and $\qubit{\phi}\in\Phi_{\infty}$,
$M$ takes  $\qubit{\phi}$ as an input and finally produces $b$ with probability $\alpha$ exactly when $\|\measure{b}{\psi_{F,\xi_{\phi}}}\|^2=\alpha$ holds, where $\qubit{\xi_{\phi}} = \qubit{\hat{S}} \qubit{\widetilde{B^k}}^{\otimes (c+1)} \qubit{\hat{2}} \otimes  \qubit{\phi}$ with $k=\ilog(\ell(\qubit{\phi}))$ and  $\qubit{\psi_{F,\xi_{\phi}}} = F(\qubit{\xi_{\phi}})$.
\end{theorem}

\begin{proof}
Let $M$ denote any polylogtime QTM equipped with a read-only input tape, a rewritable index tape, and multiple rewritable work tapes.
We further assume that $M$ is normalized.
For readability, we hereafter deal with the special case where $M$ has a single work tape. A general case of $c$ work tapes can be handled in a similar but naturally extended way.
Assume that $M$'s input tape holds a superposition $\qubit{\phi}$ of binary inputs with $\ell(\qubit{\phi}) \geq 4$ and that $M$ runs in time at most $\ilog(\ell(\qubit{\phi}))^{e}$ for a fixed constant $e\geq1$.
For convenience, let  $k=\ilog(\ell(\qubit{\phi}))$.
We denote by $Q$ the set of all inner states of $M$. Since $Q$ is finite, without loss of generality, $Q$ assumed to have the form $\{bin_{3t'}(i)\mid i\in[2^{3t'}]\}$ for an appropriate constant $t'\in\nat^{+}$. We set the initial inner state $q_0$ to be $bin_{3t'}(1)$. For each inner state $q\in Q$, since $q$ is expressed as a binary string, we can encode it into the string $\widetilde{q}^{(-)}$ defined in Section \ref{sec:different-object}.

We treat the content of each tape (except for the input tape) as a code block of the desired quantum function $F$.  We maintain the contents of the index tape and of the work tape as a part of two appropriate qustrings. We intend to simulate each move of $M$ on $\qubit{\phi}$ by applying an adequately defined quantum function.

Since $M$'s computation is a series of surface configurations of $M$, we thus need to ``express''  such a surface configuration using a single quantum state.
Initially, $M$'s index tape holds $B^{k}{\lhd}$ and $M$'s single work tape holds $B^{kt}{\lhd}$ for a fixed constant $t\in\nat^{+}$.
Let $w$ and $z$ respectively denote the contents of the index tape and of the work tape without the right endmarker ${\lhd}$ and let $q$ be any inner state of $M$.
Associated with $(q,w,z)$, we describe $M$'s current surface configuration as $q\# w_1H w_2\# z_1 H z_2$ with $w=w_1w_2$ and $z=z_1z_2$ by including a designated symbol $H$ and $M$'s inner state $q$, where $H$ is used to indicate  the locations of the index-tape head and of the work-tape head, which are scanning the leftmost symbols of $w_2$ and $z_2$, respectively.
Assume that these two tape heads are respectively scanning tape symbols $\eta$ and $\sigma$ on the index tape and the work tape, that is, $w_2=\eta x_2$ and $z_2=\sigma y_2$. For technical reason, we slightly modify the above description of a surface configuration and express it as
$w_1H B x_2\# z_1 H B y_2\# q\eta\sigma$ by inserting the extra symbol $B$.
To refer to this special form, we call it a \emph{modified (surface) configuration}. In particular, the suffix $q\eta\sigma$ is called a \emph{transition status}.

By Lemma \ref{code-skip-application}, it suffices for us to focus on each block of encoded tape content.
The modified configuration $w_1H B x_2\# z_1 H B y_2\# q\eta\sigma$ is encoded into the quantum state $\qubit{\widetilde{w_1}^{(-)} \hat{H} \hat{B} \widetilde{x}}
\qubit{\widetilde{z_1}^{(-)} \hat{H} \hat{B}  \widetilde{y}}\qubit{\widetilde{q}^{(-)} \hat{\eta}\hat{\sigma}}\qubit{\hat{\dashv}}$. We conveniently refer to it as the \emph{encoded (surface) configuration}. Note that $|\widetilde{q}^{(-)}\hat{\eta}\hat{\sigma}|=6t'+6=6(t'+1)$.
In fact, the modified initial configuration is of the form $H B^k\# H B^{kt} \# q_0BB$ and its encoding is of the form $\qubit{\hat{H}\widetilde{B^k}} \qubit{\hat{H} \widetilde{B^{kt}}} \qubit{\widetilde{q_0}^{(-)} \hat{B}\hat{B}}\qubit{\hat{\dashv}}$.

A \emph{run} of $M$, which covers from the initial surface configuration to certain halting surface configurations, can be simulated  using  the fast quantum recursion. To explain this simulation, for convenience, we split each move of $M$ into three separate ``phases'': (1) a tape content change, (2) an input access by a query, and (4) an output production.
We consider these three different phases of $M$ separately.  In phase (3),  in particular, we will use Scheme V to repeat phases (1) and (2) $\ilog(\ell(\qubit{\phi}))$ times.
In the end, we will combine (1)--(4) into a single quantum function.

(1) The first case to consider is that $M$ modifies multiple tapes
(except for the input tape) by a single move.
We begin with paying our attention to the modification of the index-tape symbol and describe how to simulate this tape-symbol modification.
In a single step, as our convention, a tape head firstly changes a tape symbol and secondly moves to an adjacent cell.
In other words, $M$ locally changes $w_1 H w_2\# z_1H z_2 \# q\eta\sigma$ to
its successor $w'_1Hw'_2 \# z'_1H z'_2 \# q'\eta'\sigma'$ by applying $\delta$.
This process can be expressed by a single quantum function defined as follows.

Let us consider $M$'s single transition of the form $\delta(q,\eta_2,\sigma_2) = \sum_{r}\alpha_r \qubit{p,\xi_2,\tau_2,d,d'}$, where $r$ refers to  $(q,\eta_2,\sigma_2,p,\xi_2,\tau_2,d,d')$.
This transition means that the index-tape head changes $\eta_2$ to $\xi_2$ and moves in direction $d$ and that the work-tape head changes $\sigma_2$ to $\tau_2$ and moves in direction $d'$.
To simulate this transition, it suffices to focus on four consecutive cells whose second cell is being scanned by the tape head. For simplicity, we call such a series an \emph{$H$-block}.
Let $\eta_1H B\eta_3$ and $\sigma_1H B\sigma_3$ denote two $H$-blocks and let $q\eta_2\sigma_2$ denote the current transition status.

(a) We make an application of Scheme IV in the following fashion. Let $v=(q,\eta_2,\sigma_2,p,\xi_2,\tau_2,d,d')$. We first change $\qubit{\widetilde{q}^{(-)}\hat{\eta_2}\hat{\sigma_2}}$ to $\qubit{\widetilde{p}^{(-)}\hat{B}\hat{B}}$ by remembering $(\xi_2,\tau_2,d,d')$ in the form of different quantum functions $g_u$, which are controlled by $Branch[\{g_u\}_{u}]$.
We then search for an $H$-block of the form  $\qubit{\hat{\eta_1}\hat{H}\hat{B}\hat{\eta_3}}$ and
change it to $\qubit{\hat{H}\hat{\eta_1}\hat{\xi_2}\hat{\eta_3}}$ if $d=-1$ and to $\qubit{\hat{\eta_1}\hat{\xi_2}\hat{H}\hat{\eta_3}}$ if $d=+1$.
Similarly, we change $\qubit{\hat{\sigma_1}\hat{H}\hat{B}\hat{\sigma_3}}$ according to the value of $d'$.
These changes can  be made by an appropriate quantum function, say, $F_{v,d}$.
This quantum function $F_{v,d}$ is realized as follows.

In the case of $\tau_2\in\{0,1\}$, we introduce $f_{B,\tau_2}$ that satisfies $f_{B,\tau_2}(\qubit{\hat{B}}\qubit{\hat{H}}) = \qubit{\hat{\tau_2}}\qubit{\hat{H}}$.
This quantum function $f_{B,\tau_2}$ is constructed as
$f_{B,\tau_2}\equiv Branch_3[\{g'_u\}_{u\in\{0,1\}^3}]\circ SecSWAP^{(3)}_{1,2} \circ Branch_3[\{g_u\}_{u\in\{0,1\}^3}]$, where $g_{010}(\qubit{100})=\qubit{00\tau_2}$ and $g_u=I$ for all other $u$'s,  $g'_{000}=g'_{001}$, $g'_{00}(\qubit{010})=\qubit{100}$, $g'_{000}(\qubit{100})=\qubit{010}$, and $g'_{000}(\qubit{x})=\qubit{x}$ for all other $x$'s, and $g'_u=I$ for all other $u$'s.
The remaining cases are similarly handled.
Now, let us define $\hat{G}$ to be
$SecSWAP^{(3)}_{1,3}$, which transforms  $\qubit{\alpha}\qubit{\hat{H}} \qubit{\beta}$ to $\qubit{\beta}  \qubit{\hat{H}}\qubit{\alpha}$ for any $\alpha,\beta\in\{0,1\}^3$.
Finally, when $d=-1$, we define $F_{v,d}$ to be $SecSWAP^{(3)}_{2,3} \circ SecSWAP^{(3)}_{1,2} \circ f_{B,\tau_2}  \circ \hat{G}$, which satisfies  $F_{v,d}(\qubit{\hat{\sigma}_1}\qubit{\hat{H}} \qubit{\hat{B}} \qubit{\hat{\sigma}_3}) = \qubit{\hat{H}} \qubit{\hat{\sigma}_1}\qubit{\hat{\tau}_2} \qubit{\hat{\sigma_3}}$.
In a similar way, when $d=+1$, we define
$F_{v,d} \equiv SecSWAP^{(3)}_{1,2} \circ SecSWAP^{(3)}_{2,3}\circ f_{B,\tau_2}  \circ \hat{G}$, which transforms $\qubit{\hat{\sigma}_1}\qubit{\hat{H}} \qubit{\hat{B}} \qubit{\hat{\sigma}_3}$ to $\qubit{\hat{\sigma}_1}\qubit{\hat{\tau}_2} \qubit{\hat{H}} \qubit{\hat{\sigma_3}}$.
A similar treatment works for the simulation of the work-tape head.

Notice that $M$ uses only two forms of quantum transitions. These transitions can be correctly simulated by Items 1)--3) of Scheme I.
Let $r_0=\hat{\dashv}$.
Proposition \ref{limited-recursion} makes it possible,  under a certain condition, to make a quantum function definable in a recursive fashion. We first define a quantum function $K$ by setting  $K(\qubit{r_0}) = \qubit{r_0}$ and $K(\qubit{xr_0}) = \sum_{u:|u|=9} h(\qubit{u}\otimes K(\measure{u}{xr_0}))$, where $h(\qubit{u}\qubit{wr_0}) = \sum_{v,d}\alpha_{v,d} F_{v,d}(\qubit{u}\qubit{wr_0})$ if $u=\hat{\sigma}_1\hat{H}\hat{\sigma}_2$, and $h(\qubit{u}\qubit{wr_0}) = \qubit{u}\qubit{wr_0}$ otherwise.  The proposition then guarantees the existence of a quantum function that mimics $K$ in the presence of the large-size qustring $\qubit{\phi}$. In the proof of the proposition, such a quantum function is constructed with the use of Scheme IV.
It is important to note that, its ground (quantum) functions are all query-independent. Thus, Scheme IV used here is also query-independent.

(b) Secondly, we apply $CodeREP^6$ to move the last six qubits  $\qubit{\hat{B}\hat{B}}$ obtained by (a) to the front.

(c) We then make the second application of Scheme IV.  We change $\qubit{\hat{H}\hat{\eta_1}\hat{\xi_2}\hat{\eta_3}}$ (resp., $\qubit{\hat{\eta_1}\hat{\xi_2}\hat{H}\hat{\eta_3}}$) to
$\qubit{\hat{H}\hat{B} \hat{\xi_2} \hat{\eta_3}}$ (resp., $\qubit{\hat{\eta_1}\hat{\xi_2}\hat{H} \hat{B}}$) by remembering $\eta_1$ (resp., $\eta_3$).
This change is handled in essence similarly to (a).
Moreover, a similar construction deals with the case of $\qubit{\hat{H}\hat{\sigma_1}\hat{\tau_2}\hat{\sigma_3}}$ (resp., $\qubit{\hat{\sigma_1}\hat{\tau_2}\hat{H}\hat{\sigma_3}}$).

Toward the end, we change the first six qubits $\qubit{\hat{B}\hat{B}}$ to $\qubit{\hat{\xi}\hat{\tau}}$ for symbols $\xi\in\{\xi_2,\eta_3\}$ and $\tau\in\{\tau_2,\sigma_3\}$.

(d) Finally, we apply $CodeREMOVE^6$ to move $\qubit{\hat{\xi}\hat{\tau}}$ back to the end.

(2) Next, we simulate $M$'s query access to its input qubits. Assume that the current encoded surface configuration contains  $\qubit{\widetilde{q_{query}}^{(-)} \hat{\eta_2}\hat{\sigma_2}}$.
Assume that $\qubit{\phi} = \sum_{s:|s|=2^k}\alpha_s\qubit{s}$ is written on the input tape and that the index tape contains $\qubit{bin_k(m)}\qubit{a}$, where $a$ is an auxiliary bit. When entering the query state $q_{query}$, $M$ changes $q_{query}$ to another inner state, say, $p$ and $\qubit{bin_k(m)}\qubit{a}$ to  $\qubit{bin_k(m)}\qubit{a\oplus s_{(m)}}$, where $s_{(m)}$ is the $m$th bit of $s$. We need to build a quantum function that simulates this entire query process.
As the first step, we change $\qubit{\widetilde{q_{query}}^{(-)}}$ to $\qubit{\widetilde{p}^{(-)}}$.
Since $\qubit{bin_k(m)}\qubit{a}$ is encoded into  $\qubit{\widetilde{bin_k(m)}}\qubit{\widetilde{a}}$, we can transform it to $\qubit{\widetilde{a}}\qubit{\widetilde{bin_k(m)}}$ and then to $\qubit{0^5}\qubit{a}\qubit{\widetilde{bin_k(m)}}$. Finally, we apply $Bit$ (defined in Corollary \ref{Bit-function}) to $\qubit{0^5}\qubit{a}\qubit{\widetilde{bin_k(m)}}$
and obtain $\sum_{s}\alpha_s\qubit{0^5}\qubit{a\otimes s_{(m)}}\qubit{\widetilde{bin_k(m)}}\qubit{s}$.
Since $Bit$ is query-dependent, Scheme IV used here is also query-dependent.

(3) We then combine the above two types of moves into one and express it by a single application of an appropriate quantum function. Note that $M$ accesses only the first $O(\log{\ell(\qubit{\phi})})$ cells of the work tape. We compose (1)--(2) by applying $Compo[\cdot,\cdot]$. We call by $F'$ the obtained quantum function.
We repeatedly apply it $\ilog(\ell(\qubit{\phi}))^{e}$ times to complete the simulation of the entire computation of $M$ until $M$ enters a halting (either accepting or rejecting) inner state. This repetition procedure is realized by the $e$ applications of $LCompo[\cdot]$ to $F'$.

(4) When $M$ finally enters a halting inner state, $M$ produces an output bit, say, $b$ on the first cell of the first work tape. By (1)--(2) described above, the encoded configuration has the form $\qubit{\hat{b}\widetilde{w_1}^{(-)}\hat{H} \widetilde{w_2}}  \qubit{\widetilde{z_1}^{(-)}\hat{H} \widetilde{z_2}} \qubit{\widetilde{q_{halt}}\hat{\eta}\hat{\sigma}} \qubit{\hat{\dashv}}$. We then change $\hat{b}$ ($=00b$) to $b00$ by applying $SWAP_{1,3}$ to prepare the ``correct'' output qubit. We combine  this quantum function with $F'$ to obtain the desired quantum function $F$.

This completes the entire simulation of $M$.
\end{proof}

The converse of Theorem \ref{QTM-simulation} is stated as Theorem \ref{converse-simulation}, which is given below. Recall that, at the start of a QTM $M$, its index tape and all work tapes hold the blank symbol $B$  (except for the right endmarker) in their tape cells. From this fact, we assume that inputs of quantum functions must be of the form $\qubit{\gamma_{\phi}} = \qubit{\widetilde{B^k}} \qubit{r_0} \otimes \qubit{\phi}$ with the designated separator $r_0 =\hat{2}$ and $k=\ilog(\ell(\qubit{\phi}))$ for any qustring $\qubit{\phi}\in \Phi_{\infty}$.

\begin{theorem}\label{converse-simulation}
For any quantum function $F$ defined by Schemes I--V, $F$ is computed by a certain polylogtime QTM $M$ in the following sense: for any $b\in\{0,1\}$, for any $\alpha\in[0,1]$,  and  $\qubit{\phi}\in\Phi_{\infty}$,   if $\ell(\qubit{\phi})$ is sufficiently large, then
$M$ on input $\qubit{\phi}$ produces $b$ with probability $\alpha$ iff
$\|\measure{b}{\psi_{F,\gamma_{\phi}}}\|=\alpha$ holds, where $\qubit{\psi_{F,\gamma_{\phi}}} = F(\qubit{\gamma_{\phi}})$.
\end{theorem}

In the description of the above theorem, we need to use the norm $\|\cdot\|$ instead of $|\cdot|$ because the superposition of $M$'s final configurations may contain not only the value $F(\qubit{\gamma_{\phi}})$ but also additional ``garbage'' information, which might possibly be a quantum state of large dimension, and we may need to ignore
it when making a measurement.

We wish to prove Theorem \ref{converse-simulation} by induction on the construction process of $F$. To make this induction work, we slightly modify the theorem into the following key lemma.

\begin{lemma}\label{lemma-converse}
For any quantum function $F$ defined by Schemes I--V, $F$ is  computed by a certain polylogtime QTM $M$ in the following sense: for any $b\in\{0,1\}$, any  $\qubit{\phi}\in\Phi_{\infty}$, and any $x$ in $NON_{r_0}(\qubit{\phi}) \cap \{0,1\}^{|r_0|k}$,
if $\ell(\qubit{\phi})$ is sufficiently large, then
$\|\measure{\psi_{F,\overline{\gamma}_{\phi,x}}} {\xi_{M,\overline{\gamma}_{\phi,x}}}\|=1$ holds, where  $\qubit{\overline{\gamma}_{\phi,x}} = \qubit{xr_0}\qubit{\phi}$,  $k=\ilog(\ell(\qubit{\phi}))$, $\qubit{\psi_{F,\overline{\gamma}_{\phi,x}}} = F(\qubit{\overline{\gamma}_{\phi,x}})$, and  $\qubit{\xi_{M,\overline{\gamma}_{\phi,x}}}$ is the superposition of final configurations of $M$ that starts with $\qubit{\phi}$ on the input tape and $\qubit{xr_0}$ on the first work tape.
\end{lemma}

Theorem \ref{converse-simulation} follows immediately from Lemma \ref{lemma-converse} by setting $xr_0$ in the lemma to be $\widetilde{B^k}$. The remaining task is to verify the lemma.

\begin{proofof}{Lemma \ref{lemma-converse}}
Let $F$ denote any quantum function in $EQS$. For any qustring $\qubit{\phi}\in \Phi_{\infty}$ and any string $x\in NON_{r_0}(\qubit{\phi}) \cap \{0,1\}^{|r_0|k}$ with $k=\ilog(\ell(\qubit{\phi}))$, let $\qubit{\overline{\gamma}_{\phi,x}} = \qubit{xr_0}\qubit{\phi}$ denote a qustring given to $F$ as an input.
Assuming that $\ell(\qubit{\phi})$ is sufficiently large, we first focus on $F$ and simulate the outcome of $F$ by an appropriate QTM that takes an input of the form $\qubit{\overline{\gamma}_{\phi,x}}$.
The desired polylogtime QTM $M$ reads $\qubit{\phi}$ on its input tape and $\qubit{xr_0}$ on its first work tape.

As seen later in (3) of this proof, Scheme IV may allow $F$ to access at most a constant number of locations of the input $\qubit{\phi}$, whereas $M$ does not. To circumvent this difficulty in simulating $F$ on the QTM, whenever $F$ modifies any qubit of $\qubit{\phi}$, $M$ remembers this qubit modification using its work tape as a reference to the future access to it since $M$ cannot alter any qubit of $\qubit{\phi}$.

We intend to prove the lemma by induction on the descriptive complexity of $F$.
We assume that the work-tape head is scanning the cell that contains the first qubit on the first work tape before each series of applications of the schemes of $EQS$.

(1) We first assert that all quantum functions $F$ defined by Items 1)--6) of Scheme I are computable by appropriate polylogtime QTMs, say, $M$  because the target qubits of these items lie in $xr_0$, which are written on the work tape, not on the input tape.
To verify this assertion, let us consider $PHASE_{\theta}$ of Item 2). Starting with  $\qubit{\phi}$ as well as $\qubit{xr_0}$, if the first bit of $xr_0$ is $1$, then we use the QTM's quantum transition function $\delta$ to make a phase shift of $e^{\imath \theta}$. Otherwise, we do nothing.
A similar treatment works for Items 3)--4).
For $SWAP$ of Item 5), we simply swap between the content of the cell currently scanned by $M$'s work-tape head and the content of its right adjacent cell.
For Item 6), it suffices to ``observe'' the first qubit on the first work tape in the computational basis $\qubit{a}$.


(2) We next show by way of induction on the construction process of the target quantum function $F$ by Scheme II.
Let us consider the quantum function  $F$ of the form $Compo[g,h]$ for two ground (quantum) functions $g$ and $h$. By induction hypothesis, there are two polylogtime QTMs $M_g$ and $M_h$ that  respectively compute $g$ and $h$ in the lemma's sense.
The desired QTM $M$ first checks whether $\ell(\qubit{\phi})\leq 1$.
This part is called the \emph{first phase} and it can be done by searching for the location of the right endmarker on the index tape.
We then run $M_h$ on $\qubit{\phi}$ as well as $\qubit{xr_0}$. After $M_h$ halts, we wish to run $M_g$ in the \emph{second phase}.

Now, there are two issues to deal with. Unlike the classical case of ``composing'' two TMs, we need to distinguish work tapes of $M_g$ and those of $M_h$ since we may not be able to erase the contents of $M_g$'s work tapes freely at the start of the simulation of $M_h$ in the second phase. Since we want to use $M_h$'s index tape as the index tape of $M$, we need to rename  $M_g$'s index tape to one of the work tapes of $M$.
Since the original input of $M_g$ is the qustring $h(\qubit{\phi})$, we also need to mimic $M_g$'s access to $h(\qubit{\phi})$ using only $\qubit{\phi}$. For this purpose, we need to remember the ``history'' of how we have modified  qubits of $\qubit{\phi}$ so far and, whenever $M_g$ accesses its input, we first consult this history log to check whether or not
the accessed qubit has already been modified.

(3) To simulate Scheme III, let $F\equiv Branch[g,h]$ for two ground   functions $g$ and $h$. By induction hypothesis, we take two polylogtime QTMs $M_g$ and $M_h$ respectively for $g$ and $h$ working with $\qubit{\phi}$ written on their input tapes and $\qubit{zr_0}$ written on their first work tapes.
Let us design the desired QTM $M$ to simulate $F$ on  $\qubit{\phi}$ and $\qubit{xr_0}$ as follows.
We first check if $\ell(\qubit{\phi})\leq1$. If so, we do nothing.
Hereafter, we assume otherwise. Since $Branch[g,h](\qubit{xr_0}\qubit{\phi}) = \ket{0}\otimes g(\measure{0}{xr_0}\otimes \qubit{\phi}) + \ket{1}\otimes h(\measure{1}{xr_0}\otimes \qubit{\phi})$, we scan the first qubit of $\qubit{xr_0}$ by a tape head and determine which machine (either $M_g$ or $M_h$) to run with the rest of the input. Since $M_g$ and $M_h$ correctly simulate $g$ and $h$, respectively, this new machine $M$ correctly simulates $F$.

(4) Concerning Scheme IV, let $F\equiv CFQRec_t[r_0,d,g,h|\PP_{|r_0|},\FF_{|r_0|}]$ and assume that
$M$'s input tape holds $\qubit{\phi}$ and its first work tape holds $\qubit{xr_0}$.
We first calculate the length $\ell(\qubit{\phi})$ by checking the size of the available area of the index tape in logarithmic time.  If either $\ell(\qubit{\phi})<t$ or $x=\lambda$, then we run $M_g$ on $\qubit{\phi}$ as well as $\qubit{xr_0}$ until it eventually halts. Now, we assume that $\ell(\qubit{\phi})\geq t$ and $x\neq \lambda$.

If $h$ is defined using none of $CodeSKIP_{+}$ and $CodeSKIP_{-}$, then the induction hypothesis guarantees the existence of a polylogtime QTM $M_h$ for $h$. In the case where $h$ is constructed using $CodeSKIP_{\tau}$ for a certain sign $\tau\in\{+,-\}$, we first build a QTM that simulates $CodeSKIP_{\tau}[r_0,g',h']$ for certain ground functions $g'$ and $h'$ without requiring ``logarithmic'' runtime.
It is important to note that such a QTM reads target qubits written on the work tape, not on the input tape.
The QTM $M$ searches for the first appearance of $r_0$ and then runs the corresponding QTMs $M_{g'}$ and $M_{h'}$ \emph{in parallel}.
Since the input length is $\ell(\qubit{\phi})$, the QTM halts within $O(\log{\ell(\qubit{\phi})})$ steps.

Similarly, we can handle $CodeREMOVE$ and $CodeREP$.

Since $x\neq\lambda$, $F(\qubit{xr_0}\qubit{\phi})$ is calculated as $\sum_{u:|u|=|r_0|} \sum_{v:|v|=\ell(\measure{u}{xr_0})} (h(\qubit{u}\qubit{v})\otimes p^{-1}_u(\measure{v}{\zeta_{u,p_u,\phi}^{(x'r_0)}})$, where $\qubit{\zeta_{u,p_u,\phi}^{(x'r_0)}} = \sum_{s:|s|=H(\qubit{\phi})}  (f_u(\measure{u}{x'r_0}\qubit{s}) \otimes \measure{s}{\psi_{p_u,\phi}})$ and $d(\qubit{xr_0})=\qubit{x'r_0}$.
Starting with $\qubit{\phi}$ as well as $\qubit{xr_0}$, we first move a work-tape head, passing through at most $\ilog(\ell(\qubit{\phi}))$ blocks of size $|r_0|$. Recursively, we move back the tape head to the start cell (i.e., cell $0$) and run $M_h$.
For this purpose, we write $0$ and $1$ on an index tape whenever we choose $p_u=I$ and $p_u=HalfSWAP$, respectively, because we need to trace the location of the start of each recursively halved input until we reach  $x=\lambda$ or $|x|>|r_0|k$. The entire algorithm thus requires $O(\log{\ell(\qubit{\phi})})$ steps.

By the definition of Scheme IV, the quantum function $F$ can
make a direct access to $\qubit{\phi}$ when $g$ is finally called to compute the value of $F$.
Notice that $g$ accesses at most $t$ locations of $\qubit{\phi}$. Therefore, during the simulation of $F$, $M$ makes the same number of queries to its input $\qubit{\phi}$.

(5) Finally, let us consider Scheme V. Assume that $F$ is defined to be $LCompo[g]$ for a ground function $g$. By induction hypothesis, we take a polylogtime QTM $M_g$ that simulates $g$. Assume further that, for a fixed constant $t\in\nat^{+}$, $M_g$ runs in $O(\log^t{\ell(\qubit{\phi})})$ time for any input $\qubit{\phi}$. Let $k=\ilog(\ell(\qubit{\phi}))$.
To simulate $F$, $M$ repeats a run of $M_g$ $k$ times since $F(\qubit{xr_0}\qubit{\phi}) = g^k(\qubit{xr_0}\qubit{\phi})$. The total runtime of $M$ is at most $k\cdot O(\log^t{\ell(\qubit{\phi})})$, which equals $O(\log^{t+1}{\ell(\qubit{\phi})})$.
\end{proofof}

We have shown in Section \ref{sec:BQLOGTIME} that the parity function, $Parity$, cannot be computed by polylogtime QTMs. It is possible to  generalize $Parity$ and treat it as a quantum function defined on $\HH_{\infty}$.
Lemma \ref{polylogtime-QTM-parity-OR} together with Lemma \ref{lemma-converse} then leads to the following conclusion.

\begin{proposition}\label{parity-notin-EQS}
The parity function is not definable within $EQS$.
\end{proposition}

\section{The Divide-and-Conquer Scheme and EQS}\label{sec:divide-and-conquer}

We have formulated the system $EQS$ in Section \ref{sec:definitions} by the use of recursion schematic definition and
discussed in Section \ref{sec:polylogtime-QTM} the $EQS$-characterization of quantum polylogtime computing.
In what follows, we intend to strengthen the system $EQS$ by appending an extra scheme. As such a scheme, we particularly consider the \emph{divide-and-conquer strategy}, which is one of the most useful algorithmic strategies in
solving many practical problems. We further show that the divide-and-conquer strategy cannot be ``realized'' within $EQS$.
This implies that the addition of this strategy as a new scheme truly strengthens the expressing power of $EQS$.

\subsection{Multi-Qubit Divide-and-Conquer Scheme}

A basic idea of the \emph{divide-and-conquer strategy} is to continue splitting an input of a given combinatorial problem into two (or more) smaller parts until each part is small enough to handle separately and efficiently and then to combine all the small parts in order to solve the problem on the given input.

To define the scheme that expresses this divide-and-conquer strategy, we first introduce a useful scheme called the \emph{half division scheme}.
Given two quantum functions $g$ and $h$ and any input quantum state $\qubit{\phi}$ in $\HH_{\infty}$, we simultaneously apply $g$ to the left half of $\qubit{\phi}$ and $h$ to the right half of $\qubit{\phi}$ and then obtain the new quantum function denoted by $HalfD[g,h]$.

{\it
\begin{enumerate}\vs{-2}
  \setlength{\topsep}{-2mm}%
  \setlength{\itemsep}{1mm}%
  \setlength{\parskip}{0cm}%

\item[*)] The \emph{half division scheme}.
From $g$ and $h$, we define $HalfD[g,h]$
as follows:
\vs{1}

\n\hs{10}{\rm (i)}  $HalfD[g,h](\qubit{\phi}) =
\qubit{\phi}$ \hs{48}if $\ell(\qubit{\phi})\leq 1$, \vs{1} \\
\n\hs{10}{\rm (ii)} $HalfD[g,h](\qubit{\phi}) =
\sum_{s:|s|=LH(\qubit{\phi})} ( g(\qubit{s}) \otimes h(\measure{s}{\phi}) )$ \hs{6}otherwise.

\end{enumerate}
}

Here is a quick example of how this scheme works.

\begin{example}
Let us consider $F_1\equiv HalfD[g,h]$ with $g=NOT$ and $h=WH$. Given an input $\qubit{0^n}\qubit{0^n}$, we obtain $F(\qubit{0^n}\qubit{0^n}) = g(\qubit{0^n}) \otimes h(\qubit{0^n}) = \frac{1}{\sqrt{2}} \qubit{10^{n-1}}\otimes (\qubit{0^n}+\qubit{10^{n-1}})$.
Similarly, consider $F_2\equiv HalfD[I,h]$. We then obtain $F_2(\qubit{0^n}\qubit{0^n}) = \frac{1}{\sqrt{2}}\qubit{0^n}\otimes (\qubit{0^n}+\qubit{10^{n-1}})$.
For the quantum function $F'\equiv HalfD[F_1,F_2]$, if $\qubit{0^{2n}}\qubit{0^{2n}}$ is an input to $F'$, then we obtain $F'(\qubit{0^{2n}}\qubit{0^{2n}}) = \frac{1}{2} (\qubit{10^{2n-1}} + \qubit{10^{n-1}10^{n-1}}) \otimes (\qubit{0^{2n}}+\qubit{0^n10^{n-1}})$.
\end{example}

We intend to formulate the \emph{multi-qubit divide-and-conquer scheme}  (Scheme DC) using $HalfD[g,h]$. Recall the quantum function $SWAP_{i,j}$ given in Lemma \ref{lemma:special}(10). We expand it by allowing its parameter $j$ to take a non-constant value. In particular, we intend to take  $LH(\ell(\qubit{\phi}))+1$ for $j$ and then define $midSWAP_1(\qubit{\phi})$ to be $SWAP_{2,LH(\ell(\qubit{\phi}))+1}(\qubit{\phi})$ for any $\qubit{\phi}\in\HH_{\infty}$.
More generally, for a constant $k\in\nat^{+}$, we set  $midSWAP_k(\qubit{\phi})$ to be $SWAP_{2k,m+k}\circ SWAP_{2k-1,m+k-1}\circ \cdots \circ SWAP_{k+2,m+2}\circ SWAP_{k+1,m+1}$, where $m=LH(\ell(\qubit{\phi}))$. With the use of $midSWAP_k$, for any given quantum function $h$, we introduce another scheme $MidApp_k[h]$ by setting $MidApp_k[h]\equiv midSWAP_k^{-1}\circ h \circ midSWAP_k$.


Let us quickly examine the behavior of $MidApp_k[\cdot]$ with a concrete example.

\begin{example}\label{example-midapp}
For a later argument, we consider the case of quantum function $h_0 \equiv SWAP^{-1}\circ CNOT\circ SWAP$, which  obviously belongs to $EQS_0$.
Let $\qubit{\phi}$ denote $\qubit{x_1x_2\cdots x_n}$ for $n\in\nat^{+}$ and $x_1,x_2,\ldots,x_n\in\{0,1\}$. The quantum function $MidApp_1[h_0]$  satisfies  that $MidApp_1[h_0](\qubit{\phi}) = \qubit{x_1}$ if $n=1$, and $MidApp_1[h_0](\qubit{\phi}) = \qubit{x_1\oplus x_{LH(n)+1}} \qubit{x_2x_3\cdots x_n}$ if $n\geq 2$. If we are allowed to use extra qubits, then $MidApp_1[h_0]$ can be realized by an appropriate quantum function, say, $K$ in $EQS$ obtained with $Bit$ given in Corollary \ref{Bit-function} in the following sense: $K(\qubit{0^{3k+8}}\qubit{x_1x_2\cdots x_n}) = \qubit{0^{3k+8}}\otimes MidApp_1[h_0](\qubit{x_1x_2\cdots x_n})$ if $k\in\nat^{+}$ satisfies $n=2^k$.
\end{example}


Now, we formally introduce Scheme DC, which ``expresses'' the divide-and-conquer strategy.

\begin{definition}
We express as $DC$ the following scheme.

\begin{enumerate}\vs{-3}
  \setlength{\topsep}{-2mm}%
  \setlength{\itemsep}{1mm}%
  \setlength{\parskip}{0cm}%

\item[(DC)] The \emph{multi-qubit divide-and-conquer scheme}.
From $g$, $h$, and $p$, and $k\in\nat^{+}$, (where $p$ is not defined using $MEAS[\cdot]$ and $g$, $h$, and $p$ are not defined using Scheme DC), we define
$F\equiv DivConq_k[g,h,p|f_1,f_2]$ as:
\vs{1}

\n\hs{10}(i) $F(\qubit{\phi}) = g(\qubit{\phi})$
\hs{50} if $\ell(\qubit{\phi})\leq k$, \\
\n\hs{9}(ii) $F(\qubit{\phi}) =
MidApp_k[h]( HalfD[ f_1,f_2] ( p(\qubit{\phi})  ))$ \hs{6}otherwise, \\
where $f_1,f_2 \in \{F,I\}$.
\end{enumerate}
The notation $EQS+DC$ denotes the smallest set including the quantum functions of Scheme I and being closed under Schemes II--V and DC.
\end{definition}


Hereafter, we discuss the usefulness of Scheme DC. Recall the parity function $Parity$ from Section \ref{sec:BQLOGTIME}. We have shown in Corollary \ref{parity-notin-EQS} that $EQS$ is not powerful enough to include $Parity$. In sharp contrast, we argue that $Parity$ is in fact definable by applying Schemes I--V and DC.

\begin{proposition}\label{Parity-computable}
There exists a quantum function $f$ in $EQS+DC$ that simulates $Parity$ in the following sense: for any $x\in\{0,1\}^*$ and any $b\in\{0,1\}$, $Parity(x)=b$ iff $\|\measure{b}{\psi_{f,x}}\|^2=1$, where $\qubit{\psi_{f,x}}= f(\qubit{x})$.
\end{proposition}

\begin{proof}
Let us recall the quantum function $h_0$ described in Example \ref{example-midapp} and define $F$ to be  $DivConq_1[g,h,p|f_1,f_2]$ with $g=p=I$ and $f_1=f_2=F$.
For any $n$-bit string $x=x_1x_2\cdots x_n$, we want to show by induction on $n\in\nat^{+}$ that (*) there exist $y_2,y_3,\ldots,y_n\in\{0,1\}$ for which  $F(\qubit{x}) = \qubit{\bigoplus_{i=1}^{n}x_i} \qubit{y_2}\cdots \qubit{y_n}$.

If $n=1$, then we instantly obtain $F(\qubit{\phi})=\qubit{\phi}$. Assume that $n=2$. Since $HalfD[f_1,f_2](\qubit{x_1x_2}) = f_1(\qubit{x_1})\otimes f_2(\qubit{x_2})$, it follows that $F(\qubit{x_1x_2}) = MidApp_1[h_0](F(\qubit{x_1})\otimes F(\qubit{x_2})) = MidApp_1[h_0](\qubit{x_1x_2}) = \qubit{x_1\oplus x_2}\qubit{x_2}$.
Let $n\geq3$ and assume by induction hypothesis that (*) is true for all indices $k\in[n-1]$; namely, $F(\qubit{x_1x_2\cdots x_k}) = \qubit{\bigoplus_{i=1}^{k}x_i} \qubit{y_2}\cdots \qubit{y_k}$ for certain suitable bits $y_2,y_3,\ldots,y_k\in\{0,1\}$. Let us concentrate on the case of $x\in\{0,1\}^{n}$ with $x=x_1x_2\cdots x_n$.
For simplicity, we write $m$ in place of $LH(n)$. Let $x'=x_1x_2\cdots x_{m}$ and $x''=x_{m+1}\cdots x_n$ so that $x=x'x''$. It thus follows that $HalfD[f_1,f_2](\qubit{x}) = f_1(\qubit{x'})\otimes f_2(\qubit{x''})$. Since $F(\qubit{x'}) = \qubit{\bigoplus_{i=1}^{m}x_i}\qubit{y_2\cdots y_{m}}$ and $F(\qubit{x''}) = \qubit{\bigoplus_{i=m+1}^{n}x_i} \qubit{y_{m+2}\cdots y_n}$ by induction hypothesis, we conclude that  $MidApp_1[h_0](F(\qubit{x'})\otimes F(\qubit{x''}))  = MidApp_1[h_0]( \qubit{\bigoplus_{i=1}^{m}x_i}\qubit{y_2\cdots y_{m}} \otimes \qubit{\bigoplus_{i=m+1}^{n}x_i} \qubit{y_{m+2}\cdots y_n} ) = \qubit{\bigoplus_{i=1}^{n}x_i}\qubit{y_2\cdots y_n}$.
This implies that (*) holds for all $n\in\nat^{+}$.
\end{proof}

\subsection{Approximately Admitting}

By Proposition \ref{Parity-computable}, we may anticipate that the multi-qubit divide-and-conquer scheme (Scheme DC) cannot be ``realized'' or even ``approximated'' within the system $EQS$. We formalize this latter notion under the new terminology of ``approximately admitting''.

Let us consider an arbitrary scheme (such as composition and fast quantum recursion) whose construction requires a series of quantum functions. Recall the notion of ground (quantum) functions from Section \ref{sec:elementary-scheme}.
Let $\SSS$ denote a class of quantum functions and assume that $\RR$ is a scheme requiring a series of $k$ ground functions taken from $\SSS$.
We say that $\SSS$ \emph{approximately admits} $\RR$ if, for any  series  $\GG=(g_1,g_2,\ldots,g_k)$ with $g_1,g_2,\ldots,g_k\in \SSS$, there exists a quantum function $f\in \SSS$ and a constant $\varepsilon\in[0,1/2)$ such that, for any $\qubit{\phi}\in\HH_{\infty}$,  $\|\measure{\psi_{f,\phi}}{\xi_{\RR,\GG,\phi}}\|^2\geq 1-\varepsilon$ holds, where $\qubit{\psi_{f,\phi}} = f(\qubit{\phi})$ and $\qubit{\xi_{\RR,\GG,\phi}} = \RR(g_1,g_2,\ldots,g_k)(\qubit{\phi})$.
With this new terminology, we can claim that all schemes listed in Lemma \ref{various-schemes}, for example, are indeed approximately admitted by $EQS$.

\begin{theorem}\label{no-admit-DC}
$EQS$ does not approximately admit the multi-qubit divide-and-conquer scheme.
\end{theorem}

\begin{proof}
Assume that $EQS$ approximately admits Scheme DC.  Since the parity function is realized in $EQS+DC$ (Proposition \ref{Parity-computable}), there exists a polylogtime QTM that computes $Parity$ due to the characterization theorem  (Theorem \ref{QTM-simulation}) of $EQS$ in terms of polylogtime QTMs. This clearly contradicts the fact that  no polylogtime QTM can compute the parity function  (Lemma \ref{polylogtime-QTM-parity-OR}).
\end{proof}

\section{Further Discussion and Future Directions}

The schematic approach toward quantum computability was initiated in \cite{Yam20} and made a great success to precisely capture quantum polynomial-time computability using the exquisite scheme of \emph{multi-qubit quantum recursion}. The use of such recursion schemes to characterize quantum computability further leads us to a study on the expressibility of the schemes for quantum computations rather than the more popular algorithmic complexity of quantum computations.
In this work, we have made an additional step toward an introduction of a more elementary form of the recursion schematic definition than the one in \cite{Yam20}.
In particular, we have investigated the scheme of \emph{(code-controlled) fast quantum recursion} in Section \ref{sec:different-object} as a basis to the class $EQS$ of ``elementary'' quantum functions and we have demonstrated the usefulness of various quantum functions based on this new scheme in connection to ``parallel'' computability in Section \ref{sec:relationships}.
An additional scheme, \emph{multi-qubit divide-and-conquer}, has been examined and shown not to be approximately admitted within the framework of $EQS$ in Section \ref{sec:divide-and-conquer}.

To promote the future research on the schematic definability of quantum computations, we wish to list seven natural open questions that have been left unanswered throughout this work.
We expect that fruitful research toward the answers to the questions would make significant progress in the near future on the descriptional aspects of quantum computing.

\begin{enumerate}\vs{-2}
  \setlength{\topsep}{-2mm}%
  \setlength{\itemsep}{1mm}%
  \setlength{\parskip}{0cm}%

\item It is quite important to discuss what recursion schemes must be chosen as a basis of $EQS$. In our formulation, Scheme IV in particular looks rather complicated compared to  other schemes. Therefore, we still need  to find more ``natural'' and ``simpler'' schemes needed to define $EQS$ precisely.

\item Scheme V looks quite different from the other schemes. For instance, the recursive application of $f_u$ to compute $F$ in Scheme IV is controlled by ``internal'' conditions, whereas the repeated application of $g$ to compute $F$ in Scheme V is controlled by an ``external''  condition. It is thus desirable to remove Scheme V from the definition of $EQS$ by simply modifying the definitions of Schemes I--IV.  How can we modify them to achieve this goal?

\item We have shown in Theorem \ref{no-admit-DC} that $EQS+DC$ has more expressing power than $EQS$ alone. However, we do not know the exact computational complexity of $EQS+DC$. It is of great importance to determine its exact complexity.

\item Concerning recursion-schematic characterizations of quantum computing, we have focussed our attention only on ``runtime-restricted'' quantum computations. Any discussion on ``space-restricted'' quantum computing has eluded from our attention so far. How can we characterize such computations in terms of recursion schemes?

\item We have discussed the relative complexity of $\bqpolylogtime$ in Section \ref{sec:BQLOGTIME} in comparison to $\nlogtime$ and $\ppolylogtime$. On the contrary, we still do not know whether $\bqpolylogtime_{\bar{\rational}}$ differs from $\mathrm{BPPOLYLOGTIME}$, which is the bounded-error analogue of $\ppolylogtime$. Are they truly different?

\item It is well-known that the choice of (quantum) amplitudes of QTMs affects their computational complexity. In the polynomial-time setting, for example, the bounded-error quantum polynomial-time class $\bqp_{\complex}$ differs from $\bqp_{\tilde{\complex}}$, where $\tilde{\complex}$ is the set of polynomial-time approximable complex numbers. On the contrary, the nondeterministic variant $\nqp_{\complex}$ collapses to  $\nqp_{\tilde{\complex}}$ \cite{YY99}. In the polylogtime setting, is it true that $\bqpolylogtime_{\complex} \neq \bqpolylogtime_{\tilde{\complex}}$?

\item It is desirable to develop a general theory of descriptional complexity based on recursion schematic definitions of quantum functions for a better understanding of quantum computability.
\end{enumerate}


\let\oldbibliography\thebibliography
\renewcommand{\thebibliography}[1]{%
  \oldbibliography{#1}%
  \setlength{\itemsep}{-2pt}%
}

\end{document}